\documentclass[12pt]{article}
\usepackage[latin9]{inputenc}
\usepackage{geometry}
\usepackage{float}
\usepackage{rotfloat}
\usepackage{mathtools}
\usepackage{amsmath}
\usepackage{amsthm}
\usepackage{amssymb}
\usepackage{graphicx}
\usepackage{setspace}
\usepackage{esint}
\usepackage[authoryear]{natbib}
\usepackage{nicefrac}
\usepackage{verbatim}
\makeatletter

\providecommand{\tabularnewline}{\\}

  \theoremstyle{plain}

\usepackage[toc,page]{appendix}

\usepackage{latexsym}
\usepackage{amsthm}\usepackage{amsfonts}\usepackage{graphicx}\usepackage{epsfig}
\usepackage{bm}
\newcommand*{\patchAmsMathEnvironmentForLineno}[1]{%
      \expandafter\let\csname old#1\expandafter\endcsname\csname #1\endcsname
      \expandafter\let\csname oldend#1\expandafter\endcsname\csname end#1\endcsname
      \renewenvironment{#1}%
         {\linenomath\csname old#1\endcsname}%
         {\csname oldend#1\endcsname\endlinenomath}}%
    \newcommand*{\patchBothAmsMathEnvironmentsForLineno}[1]{%
      \patchAmsMathEnvironmentForLineno{#1}%
      \patchAmsMathEnvironmentForLineno{#1*}}%
    \AtBeginDocument{%
    \patchBothAmsMathEnvironmentsForLineno{equation}%
    \patchBothAmsMathEnvironmentsForLineno{align}%
    \patchBothAmsMathEnvironmentsForLineno{flalign}%
    \patchBothAmsMathEnvironmentsForLineno{alignat}%
    \patchBothAmsMathEnvironmentsForLineno{gather}%
    \patchBothAmsMathEnvironmentsForLineno{multline}%
    }
\usepackage[mathlines,displaymath]{lineno}

\def\dispmuskip{\thinmuskip= 3mu plus 0mu minus 2mu \medmuskip=  4mu plus 2mu minus 2mu \thickmuskip=5mu plus 5mu minus 2mu}
\def\textmuskip{\thinmuskip= 0mu                    \medmuskip=  1mu plus 1mu minus 1mu \thickmuskip=2mu plus 3mu minus 1mu}
\def\beq{\dispmuskip\begin{equation}}    \def\eeq{\end{equation}\textmuskip}
\def\beqn{\dispmuskip\begin{displaymath}}\def\eeqn{\end{displaymath}\textmuskip}
\def\bea{\dispmuskip\begin{eqnarray}}    \def\eea{\end{eqnarray}\textmuskip}
\def\bean{\dispmuskip\begin{eqnarray*}}  \def\eean{\end{eqnarray*}\textmuskip}

\newcommand{\eps}{\epsilon}

\newcommand{\wt}{\widetilde}

\def\B{{\cal B}}

\def\d{\delta}

\def\M{{\cal M}}
\def\N{{\mathbb N}}
\def\R{{\mathbb R}}
\def\X{{\cal X}}
\def\Y{{\cal Y}}

\def\({\big ( }
\def\){\big ) }
\def\transp{{\tiny {\rm T}}}
\def\PHS{PHS}
\def\CPMMH{CPMMH}
\def\corrPMMHPG{CPHS}
\def\CPHS{CPHS}
\def\PGBS{PGBS}
\def\PMMH{PMMH}
\def\PMCMC{PMCMC}
\def\ov{\overline}
\def\d{{\rm d}}

\usepackage[mathlines,displaymath]{lineno}
\usepackage{authblk}
\makeatother

\newtheorem{assumption}{Assumption}
\newtheorem{lemma}{Lemma}
\newtheorem{corollary}{Corollary}
\newtheorem{proposition}{Theorem}

\newtheorem{definition}{Definition}
\newcommand{\bs}{\boldsymbol}
\floatstyle{ruled}
\newfloat{algorithm}{tbp}{loa}
\providecommand{\algorithmname}{Algorithm}
\floatname{algorithm}{\protect\algorithmname}
\renewcommand{\eqref}[1]{Eq.~(\ref{#1})}

\newcommand{\myquad}{\hspace{5pt}}

\def\({\Big ( }
\def\){\Big )}

\begin{document}
\include{def}
\title{The Correlated Particle Hybrid Sampler for State Space Models.}


\author[1,3]{David Gunawan}
\author[2,3]{Chris Carter}
\author[2,3]{Robert Kohn}
\affil[1]{University of Wollongong}
\affil[2]{University of New South Wales}
\affil[3]{ACEMS}
\date{}
\maketitle
\begin{abstract}
Particle Markov Chain Monte Carlo (PMCMC)  is a general computational approach to Bayesian inference for general state space models. Our article scales up PMCMC in terms of the number of observations and parameters by generating the parameters that are highly correlated  with the states with the states \lq integrated out\rq{} in a pseudo marginal step; the rest of the parameters are generated conditional on the states. The novel contribution of our article 
is to make the pseudo-marginal step  much more efficient by positively correlating the numerator and denominator in the Metropolis-Hastings acceptance probability. This is done in a novel way by expressing the target density of the PMCMC
in terms of the basic uniform or  normal random numbers used in the sequential Monte Carlo algorithm, instead of the standard way in terms of state particles. 
We also show that the new sampler combines and generalizes two separate particle MCMC approaches: particle Gibbs and the correlated pseudo marginal Metropolis-Hastings. 
We investigate the performance of the hybrid sampler  empirically by applying it to univariate and multivariate stochastic volatility models
having both a large number of parameters and a large number of latent states and show that it is much more efficient than competing PMCMC methods. 
\end{abstract}

\section{Introduction\label{sec:Introduction}}
Our article proposes a novel particle Markov chain Monte Carlo (PMCMC) sampler for the general class of state space models defined in section~\ref{sec: Preliminaries}. The sampler generates parameters that are highly correlated with the states, or are difficult to generate conditional on the states, using a pseudo marginal step with the states \lq integrated out\rq{}, while the rest of the parameters are generated by a particle Gibbs step.  \citet{Mendes2020} shows that such an approach allows the sampler to scale up in terms of the number of parameters. Our main contribution is to make the pseudo marginal step much more efficient than in \citet{Mendes2020} as explained below. 

In a seminal paper, \citet{andrieuetal2010} 
propose two PMCMC methods
for  state space models. The first is the particle marginal Metropolis-Hastings
(PMMH), where the parameters are generated with the unobserved states
\lq integrated out\rq. The second is particle Gibbs (PG) 
which generates the parameters conditional on the states. The basic idea of PMCMC
methods is to define a target distribution on an augmented space that
includes all of the parameters and the particles generated by a sequential
Monte Carlo (SMC) algorithm and has the joint posterior density of the parameters
and states as a marginal density.

\citet{Mendes2020} propose the particle hybrid sampler (\PHS) that combines the two approaches in \citet{andrieuetal2010}.
Parameters that are efficiently generated by conditioning on the states are generated using a PG step(s); all the other parameters are generated by particle marginal Metropolis-Hastings steps.
For example, the  particle marginal Metropolis-Hastings is suitable for parameters that are highly correlated with the states, or parameters that very expensive to generate efficiently conditional on the states.

\citet{Deligiannidis2018}  improve the efficiency of the PMMH sampler by proposing
the correlated pseudo-marginal (\CPMMH) method, which correlates the
random numbers used in constructing  the estimated likelihood
at the current and proposed values of the parameters.
They show that by inducing a high correlation in successive iterations between those
random numbers, it is necessary to increase the number of particles
$N$ in proportion to $T^{\nicefrac{k}{k+1}}$, where $T$ is the
number of observations and $k$ is the state dimension. The computational
complexity of the correlated pseudo-marginal method is $O\left(T^{\nicefrac{2k+1}{k+1}}\right)$,
up to a logarithmic factor, compared to $O\left(T^{2}\right)$ for the standard PMMH sampler.
This shows that the correlated pseudo-marginal method can scale up with the number of observations  more easily than  standard pseudo-marginal methods, as long as the state dimension $k$ is not too large. A number of other papers, including \citet{yildirim2018} and \citet{andrieu2018utility}, also extend the literature on PMCMC methods. 

Our article builds on the correlated pseudo-marginal (\CPMMH) method of \citet{Deligiannidis2018} and the particle hybrid sampler (\PHS{}) sampler 
of \citet{Mendes2020} and proposes a novel PMCMC sampler for state space models, which we call the correlated  particle hybrid sampler (\corrPMMHPG). It improves on both the CPMMH and \PHS{} to allow for further improvements in scaling up particle MCMC in terms of the number of 
parameters and the length of the series. The \corrPMMHPG{} expresses the target density of the PMCMC in terms of  the basic uniform or standard normal random numbers used in the sequential Monte Carlo algorithm, rather than the state particles, as done in the rest of the literature. This novel representation allows the estimated likelihoods in the numerator and denominator of the acceptance probability of the MH step of the PMMH to be correlated as in the \CPMMH{}, while still allowing the PMMH and PG steps to be combined in a hybrid form. Section~\ref{S: corr PMCMC} of our article  
also shows that the PG and \CPMMH{} samplers are special cases of the \corrPMMHPG{} sampler. While our \CPHS{} and the \PHS{} of \citet{Mendes2020} similarly combine PG and PMMH sampling, they technically involve quite different constructions of the target density. In particular, it is the use of the basic random numbers in the target density of the \CPHS{} that allows us to correlate the numerator and denominator in the PMMH step.  


We illustrate the \corrPMMHPG{} empirically using a univariate stochastic volatility model
with leverage and a multivariate factor stochastic volatility model with
leverage using real and simulated datasets. 
We use these univariate and 
multivariate models to compare the performance of our sampling scheme to several
competing sampling schemes currently available in the literature. 

\citet{Kim:1998} estimate the univariate SV model without leverage by  modelling the log of the squared returns and approximating  the $\chi^2_1$ distribution of the innovations by a seven component mixture of normals; the paper then corrects the approximation using importance sampling, which makes their approach simulation consistent.
Current popular MCMC estimation approaches for estimating the factor stochastic volatility model without leverage proposed by \citet{Chib2006}, \citet{Kastner:2017}, and \citet{Kastner2019} follow the approach proposed by \citet{Kim:1998}, but are neither exact nor flexible. The first drawback of the approaches for estimation of the factor SV models mentioned above is that they do not correct for their mixture of normals approximations  and so their estimators are not simulation consistent, with the approximation errors increasing with the dimension of the state vector.  Section~5.2 of \citet{gunawan2022flexible} reports that the MCMC sampler of \citet{Kastner:2017} can give different results to the CPHS\footnote{Note that  \citet{gunawan2022flexible} refer to the CPHS as PHS.}
The second drawback of the factor SV estimators using 
mixture of normals approximations is that they are inflexible because they require different approximations when applied to different factor SV models. For example, the factor stochastic volatility with leverage requires approximating the joint distribution of outcome and volatility innovations by a ten-component mixture of bivariate normal distributions as in \citet{Omori2007}. The CPHS is simulation consistent and it does not need to make such approximations and is simulation consistent in the sense that as the number of MCMC samples tends to infinity, the PHS iterates converge to the true posterior distributions of the states and parameters.

We also consider a factor stochastic volatility model in which the log-volatilities of the idiosyncratic errors follow a GARCH
diffusion continuous volatility time process which does not have a closed form transition density \citep{Kleppe2010,Chib2004}. It is well known that the Gibbs-type MCMC and PMCMC samplers, including the MCMC sampler of \citet{Kastner:2017} and particle Gibbs of \citet{andrieuetal2010}, perform poorly for generating parameters of a diffusion model \citep{Stramer2011}. We show in section \ref{results multivariate examples} that the particle Gibbs performs poorly in estimating the parameters of the GARCH diffusion model.

The rest of the article is organized as follows. Section~\ref{sec: Preliminaries}
outlines the basic state space model,
the sequential Monte Carlo algorithm for this model, and the  backward simulation method.
Section~\ref{S: corr PMCMC} introduces the \corrPMMHPG{}, its invariant distribution,
the constrained conditional sequential Monte Carlo algorithm, which is an important component of the \corrPMMHPG{}, and shows that the PG and the CPMMH samplers are special cases of the \corrPMMHPG{}. \ \ 
Sections \ref{S: univariate example SV model} and \ref{Multivariate example} compare the performance of the \corrPMMHPG{} to
competing PMCMC methods for estimating univariate and multivariate stochastic volatility models.
Section \ref{sec:discussion} discusses how the \corrPMMHPG{} can estimate some other factor stochastic volatility models that would be very difficult to estimate by existing methods.
The paper  has an online supplement containing some further technical and empirical results.

\section{Preliminaries\label{sec: Preliminaries}}
This section introduces the state space model,
sequential Monte Carlo (SMC) and backward simulation. We use the colon notation for collections of variables, i.e.
$a_{t}^{r:s}=\left(a_{t}^{r},...,a_{t}^{s}\right)$ and for $t\leq u$,
$a_{t:u}^{r:s}=\left(a_{t}^{r:s},...,a_{u}^{r:s}\right)$.

\subsection{State Space Models\label{sub:State-Space-Models}}
We consider the stochastic process $\{(X_t, Y_t), t \geq 1\}$ with parameter $\theta$.
The $Y_t$ are the observations and the $X_t$ form the latent state space process.
The density of $X_1 $ is
$f_1^\theta(x_1), $ the density of $X_t $ given $X_{1:t-1} = x_{1:t-1}, Y_{1:t-1} = y_{1:t-1}$ is
$f_t^\theta(x_t|x_{t-1},y_{t-1})$ ($t \geq 2$), and the density of $Y_t $ given
$X_{1:t} = x_{1:t}, Y_{1:t-1} = y_{1:t-1}$ is $g_t^\theta(y_t|x_t)$.

We assume that the parameter $\theta \in \Theta$,
 where $\Theta $ is a subset of $\R^{d_{\theta}}$, and $p(\theta)$ is the prior for
 $\theta$. The $X_t$ are $\X$ valued and the $Y_t$ are $\Y$ valued and  $g_t^\theta$ and $f_t^\theta$ are densities
with respect to dominating measures which we write as $\d x $ and $\d y$. The dominating
measures are frequently taken to be the Lebesgue measure if $\X\in\B\left(\R^{d_{x}}\right)$
and $\Y\in\B\left(\R^{d_{y}}\right)$, where $\B\left(A\right)$ is
the Borel $\sigma$-algebra generated by the set $A$. Usually $\X=\R^{d_{x}}$
and $\Y=\R^{d_{y}}$.

The joint density function of $\left(x_{1:T},y_{1:T}\right)$ is
\begin{align}\label{eq: joint density}
p\left(x_{1:T},y_{1:T}|\theta\right)=f_{1}^{\theta}\left(x_{1}\right)g_{1}^{\theta}\left(y_{1}|x_{1}\right)\prod_{t=2}^{T}f_{t}^{\theta}\left(x_{t}|x_{t-1},y_{t-1}\right)g_{t}^{\theta}\left(y_{t}|x_{t}\right).
\end{align}
The likelihood $p(y_{1:T}|\theta)$ is $\prod_{t=1}^{T}Z_{t}\left(\theta\right)$,
where $Z_{1}\left(\theta\right)=p\left(y_{1}|\theta\right)$ and $Z_{t}\left(\theta\right)=p\left(y_{t}|y_{1:t-1},\theta\right)$
for $t\geq2$. By using Bayes rule, we can express the joint posterior
density of $\theta$ and $X_{1:T}$ as
\[
p\left(x_{1:T},\theta|y_{1:T}\right)=\frac{p\left(x_{1:T},y_{1:T}|\theta\right)p\left(\theta\right)}{\overline{Z}_{T}},
\]
where the marginal likelihood of $y_{1:T}$ is $\overline{Z}_{T}=\int_{\Theta}\prod_{t=1}^T Z_{t}\left(\theta\right)p\left(\theta\right)d\theta=p\left(y_{1:T}\right)$.
Section \ref{SV with leverage} discusses an example of popular univariate state space models.
\subsection{Example: the univariate SV model with leverage  \label{SV with leverage}}
We consider the univariate stochastic volatility model
with leverage discussed and motivated by \cite{Harvey1996},
\begin{equation}\label{eq: univ SV with leverage}
\begin{aligned}
y_t& = \exp(x_t/2) \epsilon_t, \\ 
x_1 & \sim N \begin{pmatrix}  \mu, \nicefrac{\tau^2}{1-\phi^2} \end{pmatrix}, \\
x_{t+1}& = \mu + \phi (x_t - \mu) + \eta_t, \quad ( t \geq 1) \\
\begin{pmatrix}
\eps_t\\\eta_t
\end{pmatrix} & \sim N \begin{pmatrix}  \begin{pmatrix} 0 \\ 0   \end{pmatrix}, &  \begin{pmatrix} 1 &  \rho\tau  \\\rho\tau & \tau^2\end{pmatrix} \end{pmatrix};\\
\end{aligned}
\end{equation}
$x_t$ is the latent volatility process  and the sequence $(\eps_t,\eta_t)_{1:T}$  is independent. 
Hence, the observation densities and the state transition densities are: 
\begin{equation}  \label{eq: univ sv state trans eqn}
\begin{aligned}
g_t^\theta(y_t|x_t) & = N\(0, \exp(x_t)\), \myquad (t \geq 1) \\
f_1^\theta(x_1)& = N\(\mu, \nicefrac{\tau^2} {1- \phi^2}\), \\
f_t^\theta(x_{t+1}|x_t,y_t) & = N(\mu + \phi (x_t - \mu) + \rho \tau \exp(-x_t/2) y_t , \tau^2 (1-\rho^2) \) \myquad (t \geq 1).\\
\end{aligned}
\end{equation}
The persistence parameter $\phi$ is restricted to 
$|\phi| <1$ to ensure that the volatility is stationary. 
The correlation $\rho$ between $\epsilon_t$ and $\eta_t$ also needs to lie between $\pm 1$.
When $\rho = 0$ in \eqref{eq: univ SV with leverage}, we obtain the standard volatility model without leverage.



\subsection{(Correlated) Sequential Monte Carlo\label{SS: smc}}
Exact filtering and smoothing are only available for some special
cases, such as the linear Gaussian state space model. For non-linear and
non-Gaussian state space models, a sequential Monte Carlo (SMC) algorithm
can be used to approximate the joint filtering densities.

Unless stated otherwise, upper case letters indicate random variables
and lower case letters indicate the corresponding values of these
random variables. e.g., $A_{t}^{j}$ and $a_{t}^{j}$, $X_{t}$ and
$x_{t}$.

We use SMC to approximate the joint filtering densities $\left\{ p\left(x_{t}|y_{1:t},\theta\right):t=1,2,...,T\right\} $
sequentially using $N$ particles, i.e., weighted samples $\left\{x_{t}^{1:N},\overline{w}_{t}^{1:N}\right\}$,
drawn from some proposal densities
$m_{1}^{\theta}\left(x_{1}\right) =  m_{1}\left(x_{1}|Y_{1}=y_{1},\theta\right)$ and\newline
$m_{t}^{\theta}\left(x_{t}|x_{1:t-1}\right)=  m_{t}\left(x_{t}|X_{1:t-1}=x_{1:t-1},Y_{1:t}=y_{1:t},\theta\right)$ for $t \geq 2$.


\begin{definition}\label{def: define weights}
Let
\begin{align*}
w_1^i & := \frac{g_1^\theta(y_1|x_1^i)f_1^\theta(dx_1^i) }{ m_1^\theta(dx_1^i)},
w_t^i  := \frac{g_t^\theta(y_t|x_t^i)f_t^\theta(dx_{t}^i|x_{t-1}^{a_{t-1}^i},y_{t-1}) }{ m_t^\theta(dx_t^i|x_{t-1}^{a_{t-1}^i} )} \myquad
\text{for $t\geq 2$} \myquad \text{and} \myquad
\ov w_t^i  := w_t^i /\sum_{j=1}^N w_t^j .
\end{align*}
\end{definition}

\
The SMC for $t=2, \dots, T$ is implemented  by  using the resampling scheme \newline
$\M(a_{t-1}^{1:N}|\ov w_{t-1}^{1:N},x_{t-1}^{1:N})$, which depends on
$\ov w_{t-1}^{1:N}$ and $x_{t-1}^{1:N}$ for a given $t$. The argument $a_{t-1}^{1:N}$  means that $X_{t-1}^{A_{t-1}^i} = x_{t-1}^{a_{t-1}^i}$
is the ancestor of $X_t^i = x_t^i $. Section \ref{assumptionSMC} of the online supplement gives the assumptions for the proposal densities and resampling scheme.


\begin{definition}\label{def: Xt and Vx}
We define $V_{xt}^i$ as the vector random variable used to generate
$X_{t}^i$ given $\theta$ and $x_{t-1}^{a_{t-1}^i}$, where $a_{t-1}^i $ is the ancestor index of $X_t^i$, as defined above.
We write,
\begin{align}\label{eq: def cal X}
X_{1}^i=\mathfrak{X}(V_{x1}^i; \theta) \myquad \text{and} \myquad  \myquad X_{t}^i=\mathfrak{X}(V_{xt}^i; \theta, x_{t-1}^{a_{t-1}^i})
\myquad \text{for $t \geq 2$}.
\end{align}
Denote the distribution of  $V_{xt}^i$ as $\psi_{xt}\left(\cdot\right)$. Common choices for $\psi_{xt}\left(\cdot\right)$ are iid $U\left(0,1\right)$
or iid  $N\left(0,1\right)$ random variables. The function $\mathfrak{X}$ maps the random variable $V_{xt}^i$ to $X_{t}^i$.
\end{definition}

We now discuss how Definition \ref{def: Xt and Vx} is implemented for the univariate
SV model with leverage using the bootstrap filter. At time $t=1$,
we first generate $N$ random variables $v_{x1}^{i}\sim N\left(0,1\right)$,
and define 
\begin{equation}
x_{1}^{i}=\sqrt{\frac{\tau^{2}}{1-\phi^{2}}}v_{x1}^{i}+\mu,\label{V_x example SV1}
\end{equation}
for $i=1,...,N$. At time $t>1$, we generate $N$ random variables
$v_{xt}^{i}\sim N\left(0,1\right)$, and define 
\begin{equation}
x_{t}^{i}=\mu+\phi\left(x_{t-1}^{a_{t-1}^{i}}-\mu\right)+\rho\tau\exp\left(-\frac{x_{t-1}^{a_{t-1}^{i}}}{2}\right)
y_{t-1}+\sqrt{\tau^{2}\left(1-\rho^{2}\right)}v_{xt}^{i},\label{V_x example SVt}
\end{equation}
for $i=1,...,N$.

\begin{definition}\label{def: definition V}
For $t \geq 2$, we define ${V}_{A,t-1}^{1:N}$ as the
vector of random variables used to generate the ancestor indices
$A_{t-1}^{1:N}$ using the resampling scheme $\mathcal{M}\left(\cdot |\bar{w}_{t-1}^{1:N}, x_{t-1}^{1:N}\right)$
and define $\psi_{A,t-1}\left(\cdot\right)$ as the distribution of $V_{A,t-1}^{i}$.
This defines the mapping $A_{t-1}^{1:N}=\mathfrak{   A}\(V_{A,t-1}^{1:N} ; \overline{w}_{t-1}^{1:N},x_{t-1}^{1:N}\)$. The function $\mathfrak{A}$ maps the random variable $V_{A,t-1}^{1:N}$ to $A_{t-1}^{1:N}$ given $\overline{w}_{t-1}^{1:N}$ and $x_{t-1}^{1:N}$.
Common choices for $\psi_{A,t-1}\left(\cdot\right)$ are iid $U\left(0,1\right)$
or iid  $N\left(0,1\right)$ random variables.
\end{definition}


We now discuss how Definition~\ref{def: definition V} is implemented for the simple multinomial
resampling scheme.
In the usual resampling scheme, an empirical cumulative distribution function $\widehat{F}_{t-1}^{N}\left(j\right)=\sum_{i=1}^{j}\overline{w}_{t}^{i}$  is first constructed
based on the index of the particle that we want to sample.
We then generate $N$ random variables $v_{A,t-1}^{i}\sim U\left(0,1\right)$,
sample the ancestor indices $a_{t-1}^{i}=\underset{j}{\min}\widehat{F}_{t-1}^{N}\left(j\right)\geq v_{A,t-1}^{i}$,
and then obtain the selected particles $x_{t-1}^{a_{t-1}^{i}}$ for
$i=1,...,N$.
This shows an example of the mapping $A_{t-1}^{1:N}=\mathfrak{   A}\left(v^{1:N}_{A,t-1};\overline{w}_{t-1}^{1:N},x_{t-1}^{1:N}\right)$.
Section~\ref{sub:Conditional-Sequential-Monte Carlo} and 
section \ref{SS: SMC algorithms} of the online supplement show
a more complex mapping function $\mathfrak{A}$, which is 
required by the CPHS.

Next, we define the joint distribution of $\left(V_{x,1:T}^{1:N},V_{A,1:T-1}^{1:N}\right)$ as
\begin{align} \label{eq: joint V}
\psi\left(dV_{x,1:T}^{1:N},dV_{A,1:T-1}^{1:N}\right) :=\prod_{t=1}^{T}\prod_{i=1}^{N}\psi_{xt}\left(dV_{xt}^{i}\right)\prod_{t=1}^{T-1}\prod_{i=1}^{N}\psi_{At}\left(dV^{i}_{At}\right).
\end{align}
The SMC algorithm also provides an unbiased estimate of the likelihood
\begin{align} \label{eq: unbiased likelihood}
\widehat{Z}\left(V_{x,1:T}^{1:N},V_{A,1:T-1}^{1:N},\theta\right) :=\prod_{t=1}^{T}\left(N^{-1}\sum_{i=1}^{N}w_{t}^{i}\right),
\end{align}
which we will often write concisely as $\widehat{Z}\left(\theta\right)$.

For the Metropolis-within-Gibbs (MWG) steps in part~1 of the CPHS given in section~\ref{sub:Flexible-Correlated-PMMH+PG sampling scheme} to be effective and efficient, we require that the SMC algorithm ensures that the logs of the likelihood estimates
$\textrm{log}\widehat{Z}\left(V_{x,1:T}^{1:N},V_{A,1:T-1}^{1:N},\theta^\ast \right)$ 
and $\textrm{log}\widehat{Z}\left(V_{x,1:T}^{1:N},V_{A,1:T-1}^{1:N}, \theta\right)$ are close when $\theta^\ast$ and $\theta$ are close,
where $\theta^\ast$ is the proposed value of the parameters and $\theta$ is the current value.
This requires implementing the SMC algorithm carefully as a naive implementation introduces discontinuities in
the  logs of the estimated likelihoods due to the resampling steps when $\theta$ and $\theta^\ast$ are even slightly different.
The problem with the simple multinomial resampling described above is that
particles whose indices are close are not necessarily close themselves.
This breaks down the correlation between the logs of the likelihood estimates at the
current and proposed values. 


To reduce the variance of the difference in logs of the likelihood estimates \\ 
$\textrm{log}\, \widehat{Z}\left(V_{x,1:T}^{1:N},V_{A,1:T-1}^{1:N},\theta^\ast \right)-\textrm{log}\, \widehat{Z}\left(V_{x,1:T}^{1:N},V_{A,1:T-1}^{1:N}, \theta\right)$ appearing in the MWG acceptance ratio in the one dimensional case, it is necessary to sort the particles from the smallest to largest before resampling the particles. Sorting the particles before resampling helps to ensure that the particles whose indices are close are actually close to each other \citep{Deligiannidis2018}. Hence, the correlation in the logs of likelihood estimates is more likely to be preserved. However, this simple
sorting method does not extend easily to the multivariate case,
because we cannot sort multivariate states in this manner and guarantee
closeness of the particles. \citet{Deligiannidis2018} use the Hilbert sorting method of \citet{Gerber:2015}.
We follow \citet{Choppala2016} and use a simpler and faster Euclidean sorting procedure described in Algorithm~\ref{alg:Multidimensional-Euclidean-Sorting} in section~\ref{Multidimensional Sorting} of the online supplement. 
Algorithm \ref{alg:The-Sequential-Monte carlo algorithm} in section~\ref{SS: SMC algorithms} of the online supplement describes the correlated SMC algorithm that we use,
which is similar to that of \citet{Deligiannidis2018}. 
Algorithm \ref{alg:The-Sequential-Monte carlo algorithm} uses the
multinomial resampling scheme (Algorithm \ref{alg:Multinomial-Resampling-Algorithm}) in section~\ref{SS: SMC algorithms} of the online supplement.

\subsection{Backward simulation\label{SS: backward simulation}}
The \corrPMMHPG{} requires sampling from the particle  approximation of $p\left(x_{1:T}|y_{1:T},\theta\right)$.
Backward simulation is one approach to sampling from the particle approximation. It was 
introduced by \citet{godsilletal2004} and  used in 
\citet{olssonryden2011}; Algorithm \ref{alg:The-backward simulation algorithm} in section~\ref{SS: backward simulation1} of the online supplement describes the backward simulation algorithm.

We denote the selected particles and trajectory by $x_{1:T}^{j_{1:T}}=\left(x_{1}^{j_{1}},...,x_{T}^{j_{T}}\right)$
and $j_{1:T}$, respectively. Backward simulation 
samples the indices $J_{T},J_{T-1},...,J_{1}$ sequentially,
and differs from the ancestral tracing algorithm of \citet{kitagawa1996}, which
only samples one index and traces back its ancestral lineage. Backward simulation is
also more robust to the resampling scheme (multinomial resampling,
systematic resampling, residual resampling, or stratified resampling)
used in the resampling step of the algorithm \citep[see][]{chopinsingh2013}.

\section{The Correlated Particle Hybrid Sampler\label{S: corr PMCMC}}
This section presents the \corrPMMHPG{} for the Bayesian estimation of a state space model. Section \ref{sub:Target-Distributions univariate} presents the target distribution. Section \ref{sub:Flexible-Correlated-PMMH+PG sampling scheme} presents the sampling scheme.
Section \ref{sub:Conditional-Sequential-Monte Carlo} presents the constrained conditional sequential Monte Carlo (CCSMC) algorithm. Sections~\ref{CPMMH algorithm} and \ref{the PG algorithm} discuss the correlated pseudo marginal and the particle Gibbs with backward simulation (PGBS) samplers, respectively.
The methods introduced in this section are used in the univariate models in section~\ref{S: univariate example SV model} and the multivariate models
in section~\ref{Multivariate example}. 

\subsection{Target Distributions\label{sub:Target-Distributions univariate}}
The key idea of CPHS
is to construct a target distribution on an augmented space that includes
all the particles generated by the SMC algorithm and has
the joint posterior density of the latent states and parameters
$p\left(\theta,x_{1:T}|y_{1:T}\right)$ as a marginal.

The \corrPMMHPG{}  targets the distribution
\begin{multline}
\widetilde{\pi}^{N}\left(dv_{x,1:T}^{1:N},dv_{A,1:T-1}^{1:N},j_{1:T},d\theta\right)\coloneqq\\
\frac{p\left(dx_{1:T}^{j_{1:T}},d\theta|y_{1:T}\right)}{N^{T}}\times\frac{\psi\left(dv_{x,1:T}^{1:N},dv_{A,1:T-1}^{1:N}\right)}{m_{1}^{\theta}\left(dx_{1}^{j_{1}}\right)\prod_{t=2}^{T}
\overline{w}_{t-1}^{a_{t-1}^{j_{t}}}m_{t}^{\theta}\left(dx_{t}^{j_{t}}|x_{t-1}^{a_{t-1}^{j_{t}}}\right)}\times\\
\prod_{t=2}^{T}\frac{w_{t-1}^{a_{t-1}^{j_{t}}}f_{t}^{\theta}\left(x_{t}^{j_{t}}|x_{t-1}^{a_{t-1}^{j_{t}}}\right)}
{\sum_{l=1}^{N}w_{t-1}^{l}f_{t}^{\theta}\left(x_{t}^{j_{t}}|x_{t-1}^{l}\right)}, \label{eq:Target distribution}
\end{multline}
where $x_{1:T}^{j_{1:T}}:= \(x_1^{j_1}, x_2^{j_2}, \dots, x_T^{j_T}\)$. For brevity, in \eqref{eq:Target distribution} and below, we write $f_t^\theta(x_t|x_{t-1},y_{t-1})$
as $f_t^\theta(x_t|x_{t-1})$.
We now discuss some properties of the target distribution in \eqref{eq:Target distribution}. They are proved in section~\ref{S: proofs} of the online supplement. 

\begin{proposition}\label{lemma: target distn}
The target distribution in \eqref{eq:Target distribution}
has the marginal distribution
\[
\widetilde{\pi}^{N}\left(dx_{1:T}^{j_{1:T}},j_{1:T},d\theta\right)=\frac{p\left(dx_{1:T}^{j_{1:T}},d\theta|y_{1:T}\right)}{N^{T}},
\]
and hence, with some abuse of notation, we write $\widetilde{\pi}^{N}\left(dx_{1:T},d\theta\right)=p\left(dx_{1:T},d\theta|y_{1:T}\right)$.
\end{proposition}

The target density in \eqref{eq:Target distribution} involves densities of the basic standard normal and uniform distribution instead of the standard way in terms of densities of state particles. This is done to make the Metropolis-within-Gibbs step in Algorithm \ref{alg:Sampling-Scheme:-The correlated PMMH+PG} more efficient; see section \ref{sub:Flexible-Correlated-PMMH+PG sampling scheme} for further details.  
Theorem \ref{lemma: target distn}  shows that the marginal distribution in \eqref{eq:Target distribution}
has the joint posterior density of the latent states and parameters as the marginal density.

\begin{proposition} \label{lemma: alt expression target distn}
The target distribution in \eqref{eq:Target distribution}
can also be expressed as
\begin{align}\label{eq:Targetdistribution2}
 \widetilde{\pi}^{N}\left(dv_{x,1:T}^{1:N},dv_{A,1:T-1}^{1:N},j_{1:T},d\theta\right)  & =\nonumber\\
\frac{p\left(d\theta\right)\psi\left(dv_{x,1:T}^{1:N},dv_{A,1:T-1}^{1:N}\right)}{p\left(y_{1:T}\right)} \times \prod_{t=1}^{T}\(N^{-1}\sum_{i=1}^N w_t^i \)&
\overline{w}_{T}^{j_{T}}\prod_{t=2}^{T}\frac{w_{t-1}^{j_{t-1}}f_{t}^{\theta}\left(x_{t}^{j_{t}}|x_{t-1}^{j_{t-1}}\right)}
{\sum_{l=1}^{N}w_{t-1}^{l}f_{t}^{\theta}\left(x_{t}^{j_{t}}|x_{t-1}^{l}\right)}.
\end{align}
\end{proposition}
The target distribution in \eqref{eq:Targetdistribution2} is expressed in terms of the estimated likelihood $\widehat{Z}\left(v_{x,1:T}^{1:N},v_{A,1:T-1}^{1:N},\theta\right)=\prod_{t=1}^{T}\(N^{-1}\sum_{i=1}^N w_t^i \)$.
Corollary \ref{corr: target with j integrated out} below and Theorem \ref{lemma: alt expression target distn} are used to derive the acceptance probability in Part 1 in our
\corrPMMHPG{} (Algorithm \ref{alg:Sampling-Scheme:-The correlated PMMH+PG}).

\begin{corollary} \label{corr: target with j integrated out}
By integrating $j_{1:T}$ out of the target distribution in \eqref{eq:Targetdistribution2}
we obtain
\begin{align}
\widetilde{\pi}^{N}\left(dv_{x,1:T}^{1:N},dv_{A,1:T-1}^{1:N},d\theta\right) & = \nonumber \\
\frac{p\left(d\theta\right)\psi\left(dv_{x,1:T}^{1:N},dv_{A,1:T-1}^{1:N}\right)}{p\left(y_{1:T}\right)} & \prod_{t=1}^{T}\left(N^{-1}\sum_{i=1}^{N}w_{t}^{i}\right),
\end{align}
\end{corollary}

\begin{corollary}\label{cor: j cond on rest}
The conditional distribution $\widetilde{\pi}^{N}\left(j_{1:T}|v_{x,1:T}^{1:N},v_{A,1:T-1}^{1:N},\theta\right)$
is
\begin{equation}
\widetilde{\pi}^{N}\left(j_{1:T}|v_{x,1:T}^{1:N},v_{A,1:T-1}^{1:N},\theta\right)=\overline{w}_{T}^{j_{T}}\prod_{t=2}^{T}\frac{w_{t-1}^{j_{t-1}}
f_{t}^{\theta}\left(x_{t}^{j_{t}}|x_{t-1}^{j_{t-1}}\right)}{\sum_{l=1}^{N}w_{t-1}^{l}f_{t}^{\theta}\left(x_{t}^{j_{t}}|x_{t-1}^{l}\right)}.\label{eq:conditional distribution j}
\end{equation}
\end{corollary}
Corollary \ref{cor: j cond on rest} is used by backward simulation algorithm in
Algorithm \ref{alg:The-backward simulation algorithm} in section~\ref{SS: backward simulation1} of the online supplement.

\subsection{The CPHS \label{sub:Flexible-Correlated-PMMH+PG sampling scheme}}
We now outline a sampling scheme for the state space model
that allows the user to generate those parameters that can be drawn efficiently conditional on the states by a particle Gibbs step(s); all the other parameters are drawn in a Metropolis-within-Gibbs step(s) by conditioning on the basic uniform or standard normal random variables; e.g., parameters that are highly correlated with states, or parameters whose generation is expensive when conditioning on the states.
For simplicity, let $\theta\coloneqq\left(\theta_{1},\theta_{2}\right)$  partition the parameter vector into 2 components where each component may be a vector.
The \CPHS{} is given in Algorithm~\ref{alg:Sampling-Scheme:-The correlated PMMH+PG}; it generates the vector of parameters $\theta_{1}$ using the Metropolis-within-Gibbs step and the vector of parameters $\theta_{2}$  using the PG step. The components $\theta_{1}$ and $\theta_{2}$ can be further partitioned and sampled separately in multiple Metropolis-within-Gibbs steps and multiple PG steps, respectively.

We note that the correlated pseudo-marginal method of \citet{Deligiannidis2018}
correlates the random numbers, $(v^{1:N}_{x,1:T},v^{1:N}_{A,1:T-1})$ and $(v^{\ast,1:N}_{x,1:T},v^{\ast,1:N}_{A,1:T-1})$  used in constructing the estimators of the
likelihood at the current and proposed values of the parameters. This correlation is set very close to 1  to
reduce the variance of the difference in the logs of the estimated likelihoods $\log\, \widehat{Z}\left(\theta^{*},v^{\ast,1:N}_{x,1:T},v^{\ast,1:N}_{A,1:T-1}\right)-\log\, \widehat{Z}\left(\theta, v^{1:N}_{x,1:T},v^{1:N}_{A,1:T-1}\right)$
appearing in the MH acceptance ratio.
It is easy to see that part~1 (MWG sampling) of our sampling scheme is a special case of the correlated pseudo-marginal step with the correlation set to 1. We condition
on the same set of random numbers $\left(V_{x,1:T}^{1:N},V_{A,1:T-1}^{1:N}\right)$
that is generated using CCSMC in part~4 of 
Algorithm~\ref{alg:Sampling-Scheme:-The correlated PMMH+PG}.
That is, in our scheme, we deal with
$\log\, \widehat{Z}\left(\theta_{1}^{*},\theta_{2}, \left(V_{x,1:T}^{1:N},V_{A,1:T-1}^{1:N}\right)\right)-\log\, \widehat{Z}\left(\theta_{1},\theta_{2}, \left(V_{x,1:T}^{1:N},V_{A,1:T-1}^{1:N}\right)\right),$
which is the difference in the logs of the estimated likelihoods at the proposed and current values of the parameters.
The basic random numbers $\left(V_{x,1:T}^{1:N},V_{A,1:T-1}^{1:N}\right)$ are fixed in the Metropolis-within-Gibbs step (part 1),
but they are updated in part~4 of Algorithm \ref{alg:Sampling-Scheme:-The correlated PMMH+PG}.
On the other hand, Part 1 (PMMH sampling) of PHS proposed by \citet{Mendes2020} is a special case of the correlated pseudo-marginal step with the correlation set to $0$.

\begin{algorithm}[H]
\caption{ The Correlated Particle Hybrid Sampler (\corrPMMHPG{}).\label{alg:Sampling-Scheme:-The correlated PMMH+PG}}

Given initial values for $V_{x,1:T}^{1:N}$, $V_{A,1:T-1}^{1:N}$, $J_{1:T}$,
and $\theta$
\begin{enumerate}
\item [Part 1:] MWG sampling. 
\begin{enumerate}
\item Sample $\theta_{1}^{*}\sim q_{1}\left(\cdot|v_{x,1:T}^{1:N},v_{A,1:T-1}^{1:N},\theta_{2},\theta_{1}\right)$
\item Run the sequential Monte Carlo algorithm and evaluate $\widehat{Z}\left(v^{1:N}_{x,1:T},v_{A,1:T-1}^{1:N},\theta_{1}^{*},\theta_{2}\right)$.
\item Accept the proposed values $\theta_{1}^{*}$ with probability
\begin{align}\label{eq: MH ratio1}
\alpha\left(\theta_{1};\theta_{1}^{*}|v_{x,1:T}^{1:N},v_{A,1:T-1}^{1:N},\theta_{2}\right) & =\notag \\
1\land\frac{\widehat{Z}\left(v_{x,1:T}^{1:N},v_{A,1:T-1}^{1:N},\theta_{1}^{*},\theta_{2}\right)p\left(\theta_{1}^{*}|\theta_{2}\right)}
{\widehat{Z}\left(v_{x,1:T}^{1:N},v_{A,1:T-1}^{1:N},\theta_{1},\theta_{2}\right)p\left(\theta_{1}|\theta_{2}\right)} & \times\frac{q_{1}\left(\theta_{1}|v_{x,1:T}^{1:N},v_{A,1:T-1}^{1:N},\theta_{2},\theta_{1}^{*}\right)}
{q_{1}\left(\theta_{1}^{*}|v_{x,1:T}^{1:N},v_{A,1:T-1}^{1:N},\theta_{2},\theta_{1}\right)}.
\end{align}

\end{enumerate}
\item [Part 2:] Sample $J_{1:T}=j_{1:T}\sim\widetilde{\pi}^{N}\left(\cdot|v_{x,1:T}^{1:N},v_{A,1:T-1}^{1:N},\theta\right)$
given in \eqref{eq:conditional distribution j} using the
backward simulation algorithm (Algorithm~\ref{alg:The-backward simulation algorithm} in section \ref{SS: backward simulation1} of the online supplement)

\item [Part 3:] PG sampling. 
\begin{enumerate}
\item Sample $\theta_{2}^{*}\sim q_{2}\left(\cdot|x_{1:T}^{j_{1:T}},j_{1:T},\theta_{1},\theta_{2}\right)$
\item Accept the proposed values $\theta_{2}^{*}$ with probability
\begin{align}
\alpha\left(\theta_{2};\theta_{2}^{*}|x_{1:T}^{j_{1:T}},j_{1:T},\theta_{1}\right)\nonumber \\
= & 1\land\frac{\widetilde{\pi}^{N}\left(\theta_{2}^{*}|x_{1:T}^{j_{1:T}},j_{1:T},\theta_{1}\right)}{\widetilde{\pi}^{N}\left(\theta_{2}|x_{1:T}^{j_{1:T}},j_{1:T},
\theta_{1}\right)}\frac{q_{2}\left(\theta_{2}|x_{1:T}^{j_{1:T}},j_{1:T},\theta_{1},\theta_{2}^{*}\right)}
{q_{2}\left(\theta_{2}^{*}|x_{1:T}^{j_{1:T}},j_{1:T},\theta_{1},\theta_{2}\right)}. \label{eq: MH ratio 2}
\end{align}
\end{enumerate}
\item [Part 4:] Sample $\left(V_{x,1:T}^{1:N},V_{A,1:T-1}^{1:N}\right)$ from $\widetilde{\pi}^{N}\left(\cdot|x_{1:T}^{j_{1:T}},j_{1:T},\theta\right)$
using the constrained conditional sequential Monte Carlo algorithm (Algorithm \ref{alg:The-conditional Sequential-Monte carlo algorithm})
and evaluate $\widehat{Z}\left(v_{x,1:T}^{1:N},v_{A,1:T-1}^{1:N},\theta\right)$.
\end{enumerate}
\end{algorithm}
Part~1 (MWG sampling): from corollary~\ref{corr: target with j integrated out}, we can show that 
\begin{align*}
&\frac{\widetilde{\pi}^{N}\left(v_{x,1:T}^{1:N},v_{A,1:T-1}^{1:N},\theta_1^\ast|\theta_{2}\right)}
{\widetilde{\pi}^{N}\left(v_{x,1:T}^{1:N},v_{A,1:T-1}^{1:N},\theta_1|\theta_{2}\right)}
 = \frac{\widehat{Z}\left(v_{x,1:T}^{1:N},v_{A,1:T-1}^{1:N},\theta_1^{*},\theta_{2}\right)p\left(\theta_1^\ast|\theta_{2}\right)}
{\widehat{Z}\left(v_{x,1:T}^{1:N},v_{A,1:T-1}^{1:N},
\theta_1,\theta_{2}\right)p\left(\theta_1|\theta_{2}\right)},
\end{align*}
where
$\widehat{Z}\left(v_{x,1:T}^{1:N},v_{A,1:T-1}^{1:N},\theta\right)$ is defined in \eqref{eq: unbiased likelihood}.
This leads to the acceptance probability in \eqref{eq: MH ratio1}.
Part~2 follows from Corollary~\ref{cor: j cond on rest}. Using theorem ~\ref{lemma: target distn},
we can show that the Metropolis-Hastings in \eqref{eq: MH ratio 2} (Part 3: PG sampling) simplifies to 
\begin{equation}
\frac{p\left(y_{1:T}|x_{1:T}^{j_{1:T}},\theta_{2}^{*},\theta_{1}\right)
p\left(x_{1:T}^{j_{1:T}}|\theta_{2}^{*},\theta_{1}\right)p\left(\theta_{2}^{*}|\theta_{1}\right)}{p\left(y_{1:T}|x_{1:T}^{j_{1:T}},\theta_{2},\theta_{1}\right)
p\left(x_{1:T}^{j_{1:T}}|\theta_{2},\theta_{1}\right)p\left(\theta_{2}|\theta_{1}\right)}
\times\frac{q_{2}\left(\theta_{2}|x_{1:T}^{j_{1:T}},j_{1:T},\theta_{1},\theta_{2}^{*}\right)}{q_{2}\left(\theta_{2}^{*}|x_{1:T}^{j_{1:T}},j_{1:T},\theta_{1},\theta_{2}\right)}.
\end{equation}
Part~4 updates the basic random numbers $v_{x,1:T}^{1:N}$ and $v_{A,1:T-1}^{1:N}$ using constrained conditional sequential Monte Carlo (section \ref{sub:Conditional-Sequential-Monte Carlo})
and follows from \eqref{eq:Target distribution} (the target), theorem~\ref{lemma: target distn} and Definitions~\ref{def: Xt and Vx} and \ref{def: definition V}.

Algorithm \ref{alg:Sampling-Scheme:-The full correlated PMMH+PG} in section \ref{SS: The full correlated PMMH+PG} of the supplementary material gives a more general sampling scheme than Algorithm \ref{alg:Sampling-Scheme:-The correlated PMMH+PG}
with Part 1 having a Metropolis-within-Gibbs proposal that potentially moves the basic uniform and standard normal variables as well as the parameter of interest.
Possible Metropolis-within-Gibbs proposals in Part 1 include the correlated pseudo-marginal approach in \citet{Deligiannidis2018}, the block pseudo-marginal approach in \citet{Tran:2016}, and the standard pseudo-marginal approach in \citet{Andrieu:2009}.
The advantage of Algorithm \ref{alg:Sampling-Scheme:-The correlated PMMH+PG} over Algorithm \ref{alg:Sampling-Scheme:-The full correlated PMMH+PG} in section \ref{SS: The full correlated PMMH+PG} of the online supplement
is that Part 1 has fewer tuning parameters that need to be chosen by the user
and is computationally more efficient because the basic uniform and standard normal variables are fixed and not generated. 

\subsection{Constrained conditional sequential Monte Carlo \label{sub:Conditional-Sequential-Monte Carlo}}
This section discusses the constrained conditional sequential Monte Carlo (CCSMC) algorithm  (Algorithm~\ref{alg:The-conditional Sequential-Monte carlo algorithm}),
which is used in Part~4 of the~\corrPMMHPG{} (Algorithm~\ref{alg:Sampling-Scheme:-The correlated PMMH+PG}).
The CCSMC algorithm is a sequential Monte Carlo algorithm in which a particle
trajectory $x_{1:T}^{j_{1:T}}=\left(x_{1}^{j_{1}},...,x_{T}^{j_{T}}\right)$
and the associated sequence of indices $j_{1:T}$ are kept unchanged, which means that some elements of $V_{x,1:T}^{1:N}$ and $V_{A,1:T-1}^{1:N}$
are constrained.  It is a constrained version of the conditional SMC sampler in \cite{andrieuetal2010} that follows from the density of all the random variables that are generated
by sequential Monte Carlo algorithm conditional on $\left(x_{1:T}^{j_{1:T}},j_{1:T},\theta,y_{1:T}\right)$,
\begin{equation}
\frac{\psi\left(dv_{x,1:T}^{1:N},dv_{A,1:T-1}^{1:N}\right)}{m_{1}^{\theta}\left(dx_{1}^{j_{1}}\right)\prod_{t=2}^{T}\overline{w}_{t-1}^{a_{t-1}^{j_{t}}}m_{t}^{\theta}\left(dx_{t}^{j_{t}}|x_{t-1}^{a_{t-1}^{j_{t}}}\right)},
\end{equation}
which appears in the augmented target density in \eqref{eq:Target distribution}.


The CCSMC algorithm takes the number of particles $N$, the parameters $\theta$, and the reference trajectory $x_{1:T}^{j_{1:T}}$ as inputs.
It produces the set of particles $x_{1:T}^{1:N}$, ancestor indices $a_{1:T-1}^{1:N}$, and weights $w_{1:T}^{1:N}$; 
it also produces the random variables used to propagate state particles $V_{x,1:T}^{1:N}$ and the random numbers used in the resampling steps $V_{A,1:T-1}^{1:N}$.
Both sets of random variables are used as the inputs of the sequential Monte Carlo algorithm in part~(1b) of 
Algorithm~\ref{alg:Sampling-Scheme:-The correlated PMMH+PG}. 

\begin{algorithm}[H]
\caption{The constrained conditional sequential Monte Carlo algorithm \label{alg:The-conditional Sequential-Monte carlo algorithm} }

Inputs: $N$, $\theta$, $x_{1:T}^{j_{1:T}}$, and $j_{1:T}$

Outputs: $x_{1:T}^{1:N}$, $a_{1:T-1}^{1:N}$, $w_{1:T}^{1:N}$, $V_{x,1:T}^{1:N}$, and $V_{A,1:T-1}^{1:N}$

Fix $X_{1:T}^{j_{1:T}}=x_{1:T}^{j_{1:T}}$, $A_{1:T-1}^{J}=j_{1:T-1}$,
and $J_{T}=j_{T}$.
\begin{enumerate}
\item For $t=1$

\begin{enumerate}
\item Sample $v_{x1}^{i}\sim\psi_{x1}\left(\cdot\right)$ and set $X_{1}^{i}=x_1^{i} = \mathfrak{X}\left(v_{x1}^{i}; \theta, \cdot\right)$
for $i=1,...,N\setminus\left\{ j_{1}\right\} $.
\item Obtain $v_{x1}^{j_{1}}$ such that $x_{1}^{j_{1}}=\mathfrak{X}\left(v_{x1}^{j_{1}};\theta,\cdot\right)$.
\item Compute the importance  weights $w_{1}^{i}=\frac{f_{1}^{\theta}\left(x_{1}^{i}\right)g_{1}^{\theta}\left(y_{1}|x_{1}^{i}\right)}{m_{1}^{\theta}\left(x_{1}^{i}\right)}$,
for $i=1,...,N$,
and normalize $\overline{w}_{1}^{i}= w_{1}^{i}/\sum_{j=1}^{N}w_{1}^{j}$.
\end{enumerate}

\item For $t\geq2$

\begin{enumerate}

\item Sort the particles $x_{t-1}^{i}$ using the simple sorting procedure of \citet{Choppala2016} and obtain the sorted index $\zeta_{i}$ for $i=1,...,N$ and the sorted particles and weights $\wt x_{t-1}^i = x_{t-1}^{\zeta_i} $ and $\wt {\ov w}^i_{t-1} = \ov w_{t-1}^{\zeta_i}$, for $i=1, \dots, N$.

\item Use a  constrained sampling algorithm, for example, the constrained multinomial sampler
(Algorithm~\ref{alg:Multinomial-Resampling-Algorithm for CCSMC} in section \ref{constrainedresampling} of the online supplement),
\begin{enumerate} \item Generate the random variables used in the resampling step $v_{At-1}^{1:N}$ and obtain the ancestor indices based on the sorted particles $\widetilde{A}_{t-1}^{1:N\setminus\left(j_{t}\right)}$.
\item Obtain the ancestor indices in the original order of the particles $A_{t-1}^{1:N\setminus\left(j_{t}\right)}$.
\end{enumerate}
\item Sample $v_{xt}^{i}\sim\psi_{xt}\left(\cdot\right)$ for $i=1,...,N\setminus\left\{ j_{t}\right\}$
and obtain $v_{xt}^{j_{t}}$ such that $x_{t}^{j_{t}}=\mathfrak{X}(v_{xt}^{j_{t}};\theta,x_{t-1}^{a_{t-1}^{j_{t}}})$.
\item Set $x_{t}^{i}=\mathfrak{X}\left(v_{xt}^{i};\theta,x_{t-1}^{a_{t-1}^{i}}\right)$
for $i=1,...,N\setminus\left\{ j_{t}\right\} $.
\item Compute the importance weights,
\begin{align*}
    w_{t}^{i} & =\frac{f_{t}^{\theta}\left(x_{t}^{i}|x_{t-1}^{a_{t-1}^{i}},y_{t-1}\right)g_{t}^{\theta}
   \left(y_{t}|x_{t}^{i}\right)}{m_{t}^{\theta}\left(x_{t}^{i}|x_{t-1}^{a_{t-1}^{i}}\right)},
\quad \text{for} \quad i=1,...,N,
\end{align*}
and  normalize the $\overline{w}_{t}^{i}$.
\end{enumerate}
\end{enumerate}
\end{algorithm}

At $t=1$, step (1a) samples the basic random numbers $v_{x1}^{i}$ from the standard normal distribution $\N(0,1)$ for $i=1,...,N$ and obtains the set of particles $x_{1}^{1:N}$ using \eqref{V_x example SV1}, except for the reference particles $x_{1}^{j_{1}}$. We obtain the basic random number  $v_{x1}^{j_{1}}$ associated with the reference particle  $x_{1}^{j_{1}}$ in step (1b) using 
\begin{align*}
\quad v_{x1}^{j_{1}} & =\bigg ( \left ( 1-\phi^{2}\right)/\tau^{2}\bigg)^{\frac12}  \left(x_{1}^{j_{1}}-\mu\right), 
\end{align*}
for the univariate SV model with leverage using the bootstrap filter and compute the weights of all the particles in step (1c).  

Section~\ref{SS: smc} shows that when the parameters $\theta$ 
change, the resampling step creates discontinuities and breaks 
the relationship between the likelihood terms at the current and 
proposed values. Step (2a) sorts the particles from 
smallest to largest using the simple sorting procedure of 
\citet{Choppala2016} and obtains the sorted particles and weights.
The particles are then resampled using constrained multinomial resampling given in Algorithm~\ref{alg:Multinomial-Resampling-Algorithm for CCSMC} in section~\ref{constrainedresampling} of the online supplement
and the ancestor index $A_{1:T-1}^{1:N\setminus\left(j_{t}\right)}$ in the original order of the particles in step~(2b) is obtained.
Step (2c) samples the basic random numbers $v_{xt}^{i}$ from the standard normal distribution $\N(0,1)$ for $i=1,...,N$ and obtains the set of particles $x_{t}^{1:N}$
using \eqref{V_x example SVt}, except for the reference particles $x_{t}^{j_{t}}$.
We obtain the basic random number  $v_{xt}^{j_{t}}$ associated with the reference particle  $x_{t}^{j_{t}}$ in step (2d) using
\begin{align*}
\quad  v_{xt}^{j_{t}} & = \frac{ \bigg ( x_{t}^{j_{t}}-\mu-\phi\left(x_{t-1}^{a_{t-1}^{j_{t}}}-
\mu\right)-\rho\tau\exp\left(-x_{t-1}^{a_{t-1}^{j_{t}}}/2\right)y_{t-1}
\bigg )}   { \bigg (\tau^{2}\left(1-\rho^{2}\right)\bigg )^{\frac12}}
\end{align*}
for the univariate SV model with leverage and
the bootstrap filter and compute the weights of all the 
particles in step (2e).  


\subsection{Correlated pseudo marginal Metropolis-Hastings (CPMMH) Algorithm \label{CPMMH algorithm}}

This section shows that the correlated pseudo marginal Metropolis-Hastings
(\CPMMH) algorithm of \citet{Deligiannidis2018} is a special case
of Algorithm~\ref{alg:Sampling-Scheme:-The correlated PMMH+PG} in section~\ref{sub:Flexible-Correlated-PMMH+PG sampling scheme}. This approach is useful for generating
parameters that are highly correlated with the state variables. Algorithm~\ref{alg:The-Correlated-Pseudo Marginal} is the \CPMMH{} algorithm. 
One iteration of the \CPMMH{} is now discussed. 

\begin{algorithm}[H]
\caption{The correlated pseudo marginal Metropolis-Hastings of \citet{Deligiannidis2018}
\label{alg:The-Correlated-Pseudo Marginal}}

Given initial values for $V_{x,1:T}^{1:N}$, $V_{A,1:T-1}^{1:N}$, $J_{1:T}$,
and $\theta$
\begin{enumerate}
\item Sample $\left(\theta^{*},v_{x,1:T}^{*,1:N},v_{A,1:T-1}^{*,{1:N}}\right)\sim q\left(\theta^{*},\d v_{x,1:T}^{*,1:N},\d v_{A,1:T-1}^{*,1:N}|v_{x,1:T}^{1:N},v_{A,1:T-1}^{1:N},\theta\right)$.
\item Run the sequential Monte Carlo algorithm 
(Algorithm~\ref{alg:The-Sequential-Monte carlo algorithm}), estimate the likelihood $\widehat{Z}\left(\theta^{*},v_{x,1:T}^{*,1:N},v_{A,1:T-1}^{*,1:N}\right)$,
and sample $J_{1:T}^{*}=j_{1:T}^{*}$ using backward simulation
(Algorithm~\ref{alg:The-backward simulation algorithm} in section~\ref{SS: backward simulation1} of the online supplement).
\item Accept the proposed values $\left(\theta^{*},v_{x,1:T}^{*,1:N},v_{A,1:T-1}^{*,1:N},j_{1:T}^{*}\right)$
with probability
\begin{multline}
\alpha\left(\theta,v_{x,1:T}^{,1:N},v_{A,1:T-1}^{1:N},j_{1:T};\theta^{*},v_{x,1:T}^{*,1:N},v_{A,1:T-1}^{*,1:N},j_{1:T}^{*}\right)=\\
1\land\frac{\widehat{Z}\left(v_{x,1:T}^{*,1:N},v_{A,1:T-1}^{*,1:N},\theta^{*}\right)p\left(\theta^{*}\right)\psi\left(dv_{x,1:T}^{*,1:N},dv_{A,1:T-1}^{*,1:N}\right)}{\widehat{Z}\left(v_{x,1:T}^{1:N},v_{A,1:T-1}^{1:N},\theta\right)p\left(\theta\right)\psi\left(\d v_{x,1:T}^{1:N},\d v_{A,1:T-1}^{1:N}\right)}\\
\times\frac{q\left(\theta,\d v_{x,1:T}^{1:N},\d v_{A,1:T-1}^{1:N}|v_{x,1:T}^{*,1:N},v_{A,1:T-1}^{*,1:N},\theta^{*}\right)}{q\left(v_{x,1:T}^{*,1:N},v_{A,1:T-1}^{*,1:N},\theta^{*}|\theta,\d v_{x,1:T}^{1:N},\d v_{A,1:T-1}^{1:N}\right)}.\label{eq:acceptanceprobCPMMH}
\end{multline}
\end{enumerate}
\end{algorithm}
  
Step~1 of  Algorithm~\ref{alg:The-Correlated-Pseudo Marginal} generates the proposal
for parameter $\theta$ and the random numbers $v_{x,1:T}^{*,1:N}$
and $v_{A,1:T-1}^{*,1:N}$ from the proposal density 
\begin{multline}
q\left(v_{x,1:T}^{*,1:N},v_{A,1:T-1}^{*,1:N},\theta^{*}|\theta,dv_{x,1:T}^{1:N},dv_{A,1:T-1}^{1:N}\right)=\\
q_{\theta}\left(\theta^{*}|\theta,v_{x,1:T}^{1:N},v_{A,1:T-1}^{1:N}\right)q\left(v_{x,1:T}^{*,1:N},v_{A,1:T-1}^{*,1:N}|\theta,v_{x,1:T}^{1:N},v_{A,1:T-1}^{1:N}\right).
\label{proposalCPMMH}
\end{multline}
Following~\citet{Deligiannidis2018}, the proposal for the random
variables $v_{x,1:T}^{*,1:N}$ is
\begin{equation}
v_{x,t}^{*i}=\rho_x v_{x,t}^{i}+\sqrt{1-\rho_x^{2}}\epsilon_{x,t}^{*i},
\end{equation}
for $t=1,...,T$ and $i=1,...,N$, where $\epsilon_{x,t}^{*i}\sim\N\left(0,1\right)$
and $\rho_x\in\left(-1,1\right)$ is the correlation coefficient which is set very close to 1. For the uniform random numbers, $v_{A,1:T-1}^{1:N}$,
the transformation to normality is first applied to obtain $z_{t}^{i}=\Phi^{-1}\left(v_{A,t}^{i}\right)$
and set $z_{t}^{*i}=\rho_x z_{t}^{i}+\sqrt{1-\rho_x^{2}}\epsilon_{A,t}^{*i}$,
where $\epsilon_{A,t}^{*i}\sim\N\left(0,1\right)$. We then obtain
$v_{A,t}^{i*}=\Phi\left(z_{t}^{*i}\right)$ for $t=1,...,T-1$ and
$i=1,...,N$. Step~2 runs the sequential Monte Carlo algorithm
(Algorithm \ref{alg:The-Sequential-Monte carlo algorithm}) in section~\ref{SS: SMC algorithms} of the online supplement, estimates the likelihood $\widehat{Z}\left(\theta^{*},v_{x,1:T}^{*,1:N},v_{A,1:T-1}^{*,1:N}\right)$
evaluated at the proposed value $\left(\theta^{*},v_{x,1:T}^{*,1:N},v_{A,1:T-1}^{*,1:N}\right)$,
and samples the index $j_{1:T}^{*}$ using backward simulation
(Algorithm \ref{alg:The-backward simulation algorithm}) in section~\ref{SS: backward simulation1} of the online supplement. Finally, the proposed values $\left(\theta^{*},v_{x,1:T}^{*,1:N},v_{A,1:T-1}^{*,1:N},j_{1:T}^{*}\right)$
are accepted with the probability in \eqref{eq:acceptanceprobCPMMH}. Using the proposal in \eqref{proposalCPMMH}, the acceptance probability simplifies to
\begin{multline}
\alpha\left(\theta,v_{x,1:T}^{,1:N},v_{A,1:T-1}^{1:N},j_{1:T};\theta^{*},v_{x,1:T}^{*,1:N},v_{A,1:T-1}^{*,1:N},j_{1:T}^{*}\right)=\\
1\land\frac{\widehat{Z}\left(v_{x,1:T}^{*,1:N},v_{A,1:T-1}^{*,1:N},\theta^{*}\right)p\left(\theta^{*}\right)}{\widehat{Z}\left(v_{x,1:T}^{1:N},v_{A,1:T-1}^{1:N},\theta\right)p\left(\theta\right)}\times\frac{q_{\theta}\left(\theta|\theta^{*},v_{x,1:T}^{*,1:N},v_{A,1:T-1}^{*,1:N}\right)}{q_{\theta}\left(\theta^{*}|\theta,v_{x,1:T}^{1:N},v_{A,1:T-1}^{1:N}\right)}.\label{eq:acceptanceprobCPMMH-1}
\end{multline}
If the proposed values are accepted, then  $\theta=\theta^{*}$,
$v_{x,t}^{i}=v_{x,t}^{*i}$ for all $t=1,...,T$, $v_{A,t}^{i}=v_{A,t}^{*i}$
for all $t=1,...,T-1$, $i=1,...,N$, and the selected
particles are updated to $x_{1:T}^{j_{1:T}}=x_{1:T}^{j_{1:T}^{*}}$. 
 

\subsection{Particle Gibbs with backward simulation (PGBS) algorithm \label{the PG algorithm}}

This section shows that the particle Gibbs of \citet{andrieuetal2010}
is also a special case of the Algorithm~\ref{alg:Sampling-Scheme:-The correlated PMMH+PG} in section~\ref{sub:Flexible-Correlated-PMMH+PG sampling scheme}, with a few
modifications. The basic idea of the PG algorithm is to generate the parameters
conditional on the states and the states conditional on the parameters.


The particle Gibbs of \citet{andrieuetal2010} uses ancestral tracing \citep{kitagawa1996} to sample the particle approximation of $p(x_{1:T}|y_{1:T},\theta)$. Ancestral tracing samples an index $J=j$ with probability $\overline{w}^j_{T}$, and then traces back its ancestral lineage $b^{j}_{1:T}$, $(b^{j}_{T}=j,b^{j}_{t-1}=a^{b^{j}_t}_{t-1})$ and selects the particle $x^{j}_{1:T}=(x^{b^{j}_1}_1,...,x^{b^{j}_T}_T)$. \citet{Whiteley2010}, \citet{lindstenschon2012a} and \citet{Lindsten2013} suggest that ancestor sampling leads to poor mixing of the Markov chain for the latent states. They use backward simulation to sample from 
$p(x_{1:T}|y_{1:T},\theta)$ and show that the mixing of the Markov chain improves appreciably. We call their sampler the particle Gibbs with backward simulation (\PGBS). Backward simulation is discussed in section~\ref{SS: backward simulation} and Algorithm~\ref{alg:The-backward simulation algorithm} in section~\ref{SS: backward simulation1} of the online supplement.

The  PG algorithm uses the standard conditional sequential
Monte Carlo (CSMC) (Algorithm~\ref{alg:The-conditional Sequential-Monte carlo algorithm standard PG} in 
section~\ref{CSMC_algorithm} of the online supplement) of \citet{andrieuetal2010}  instead of constrained
conditional sequential Monte Carlo (Algorithm~\ref{alg:The-conditional Sequential-Monte carlo algorithm}). 
We note that the CSMC algorithm is a special case of the CCSMC algorithm. It does not need to (1)~obtain $v^{j_t}_{x,t}$ for $t=1,...,T$ (steps 1b and 2c of the CCSMC algorithm) and 
(2)~sort the particles before resampling (step 2a). The constrained multinomial resampling is replaced by the conditional multinomial resampler in 
Algorithm~\ref{alg:Multinomial-Resampling-Algorithm for standard CSMC} in section~\ref{CSMC_algorithm} of the online supplement). The conditional multinomial resampling scheme takes the particles ${x}_{t-1}^{1:N}$ and weights ${\overline{w}}_{t-1}^{1:N}$ as the inputs and produces the ancestor indices ${A}_{t-1}^{1:N}$. The first step computes the cumulative weights of the particles; the second step generates the ancestor indices ${A}_{t-1}^{1:N}$.  

We now discuss a single iteration of the \PGBS{} algorithm (Algorithm~\ref{alg:The-Particle-Gibbs}). 
The first step of  Algorithm~\ref{alg:The-Particle-Gibbs} generates the proposal for the parameter $\theta$. The proposal is accepted with the probability
in  \eqref{eq:acceptanceprobabilityPG}. These two steps are conditioned
on the selected particles $x_{1:T}^{j_{1:T}}$ from the previous iteration. Using theorem~\ref{lemma: target distn},
we can show that the Metropolis-Hastings ratio in \eqref{eq:acceptanceprobabilityPG} simplifies to 
\begin{equation}
\frac{p\left(y_{1:T}|x_{1:T}^{j_{1:T}},\theta_{2}^{*},\theta_{1}\right)
p\left(x_{1:T}^{j_{1:T}}|\theta_{2}^{*},\theta_{1}\right)p\left(\theta_{2}^{*}|\theta_{1}\right)}{p\left(y_{1:T}|x_{1:T}^{j_{1:T}},\theta_{2},\theta_{1}\right)
p\left(x_{1:T}^{j_{1:T}}|\theta_{2},\theta_{1}\right)p\left(\theta_{2}|\theta_{1}\right)}
\times\frac{q_{2}\left(\theta_{2}|x_{1:T}^{j_{1:T}},j_{1:T},\theta_{1},\theta_{2}^{*}\right)}{q_{2}\left(\theta_{2}^{*}|x_{1:T}^{j_{1:T}},j_{1:T},\theta_{1},\theta_{2}\right)}.
\end{equation}
Step 3 is the CSMC algorithm (Algorithm \ref{alg:The-conditional Sequential-Monte carlo algorithm standard PG} in 
section~\ref{CSMC_algorithm} of the online supplement) that generates $N-1$
new particles while keeping a particle trajectory $x_{1:T}^{j_{1:T}}=\left(x_{1}^{j_{1}},x_{2}^{j_{2}},...,x_{T}^{j_{T}}\right)$
and the associated indices $j_{1:T}$ fixed. 
Step~4 samples the new
index vector $j_{1:T}=\left(j_{1},...,j_{T}\right)$ using Algorithm \ref{alg:The-backward simulation algorithm} in section \ref{SS: backward simulation1} of the online supplement to carry out backward simulation. 

\begin{algorithm}[H]
\caption{The Particle Gibbs with backward simulation (PGBS) algorithm \label{alg:The-Particle-Gibbs}}

Given initial values for $V_{x,1:T}^{1:N}$, $V_{A,1:T-1}^{1:N}$, $J_{1:T}$,
and $\theta$. 
\begin{enumerate}
\item Sample $\theta\sim q\left(\cdot|x_{1:T}^{j_{1:T}},j_{1:T}\right)$.
\item Accept the proposed values $\theta^{*}$ with probability
\begin{multline}
\alpha\left(\theta;\theta^{*}|x_{1:T}^{j_{1:T}},j_{1:T}\right)=\\
1\land\frac{\widetilde{\pi}^{N}\left(\theta^{*}|x_{1:T}^{j_{1:T}},j_{1:T}\right)}{\widetilde{\pi}^{N}\left(\theta|x_{1:T}^{j_{1:T}},j_{1:T}\right)}\times\frac{q\left(\theta|x_{1:T}^{j_{1:T}},j_{1:T},\theta^{*}\right)}{q\left(\theta^{*}|x_{1:T}^{j_{1:T}},j_{1:T},\theta\right)}.\label{eq:acceptanceprobabilityPG}
\end{multline}
\item Sample $\left(X_{1:T}^{1:N},A_{1:T-1}^{1:N}\right)$ from $\widetilde{\pi}\left(\cdot|x_{1:T}^{j_{1:T}},j_{1:T},\theta\right)$
using the conditional sequential Monte Carlo algorithm 
(Algorithm~\ref{alg:The-conditional Sequential-Monte carlo algorithm standard PG} in section~\ref{CSMC_algorithm} of the online supplement).
\item Sample $J_{1:T}=j_{1:T}\sim\widetilde{\pi}^{N}\left(\cdot|x_{1:T}^{j_{1:T}},j_{1:T},\theta\right)$
using the backward simulation algorithm (Algorithm \ref{alg:The-backward simulation algorithm} in section~\ref{SS: backward simulation1} of the online supplement). 
\end{enumerate}
\end{algorithm}

\section{Univariate Example: The univariate stochastic volatility model with leverage \label{S: univariate example SV model}}
This section applies the \corrPMMHPG{} defined in
section \ref{sub:Flexible-Correlated-PMMH+PG sampling scheme} to the univariate SV model with leverage.

\subsection{Preliminaries\label{SS: preliminaries}}

To define our measure of the inefficiency of a sampler 
that takes computing time into account, we first define the 
integrated autocorrelation time (IACT) for a univariate 
function $\psi(\theta)$ of $\theta$ as
\[
\textrm{IACT}({\psi})=1+2\sum_{j=1}^{\infty}\rho_{j,\psi},
\]
where $\rho_{j,\psi}$ is the $j$th autocorrelation 
of the iterates of $\psi(\theta)$ in the MCMC after 
the chain has converged. We use the CODA package of
\citet{Plummer2006} to estimate the IACT values of the 
parameters.
A low value of the IACT estimate suggests that the 
Markov chain mixes well. Our measure of the inefficiency
of a sampler for a given parameter $\theta$ based on 
$\textrm{IACT}(\psi)$ is
the time normalised variance (TNV),
$
\textrm{TNV}(\psi) = \textrm{IACT}(\psi)\times\textrm{CT},
$
where CT is the computing time in seconds per iteration of the MCMC. The estimate of TNV  is the estimate of the IACT times the computing time.
The relative time normalized variance of the sampler for $\psi$ (${\rm RTNV}(\psi)$) is the TNV relative to the TNV for the \corrPMMHPG{}.

For a given sampler and for a model with a large number of parameters, let $\textrm{IACT}_{\textrm{MAX}}(\theta)$ and $\textrm{IACT}_{\textrm{MEAN}}(\theta)$
be the maximum and mean of the  IACT values over all the
parameters in the model. In addition, let $\textrm{IACT}_{\textrm{MAX}}(x_{1:T})$ and $\textrm{IACT}_{\textrm{MEAN}}(x_{1:T})$ be the maximum and mean of the IACT values over the latent variables $x_{1:T}$.

\subsection{Empirical Results \label{Univariate Empirical Results}}

We use the following notation from \citet{Mendes2020},
$\textrm{PMMH}\left(\theta\right)$
means using PMMH to sample the parameter vector $\theta$;
$\textrm{PHS}\left(\theta_{1};\theta_{2}\right)$
means sampling $\theta_{1}$ in the PMMH step and $\theta_{2}$ in
the PG step; and $\textrm{PG}\left(\theta\right)$ means sampling
$\theta$ using the PG sampler.
Finally, $\textrm{\corrPMMHPG}\left(\theta_{1};\theta_{2}\right)$ means sampling $\theta_{1}$ in the MWG step in part~1 of Algorithm \ref{alg:Sampling-Scheme:-The correlated PMMH+PG}
and $\theta_{2}$ in the PG step.

Our approach for determining which parameters to estimate by MWG and which to estimate by  PG in the sampling scheme
is to first run the
PG algorithm for all the parameters to identify which parameters
have large IACT's. We then generate these parameters in the MWG step.

{\bf Priors:} 
We now specify the prior distributions of the parameters.  We follow \citet{Kim:1998}
and choose the prior for the persistence parameter  $\phi $
as  $\left(\phi+1\right)/2\sim\textrm{Beta}\left(a_{0},b_{0}\right)$,  with $a_{0}=100$ and $b_{0}=1.5$, i.e.,

\begin{align*} 
p\left(\phi\right)=\frac{1}{2B\left(a_{0},b_{0}\right)}\left(\frac{1+\phi}{2}\right)^{a_{0}-1}\left(\frac{1-\phi}{2}\right)^{b_{0}-1}.
\end{align*}
The prior for $\tau$ is the half Cauchy, i.e.,
$p\left(\tau\right)\propto\frac{I(\tau > 0 )} {1+\tau^{2}}$. The
prior for $p\left(\mu\right)\propto1$.
We reparametrize $\rho = \tanh(\xi)$ and put a flat prior on $\xi$.
We note that because of the large sample size, the results are insensitive
to these prior choices. 


We compare the performance of the following samplers:
(I)~\corrPMMHPG $(\left(\rho,\tau^{2};\mu,\phi\right)$,
(II)~the particle Gibbs with backward simulation approach of \citet{Whiteley2010} $\left(\textrm{PGBS}\left(\mu,\tau^{2},\rho,\phi\right)\right)$,
(III)~the particle Gibbs with data augmentation approach of
\citet{Fearnhead2016} $\left(\textrm{PGDA}\left(\mu,\tau^{2},\rho,\phi\right)\right)$, 
(IV)~the correlated pseudo-marginal sampler of \citet{Deligiannidis2018} $\left(\textrm{CPMMH}\left(\mu,\tau^{2},\rho,\phi\right)\right)$,
and (V)~the $\textrm{PHS}\left(\rho,\tau^{2};\mu,\phi\right)$ sampler of \citet{Mendes2020}.
The tuning parameters of the PGDA sampler are set optimally according to \citet{Fearnhead2016}. The correlation parameter $\rho_x$ in the \CPMMH{} is set to 0.999.
We apply these methods to a sample of daily US food
industry stock returns data obtained from the Kenneth French website,
using a sample from December 11th, 2001 to the 11th of November 2013,
a total of $3001$ observations, $\phi=0.98$, $\tau^{2}=0.10$, $\rho=-0.45$, and $\mu=-0.42$.


{\bf Empirical Results:}
We use $N=20,50,100$ particles for the \corrPMMHPG{} and CPMMH sampler, $N=200,500,1000$ particles for the PGBS sampler,
and $N=1000,2000,5000$ particles for the PGDA sampler and the 
PHS.
In this example, we use the 
bootstrap particle filter to sample the particles for all samplers and the adaptive random walk Metropolis-Hastings of \citet{robertsrosenthal2009}
for the MWG step in the \corrPMMHPG{}, the \PMMH{}
step in the \PHS{} sampler, and the \CPMMH{}  sampler as the 
proposal density for the parameters.
The particle filter and the parameter samplers
are implemented in Matlab. 
We ran the five sampling schemes for $15000$ iterations,
discarding the initial
$5000$ iterations as warmup.


Table \ref{tab:Inefficiency-factor-of univariate SV leverage example-T3000} shows
the IACT, TNV and RTNV estimates for the parameters and the latent volatilities in the univariate SV model with leverage estimated
using the five different samplers for the US food
industry stock returns data with $T=3000$ observations. 

\begin{sidewaystable}[H]
\caption{Univariate SV model with leverage for the five samplers. Sampler I:
$\textrm{CPHS}\left(\rho,\tau^{2};\mu,\phi\right)$, Sampler II: $\textrm{PGBS}\left(\mu,\tau^{2},\phi,\rho\right)$,
Sampler III: $\textrm{PGDA}\left(\mu,\tau^{2},\phi,\rho\right)$,
Sampler IV: $\textrm{ CPMMH}\left(\mu,\tau^{2},\phi,\rho\right)$,
and Sampler V: $\textrm{PHS}\left(\tau^{2},\rho;\mu,\phi\right)$
for US stock returns data with $T=3000$. Time is the time in seconds for one iteration of the algorithm. The RTNV is the TNV relative to the CPHS with $N=50$. The IACT, TNV, and RTNV are defined in section \ref{SS: preliminaries}.   \label{tab:Inefficiency-factor-of univariate SV leverage example-T3000}}

\centering{}%
\begin{tabular}{c|c|c|c|c|c|c|c|c|c|c|c|c|c|c|c|}
\hline 
{\footnotesize{}Param} & \multicolumn{3}{c|}{{\footnotesize{}I}} & \multicolumn{3}{c|}{{\footnotesize{}II}} & \multicolumn{3}{c|}{{\footnotesize{}III}} & \multicolumn{3}{c|}{{\footnotesize{}IV}} & \multicolumn{3}{c|}{{\footnotesize{}V}}\tabularnewline
\hline 
{\footnotesize{}$N$} & {\footnotesize{}20} & {\footnotesize{}50} & {\footnotesize{}100} & {\footnotesize{}200} & {\footnotesize{}500} & {\footnotesize{}1000} & {\footnotesize{}1000} & {\footnotesize{}2000} & {\footnotesize{}5000} & {\footnotesize{}20} & {\footnotesize{}50} & {\footnotesize{}100} & {\footnotesize{}1000} & {\footnotesize{}2000} & {\footnotesize{}5000}\tabularnewline
\hline 
{\footnotesize{}$\widehat{\textrm{IACT}}(\phi)$} & {\footnotesize{}$17.63$} & {\footnotesize{}$16.44$} & {\footnotesize{}$12.45$} & {\footnotesize{}177.48} & {\footnotesize{}191.48} & {\footnotesize{}100.60} & {\footnotesize{}155.84} & {\footnotesize{}61.34} & {\footnotesize{}43.18} & {\footnotesize{}114.64} & {\footnotesize{}22.81} & {\footnotesize{}18.67} & {\footnotesize{}18.81} & {\footnotesize{}12.95} & {\footnotesize{}10.27}\tabularnewline
{\footnotesize{}$\widehat{\textrm{IACT}}(\mu)$} & {\footnotesize{}$1.27$} & {\footnotesize{}$1.27$} & {\footnotesize{}$1.02$} & {\footnotesize{}1.26} & {\footnotesize{}1.08} & {\footnotesize{}1.20} & {\footnotesize{}91.75} & {\footnotesize{}42.88} & {\footnotesize{}38.29} & {\footnotesize{}109.06} & {\footnotesize{}31.27} & {\footnotesize{}46.93} & {\footnotesize{}1.06} & {\footnotesize{}1.09} & {\footnotesize{}1.14}\tabularnewline
{\footnotesize{}$\widehat{\textrm{IACT}}(\tau^{2})$} & {\footnotesize{}$34.08$} & {\footnotesize{}$24.68$} & {\footnotesize{}$22.18$} & {\footnotesize{}536.84} & {\footnotesize{}656.13} & {\footnotesize{}407.15} & {\footnotesize{}117.31} & {\footnotesize{}46.73} & {\footnotesize{}40.83} & {\footnotesize{}73.53} & {\footnotesize{}44.69} & {\footnotesize{}22.48} & {\footnotesize{}33.14} & {\footnotesize{}20.21} & {\footnotesize{}16.75}\tabularnewline
{\footnotesize{}$\widehat{\textrm{IACT}}_{\textrm{MAX}}(\rho)$} & {\footnotesize{}$23.56$} & {\footnotesize{}$14.68$} & {\footnotesize{}$12.03$} & {\footnotesize{}284.07} & {\footnotesize{}281.14} & {\footnotesize{}270.43} & {\footnotesize{}91.57} & {\footnotesize{}38.83} & {\footnotesize{}25.94} & {\footnotesize{}141.27} & {\footnotesize{}20.54} & {\footnotesize{}17.99} & {\footnotesize{}25.17} & {\footnotesize{}15.77} & {\footnotesize{}12.10}\tabularnewline
\hline 
{\footnotesize{}$\widehat{\textrm{IACT}}_{\textrm{MAX}}(\theta)$} & {\footnotesize{}$34.08$} & {\footnotesize{}$24.68$} & {\footnotesize{}$22.18$} & {\footnotesize{}$536.84$} & {\footnotesize{}656.13} & {\footnotesize{}407.15} & {\footnotesize{}155.84} & {\footnotesize{}61.34} & {\footnotesize{}43.18} & {\footnotesize{}$141.27$} & {\footnotesize{}$44.69$} & {\footnotesize{}$46.93$} & {\footnotesize{}33.14} & {\footnotesize{}20.21} & {\footnotesize{}16.75}\tabularnewline
{\footnotesize{}$\widehat{\textrm{TNV}}_{\textrm{MAX}}(\theta)$} & {\footnotesize{}$8.52$} & {\footnotesize{}$8.39$} & {\footnotesize{}$11.09$} & {\footnotesize{}$85.89$} & {\footnotesize{}$229.65$} & {\footnotesize{}$252.43$} & {\footnotesize{}$236.88$} & {\footnotesize{}$180.95$} & {\footnotesize{}$270.74$} & {\footnotesize{}$21.19$} & {\footnotesize{}$8.04$} & {\footnotesize{}$11.73$} & {\footnotesize{}$37.12$} & {\footnotesize{}$39.41$} & {\footnotesize{}$61.98$}\tabularnewline
{\footnotesize{}$\widehat{\textrm{RTNV}}_{\textrm{MAX}}(\theta)$} & {\footnotesize{}$1.02$} & {\footnotesize{}$1$} & {\footnotesize{}$1.32$} & {\footnotesize{}$10.20$} & {\footnotesize{}$27.37$} & {\footnotesize{}$30.09$} & {\footnotesize{}$28.23$} & {\footnotesize{}$21.57$} & {\footnotesize{}$32.27$} & {\footnotesize{}$2.53$} & {\footnotesize{}$0.96$} & {\footnotesize{}$1.40$} & {\footnotesize{}$4.42$} & {\footnotesize{}$4.70$} & {\footnotesize{}$7.39$}\tabularnewline
{\footnotesize{}$\widehat{\textrm{IACT}}_{\textrm{MEAN}}(\theta)$} & {\footnotesize{}$19.13$} & {\footnotesize{}$14.27$} & {\footnotesize{}$11.92$} & {\footnotesize{}249.91} & {\footnotesize{}282.46} & {\footnotesize{}194.85} & {\footnotesize{}114.12} & {\footnotesize{}47.45} & {\footnotesize{}37.06} & {\footnotesize{}$109.62$} & {\footnotesize{}$29.83$} & {\footnotesize{}$26.52$} & {\footnotesize{}19.55} & {\footnotesize{}12.50} & {\footnotesize{}10.07}\tabularnewline
{\footnotesize{}$\textrm{TNV}_{\textrm{MEAN}}(\theta)$} & {\footnotesize{}$4.78$} & {\footnotesize{}$4.85$} & {\footnotesize{}$5.96$} & {\footnotesize{}$39.99$} & {\footnotesize{}$98.86$} & {\footnotesize{}$120.81$} & {\footnotesize{}$173.46$} & {\footnotesize{}$139.98$} & {\footnotesize{}$232.37$} & {\footnotesize{}$16.44$} & {\footnotesize{}$5.37$} & {\footnotesize{}$6.63$} & {\footnotesize{}$21.90$} & {\footnotesize{}$24.38$} & {\footnotesize{}$37.26$}\tabularnewline
{\footnotesize{}$\textrm{RTNV}_{\textrm{MEAN}}(\theta)$} & {\footnotesize{}$0.99$} & {\footnotesize{}$1$} & {\footnotesize{}$1.23$} & {\footnotesize{}$8.25$} & {\footnotesize{}$20.38$} & {\footnotesize{}$24.91$} & {\footnotesize{}$35.76$} & {\footnotesize{}$28.86$} & {\footnotesize{}$47.91$} & {\footnotesize{}$3.39$} & {\footnotesize{}$1.11$} & {\footnotesize{}$1.37$} & {\footnotesize{}$4.52$} & {\footnotesize{}$5.03$} & {\footnotesize{}$7.68$}\tabularnewline
\hline 
{\footnotesize{}$\widehat{\textrm{IACT}}_{\textrm{MAX}}$$\left(x_{1:T}\right)$} & {\footnotesize{}$34.08$} & {\footnotesize{}$24.69$} & {\footnotesize{}$22.18$} & {\footnotesize{}17.08} & {\footnotesize{}20.93} & {\footnotesize{}11.80} & {\footnotesize{}58.99} & {\footnotesize{}25.90} & {\footnotesize{}15.24} & {\footnotesize{}419.10} & {\footnotesize{}334.76} & {\footnotesize{}387.20} & {\footnotesize{}4.42} & {\footnotesize{}2.92} & {\footnotesize{}2.43}\tabularnewline
{\footnotesize{}$\textrm{TNV}_{\textrm{MAX}}$$\left(x_{1:T}\right)$} & {\footnotesize{}8.52} & {\footnotesize{}8.39} & {\footnotesize{}11.09} & {\footnotesize{}2.73} & {\footnotesize{}7.33} & {\footnotesize{}7.32} & {\footnotesize{}89.66} & {\footnotesize{}76.41} & {\footnotesize{}95.55} & {\footnotesize{}62.87} & {\footnotesize{}60.26} & {\footnotesize{}96.80} & {\footnotesize{}4.95} & {\footnotesize{}5.69} & {\footnotesize{}8.99}\tabularnewline
{\footnotesize{}$\textrm{RTNV}_{\textrm{MAX}}$$\left(x_{1:T}\right)$} & {\footnotesize{}1.02} & {\footnotesize{}1} & {\footnotesize{}1.32} & {\footnotesize{}0.33} & {\footnotesize{}0.87} & {\footnotesize{}0.87} & {\footnotesize{}10.69} & {\footnotesize{}9.11} & {\footnotesize{}11.39} & {\footnotesize{}7.49} & {\footnotesize{}7.18} & {\footnotesize{}11.54} & {\footnotesize{}0.59} & {\footnotesize{}0.68} & {\footnotesize{}1.07}\tabularnewline
{\footnotesize{}$\widehat{\textrm{IACT}}_{\textrm{MEAN}}$$\left(x_{1:T}\right)$} & {\footnotesize{}3.91} & {\footnotesize{}$2.32$} & {\footnotesize{}$1.81$} & {\footnotesize{}1.91} & {\footnotesize{}1.87} & {\footnotesize{}1.57} & {\footnotesize{}14.08} & {\footnotesize{}7.92} & {\footnotesize{}5.93} & {\footnotesize{}83.75} & {\footnotesize{}30.29} & {\footnotesize{}15.83} & {\footnotesize{}1.22} & {\footnotesize{}1.13} & {\footnotesize{}1.12}\tabularnewline
{\footnotesize{}$\textrm{TNV}_{\textrm{MEAN}}$$\left(x_{1:T}\right)$} & {\footnotesize{}0.98} & {\footnotesize{}0.79} & {\footnotesize{}0.91} & {\footnotesize{}0.31} & {\footnotesize{}0.65} & {\footnotesize{}0.97} & {\footnotesize{}21.40} & {\footnotesize{}23.36} & {\footnotesize{}37.18} & {\footnotesize{}12.56} & {\footnotesize{}5.45} & {\footnotesize{}3.96} & {\footnotesize{}1.37} & {\footnotesize{}2.20} & {\footnotesize{}4.14}\tabularnewline
{\footnotesize{}$\textrm{RTNV}_{\textrm{MEAN}}$$\left(x_{1:T}\right)$} & {\footnotesize{}1.24} & {\footnotesize{}1} & {\footnotesize{}1.15} & {\footnotesize{}0.39} & {\footnotesize{}0.82} & {\footnotesize{}1.23} & {\footnotesize{}27.10} & {\footnotesize{}29.57} & {\footnotesize{}47.06} & {\footnotesize{}15.09} & {\footnotesize{}6.90} & {\footnotesize{}5.01} & {\footnotesize{}1.73} & {\footnotesize{}2.78} & {\footnotesize{}5.24}\tabularnewline
\hline 
{\footnotesize{}Time} & {\footnotesize{}$0.25$} & {\footnotesize{}$0.34$} & {\footnotesize{}$0.50$} & {\footnotesize{}$0.16$} & {\footnotesize{}$0.35$} & {\footnotesize{}$0.62$} & {\footnotesize{}$1.52$} & {\footnotesize{}$2.95$} & {\footnotesize{}$6.27$} & {\footnotesize{}$0.15$} & {\footnotesize{}$0.18$} & {\footnotesize{}$0.25$} & {\footnotesize{}$1.12$} & {\footnotesize{}$1.95$} & {\footnotesize{}$3.70$}\tabularnewline
\hline 
\end{tabular}
\end{sidewaystable}

The table shows that:
(1) The PGBS sampler has large IACT values for parameters, $\tau^{2}$ and $\rho$,
and that putting the two parameters in the MWG step of the \corrPMMHPG{}, and the PMMH step of the PHS
improves the mixing significantly. The  
\CPMMH{} sampler has similar IACT values for parameters $\tau^{2}$ and $\rho$ to \corrPMMHPG{}  and PHS. 
(2)~Increasing the number of particles in  PGBS does not improve the IACT of the parameters $\phi$, $\tau^2$, and $\rho$.
(3)~PGDA requires more than $N=1000$ particles to improve the IACT of the parameters $\phi$, $\tau^2$, and $\rho$ compared to PGBS.
Increasing the number of particles in PGDA improves the mixing of the Markov chains.
(4)~In terms of $\textrm{TNV}_{\textrm{MAX}}$, the \corrPMMHPG{} with only $N=50$ particles
is 10.20, 27.37, 30.09 times better than PGBS  with $N=200,500,1000$ particles,
respectively, is 28.23, 21.57, 32.27 times better than the PGDA sampler with $N=1000,2000,5000$ particles, and is 4.42, 4.70, and 7.39 times better
than \PHS{} with $N=1000,2000,5000$ particles for estimating the parameters of the univariate SV model with leverage. Similar conclusions hold for $\textrm{TNV}_{\textrm{MEAN}}$. 
(5)~The performance of the correlated PMMH sampler is comparable to the \corrPMMHPG{} for this model. In terms of $\textrm{TNV}_{\textrm{MEAN}}$, the \corrPMMHPG{} with $N=50$ particles is 3.39, 1.11, and 1.37 times better than the correlated PMMH sampler with $N=20,50,100$ particles. Similar conclusions hold for $\textrm{TNV}_{\textrm{MAX}}$.
(6)~In terms of $\textrm{TNV}_{\textrm{MAX}}$, the \corrPMMHPG{} with $N=50$ particles is 10.69, 9.11, and 11.39 times better than the PGDA sampler with $N=1000,2000,5000$ particles, and is 7.49, 7.18, and 11.54 better than the CPMMH sampler with $N=20,50,100$ particles for estimating the latent volatilities. The performance of the PHS sampler with $N=1000,2000$ particles and the PGBS sampler with $N=200,500,1000$ particles are slightly better than the \corrPMMHPG{} for estimating latent volatilities.
(7)~It is clear that the \corrPMMHPG{} performs well for estimating the parameters and the latent volatilities for the univariate SV model with leverage for estimating the latent volatilities. The performance of \corrPMMHPG{} is comparable to the CPMMH for estimating the parameters and is comparable to the \PGBS{} for estimating the latent volatilities. 
Section~\ref{additionalunivariateSV} of the online supplement gives further results for univariate stochastic volatility model with leverage with different number of observations $T=2000,4000,...,20000$ observations. The example shows that the \CPHS{} is much better than CPMMH for estimating both the model parameters and the latent volatilities.


We now consider the univariate SV model with leverage and a large
number of covariates. The measurement equation is 
\begin{equation}
y_{t}=w_{t}^{\top}\beta+\exp\left(x_{t}/2\right)\epsilon_{t},\;\textrm{where}\;\;\epsilon_{t}\sim\N\left(0,1\right).
\end{equation}
We compare the performance of the following samplers: 
(I) $\textrm{\corrPMMHPG}\left(\rho,\tau^{2};\mu,\phi,\beta\right)$,
(II) $\textrm{PGBS}\left(\mu,\tau^{2},\rho,\phi,\beta\right)$,
(III) $\textrm{CPMMH}\left(\mu,\tau^{2},\rho,\phi,\beta\right)$. We apply the methods to simulated data with $T=6000$ observations, $\phi=0.98$, $\tau^2=0.1$, $\rho=-0.2$, $\mu=-0.42$, $\beta_k$ is generated from normal distribution with mean zero and standard deviation 0.1, for all $k=1,...,50$. The covariates generated randomly from $\N(0,1)$. The prior for $\beta_k$ is $\N(0,1)$ for $k=1,...,50$. 

\begin{table}[H]
\caption{Univariate SV model with leverage and 50 covariates. Sampler I: $\textrm{CPHS}\left(\rho,\tau^{2};\mu,\phi,\beta\right)$,
Sampler II: $\textrm{PGBS}\left(\mu,\tau^{2},\phi,\rho,\beta\right)$, Sampler
III: $\textrm{CPMMH}\left(\mu,\tau^{2},\phi,\rho,\beta\right)$ for simulated
data with $T=6000$. Time is the time in seconds for one iteration of the algorithm. The RTNV is the TNV relative to the CPHS with $N=100$. The IACT, TNV, and RTNV are defined in section \ref{SS: preliminaries}. \label{UnivariateSVModelwithCovariates}}

\centering{}%
\begin{tabular}{c|c|c|c|c|c|}
\hline 
Param & \multicolumn{1}{c|}{I} & \multicolumn{1}{c|}{II} & \multicolumn{3}{c|}{III}\tabularnewline
\hline 
$N$ & 100 & 1000 & 100 & 500 & 1000\tabularnewline
\hline 
$\widehat{\textrm{IACT}}(\phi)$ & 5.20 & 39.45 & 223.78 & 99.50 & 149.60\tabularnewline
$\widehat{\textrm{IACT}}(\mu)$ & 1.07 & 1.00 & 139.16 & 149.02 & 99.30\tabularnewline
$\widehat{\textrm{IACT}}(\tau^{2})$ & 12.84 & 179.30 & 97.15 & 134.03 & 99.78\tabularnewline
$\widehat{\textrm{IACT}}(\rho)$ & 11.63 & 102.60 & 133.71 & 97.44 & 102.45\tabularnewline
$\widehat{\textrm{IACT}}_{\textrm{MEAN}}\left(\beta\right)$ & 1.56 & 1.54 & 98.13 & 88.97 & 85.19\tabularnewline
$\widehat{\textrm{IACT}}_{\textrm{MAX}}\left(\beta\right)$ & 1.78 & 1.83 & 146.72 & 121.53 & 111.90\tabularnewline
\hline 
$\widehat{\textrm{IACT}}_{\textrm{MAX}}(\theta)$ & 12.84 & 179.30 & 223.78 & 149.02 & 149.60\tabularnewline
$\widehat{\textrm{TNV}}_{\textrm{MAX}}(\theta)$ & 14.12 & 227.71 & 114.13 & 269.73 & 477.22\tabularnewline
$\widehat{\textrm{RTNV}}_{\textrm{MAX}}(\theta)$ & 1 & 16.13 & 8.08 & 19.10 & 33.80\tabularnewline
$\widehat{\textrm{IACT}}_{\textrm{MEAN}}(\theta)$ & 2.01 & 7.40 & 101.85 & 91.27 & 87.23\tabularnewline
$\textrm{TNV}_{\textrm{MEAN}}(\theta)$ & 2.21 & 9.40 & 51.94 & 165.20 & 278.26\tabularnewline
$\textrm{RTNV}_{\textrm{MEAN}}(\theta)$ & 1 & 4.25 & 23.50 & 74.75 & 125.91\tabularnewline
\hline 
$\widehat{\textrm{IACT}}_{\textrm{MAX}}$$\left(x_{1:T}\right)$ & 4.21 & 2.17 & 127.76 & 68.63 & 25.73\tabularnewline
$\textrm{TNV}_{\textrm{MAX}}$$\left(x_{1:T}\right)$ & 4.63 & 2.76 & 65.16 & 124.22 & 82.08\tabularnewline
$\textrm{RTNV}_{\textrm{MAX}}$$\left(x_{1:T}\right)$ & 1 & 0.60 & 14.07 & 26.83 & 17.73\tabularnewline
$\widehat{\textrm{IACT}}_{\textrm{MEAN}}$$\left(x_{1:T}\right)$ & 1.21 & 1.04 & 10.83 & 7.70 & 7.37\tabularnewline
$\textrm{TNV}_{\textrm{MEAN}}$$\left(x_{1:T}\right)$ & 1.33 & 1.32 & 5.52 & 13.94 & 23.51\tabularnewline
$\textrm{RTNV}_{\textrm{MEAN}}$$\left(x_{1:T}\right)$ & 1 & 0.99 & 4.15 & 10.48 & 17.68\tabularnewline
\hline 
Time & 1.10 & 1.27 & 0.51 & 1.81 & 3.19\tabularnewline
\hline 
\end{tabular}
\end{table}

Table \ref{UnivariateSVModelwithCovariates} reports the estimated IACT, TNV, and RTNV values for the parameters and the latent volatilities in the univariate SV model with leverage and covariates estimated using $3$ different samplers. The table suggests the following: (1) The \corrPMMHPG{} is much better than the CPMMH and the PGBS samplers for estimating the parameters in the univariate SV model. This example shows very clearly how the flexibility of \corrPMMHPG{} can be used to obtain good results. The vector of parameters $\beta$ are high-dimensional and not highly correlated with the states, so it is important to generate them in a PG step. Both $\tau^2$ and $\rho$ are generated in an MWG step because they are highly correlated with the states. In general, it is useful to generate the parameters that are highly correlated with the states using a MWG step. If there is a subset of parameters that is not highly correlated with the states, then it is better to generate them using a PG step that conditions on the states, especially when the number of parameters is large. In addition, using PG{} is preferable  in general to MWG because it is easier to obtain better proposals within a PG framework. (2) The CPMMH performs much worse than \corrPMMHPG{} because the adaptive random walk is inefficient in high dimensions. 




\section{Multivariate example \label{Multivariate example}}
This section applies the \corrPMMHPG{} to a model having a large number of observations,
a large number of parameters, and a large number of latent states.
Section~\ref{S: factor SV model} discusses the multivariate factor stochastic volatility model with leverage.
Section~\ref{results multivariate examples} compares the performance of the \corrPMMHPG{} to competing PMCMC methods
to estimate a multivariate factor stochastic volatility model with leverage using real datasets having a medium and large number of observations.

\subsection{The factor stochastic volatility model\label{S: factor SV model}}
The factor SV model is often used to parsimoniously model a vector of returns; 
see, e.g., \cite{Chib2006, Kastner:2017,Kastner2019}. However, estimating time-varying multivariate factor SV models 
can be very
challenging because the likelihood involves computing an integral over a very high dimensional latent state and the number of parameters in the model can be very large.
We use this complex model to illustrate the \corrPMMHPG{}.

Suppose that $P_{t}$ is a $S\times1$ vector of daily
stock prices and define $y_{t}\coloneqq\log P_{t}-\log P_{t-1}$
as the log-return of the stocks in period $t$. 
We model $ y_{t}$ as the factor SV model
\begin{align}
y_{t}= \beta f_{t}+V_{t}^{\frac{1}{2}} \epsilon_{t},\:\left(t=1,...,T\right),\label{eq:factor model}
\end{align}
where $ f_{t}$ is a $K\times1$ vector of latent factors
(with $K \ll S$), $ \beta$ is a $S\times K$ factor
loading matrix of the unknown parameters.
Section~\ref{S: sampling factor loading matrix}  of the supplement discusses parametrization and identification issues for the factor loading matrix $ \beta$
and the latent factors $ f_t$.

We assume that:  the $ \epsilon_{t}\sim N\left(0,I\right)$; the idiosyncratic error volatility matrix $ V_t$ is diagonal,
with diagonal elements $ \exp(h_{st})$; the log volatility processes $\{h_{st},  t \geq 1\} $ are independent for
$s=1, \dots, S$, and each follows an independent autoregressive process of the form
\begin{equation} \label{eq: h eqn for factor sv}
\begin{aligned}
h_{s1} & \sim N\begin{pmatrix} \mu_{\epsilon s} , \nicefrac{\tau^2_{\epsilon s}}{1 - \phi^2_{\epsilon s}}  \end{pmatrix},  \\
h_{s,t+1} & = 
\mu_{\epsilon s} + \phi_{\epsilon s} \(h_{st} -\mu_{\epsilon s} \)  + \eta_{\epsilon st}, \\
{\rm with} \quad
\begin{pmatrix} \epsilon_{st} \\ \eta_{\epsilon st} \end{pmatrix} & \sim  N \begin{pmatrix} \begin{pmatrix} 0 \\ 0 \end{pmatrix}
\begin{pmatrix} 1 & \rho_{\epsilon s} \tau_{\epsilon s}\\ \rho_{\epsilon s} \tau_{\epsilon s} & \tau^2_{\epsilon s}
\end{pmatrix}
\end{pmatrix}.
\end{aligned}
\end{equation}


We also consider that the log-volatilities $h_{s,t}$ follow a GARCH
diffusion continuous time volatility process, which 
does not have a closed form transition
density.  The continuous time GARCH diffusion process $\left\{ h_{s,t}\right\} _{t\geq1}$
is \citep{Kleppe2010,Chib2004} 
\begin{equation}
dh_{s,t}=\left\{ \alpha_{\epsilon,s}\left(\mu_{\epsilon,s}-\exp\left(h_{s,t,j}\right)\right)\exp\left(-h_{s,t,j}\right)-\frac{\tau_{\epsilon,s}^{2}}{2}\right\} dt+\tau_{\epsilon,s}dW_{s,t},
\end{equation}
where the $W_{s,t}$ are independent Wiener processes. We use the GARCH diffusion model to show that the CPHS can be used to estimate state space models that do not have closed form transition densities. In this example, the CPHS is applied to an Euler approximation
of the diffusion process driving the log volatilities.
The Euler scheme places
$M-1$ evenly spaced points between times $t$ and $t+1$. We denote
the intermediate volatility components by $h_{s,t,1},...,h_{s,t,M-1}$
and set $h_{s,t,0}=h_{s,t}$ and $h_{s,t,M}=h_{s,t+1}$. The equation
for the Euler evolution, starting at $h_{s,t,0}$ is 
\begin{equation}
h_{s,t,j+1}|h_{s,t,j}\sim N\left(h_{s,t,j}+\left\{ \alpha_{\epsilon,s}\left(\mu_{\epsilon,s}-\exp\left(h_{s,t,j}\right)\right)\exp\left(-h_{s,t,j}\right)-\frac{\tau_{\epsilon,s}^{2}}{2}\right\} \delta,\tau_{\epsilon,s}^{2}\delta\right),
\end{equation}
for $j=0,...,M-1$, where $\delta=1/M$.

The factors $f_{kt}, k=1, \dots, K$ are assumed independent with $ f_t \sim N (  f_t; 0 ,  D_t) $; $  D_t$ is a diagonal matrix
with $k$th diagonal element $\exp(\lambda_{kt})$. Each  log volatility $\lambda_{kt}, k=1, \dots, K$ is assumed to follow an independent autoregressive process of the form
\begin{equation} \label{eq: lambda eqn for factor sv}
\begin{aligned}
\lambda_{k1} & \sim N\begin{pmatrix}   0 , \frac{\tau^2_{f k}}{1 - \phi^2_{f k}}  \end{pmatrix}, \quad
\lambda_{k,t+1}  = \phi_{f k } \lambda_{k t }   +  \eta_{f k t}, \quad  \eta_{f k t} \sim N\begin{pmatrix}
 0,\tau^2_{fk}\end{pmatrix} \quad (t\geq 1). \\
\end{aligned}
\end{equation}

{\bf Prior specification} 
For $s=1,...,S$ and $k=1,...,K$, we
choose the priors for the
persistence parameters  $\phi_{\epsilon s} $ and $\phi_{fk}$, the priors for $\tau_{\epsilon s}, \tau_{f k}, \mu_{\epsilon s}$ and
$\rho_{\epsilon s}$  as in section~\ref{Univariate Empirical Results}. For every unrestricted element of the factor loadings
matrix $\beta$, we follow \citet{Kastner:2017} and choose standard normal distributions. The priors for the GARCH parameters are $\alpha_{\epsilon,s}\sim IG(v_0/2,s_0/2)$, $\tau_{\epsilon,s} \sim IG(v_0/2,s_0/2)$, where $v_0=10$ and $s_0=1$, $p(\mu_{\epsilon,s}) \propto 1$ for $s=1,...,S$. 
These prior densities cover most possible values in practice.

Although the multivariate factor SV model can be written in standard state space form as in section \ref{sub:State-Space-Models},
it is more efficient to take advantage its conditional independence structure and develop a sampling scheme
on multiple independent univariate state space models.

{\bf Conditional Independence and sampling in the factor SV model:} 
The key to making the estimation of the factor SV model tractable
is that given the values of $\left( y_{1:T}, f_{1:T},\beta \right)$,
the factor model given in \eqref{eq:factor model} separates
into $S+K$ independent components consisting of $K$ univariate SV
models for the latent factors and $S$ univariate SV models with (or
without) leverage for the idiosyncratic errors. That is, given $\left(y_{1:T}, f_{1:T}, \beta \right)$,
we have $S$  univariate SV models with leverage, with $\epsilon_{s t}$ the $t$th \lq observation\rq{} on the $s$th SV model, and we have $K$ univariate SV models without
leverage, with $f_{k t}$ the $t$th observation on the $k$th univariate SV model.
Section~\ref{subsec:Correlated-PMMH+PG-sampling} of the supplement discusses the \corrPMMHPG{} for the factor SV model and makes full use of its conditional independence structure.
Section~\ref{S: deep interweaving}
of the supplement also discusses a deep interweaving strategy for the loading matrix $ \beta$ and the factor $ f_t$ that helps the sampler mix better.

{\bf Target density:}
Section~\ref{subsec:Target-Distributions-for factor SV}  of the supplement gives the target distribution for the factor SV model, which is
a composite of the target densities
of the univariate SV models for the idiosyncratic errors and the factors, together with densities for $\bs \beta$ and $\bs f_{1:T}$.

\subsection{Empirical Study \label{results multivariate examples}}
This section presents empirical results for the multivariate factor SV model with leverage model described in section \ref{S: factor SV model}.

{\em Estimation details:} 
We applied our method to a sample of daily US industry stock returns
data. The data, obtained from the website of Kenneth French, consists
of daily returns for $S=26$ value weighted US industry portfolios. We
use a sample from November 19th, 2009 to 11th of November, 2013 a
total of 1000 observations and a sample from December 11th, 2001 to
the 11th of November, 2013, a total of 3001 observations. The computation
is done in Matlab and is run on 28 CPU-cores of GADI high-performance
computer cluster at the National Computing Infrastructure Australia\footnote{https://nci.org.au/}. 

We do not compare our method to the correlated pseudo-marginal method of \citet{Deligiannidis2018}
which generates the parameters, with the factors and idiosyncratic latent
log volatilities  integrated out for two reasons.
First, using a pseudo-marginal method results in a $S+K= 30 $ dimensional state space model
and  \citet{Mendes2020} show that it is
very hard to preserve the correlation between the logs of the estimated likelihoods at the current and proposed
values for such a model. Thus the correlated pseudo-marginal method
would get stuck unless enough particles are used to ensure that the
variance of log of the estimated  likelihood is close to 1.
\citet{Deligiannidis2018} also discuss the issue of how the correlation between the log of the estimated likelihoods at the current and proposed values decreases as the dimension increases.
Second, the dimension of the parameter space
in the factor SV model is large which makes it difficult to implement the pseudo-marginal sampler
efficiently as it is difficult to obtain good proposals for the parameters because
the first and second derivatives of the likelihood with respect to the parameters
can only be estimated, while the random walk proposal is easy to implement but is very inefficient
in high dimensions.

\citet{Kastner:2017} develop the Gibbs-type MCMC algorithm to estimate multivariate factor stochastic volatility model. They use the approach proposed by \citet{Kim:1998} to approximate the distribution of innovations in the log outcomes by a mixture of normals. Their estimator is not simulation consistent because it does not correct for these approximations. In addition, the Gibbs-type MCMC sampler cannot be used to estimate the factor SV model with the log volatility following diffusion processes with intractable state transition density. It is well known that the Gibbs sampler is inefficient for generating parameters of a diffusion model, in particular the variance parameter \citep{Stramer2011}.



We report results for a modified version of the PGDA sampler applied to the multivariate factor stochastic volatility model as 
\citet{Mendes2020} shows that the PGDA sampler of \citet{Fearnhead2016} does not work well when the model has many parameters
because the PGDA updates the pseudo observations of the parameters by MCMC and updates all of the latent states and parameters jointly using the particle filter.
We now extend the PGDA method to estimate the multivariate factor SV model and call it the refined PGDA method.
The sampler first generates the factor loading matrix and latent factors by a PG step and then,
conditioning on the latent factors and the factor loading matrix, we obtain $S$ univariate SV models with leverage and $K$ univariate SV models.
Then, for each of the univariate models, we can apply the \citet{Fearnhead2016} approach by updating the pseudo observation of the SV parameters by MCMC
and updating the parameters and the latent states jointly by particle filter. The tuning parameters of the PGDA sampler are set optimally according to \citet{Fearnhead2016}.
We compare the following samplers: 
(I) the $\textrm{\corrPMMHPG}\left({\rho}_{\epsilon},{\tau}_{\epsilon}^{2},{\tau}_{f}^{2};{f}_{1:T},{\beta},{\mu}_{\epsilon},\phi_{\epsilon},\phi_{f}\right)$,
(II) the $\textrm{PGBS}\left({f}_{1:T},{\beta},{\rho}_{\epsilon},{\tau}_{\epsilon}^{2},{\mu}_{\epsilon},\phi_{\epsilon},\phi_{f},{\tau}_{f}^{2}\right)$,
(III)~the refined $\textrm{PGDA}\left({f}_{1:T},{\beta},{\rho}_{\epsilon},{\tau}_{\epsilon}^{2},{\mu}_{\epsilon},\phi_{\epsilon},\phi_{f},{\tau}_{f}^{2}\right)$,
and\newline 
(IV)~the $\textrm{PHS}\left({\rho}_{\epsilon},{\tau}_{\epsilon}^{2},{\tau}_{f}^{2};{f}_{1:T},{\beta},{\mu}_{\epsilon},\phi_{\epsilon},\phi_{f}\right)$.


{\em Empirical Results}
Tables \ref{tab:IACT T=00003D1000 N=00003D100} to \ref{tab:Inefficiency-factor-(IACT) T1000 N1000} in section \ref{S: further empirical results for factor SV model} of the online supplement
show the mean and maximum IACT values for each parameter in
the factor SV model with leverage for $T=1000$ observations, $S=26$ stock returns, and $K=1$ factor.
We compare the \corrPMMHPG{} with
$N=100$ with the PGBS, the refined PGDA, and the PHS
samplers with $N\in\left\{ 100,250,500,1000\right\} $. The tables show that the factor
loading matrix $\beta$ is sampled efficiently
by all samplers with comparable IACT values.
However, the \corrPMMHPG{} has much smaller IACT values than the PGBS sampler
for the $\tau^{2}$, $\phi$, and $\rho$ parameters for all $N$ (number of particles).
The \corrPMMHPG{} has much smaller IACT values for the $\tau^{2}$, $\phi$, and $\rho$ parameters than the
\PHS{}  and the refined PGDA samplers when $N=100$ and $250$ and has comparable IACT values  for the $\tau^{2}$, $\phi$, and $\rho$ parameters when $N=500$
and $1000$.

Table \ref{tab:Comparing-Sampler-I: Empirical Results1} summarises
the estimation result for a factor SV model with leverage with $T=1000$ observations,
$S=26$ stock returns, and $K=1$ factor.  The table shows that: (1) Increasing the number of particles
from $N=100$ to $1000$ does not seem to improve the performance
of the PGBS sampler. This indicates that for the parameters that are
highly correlated with the states, such as the $\tau_{\epsilon}^{2}$,
increasing the number of particles does not improve the mixing of
the parameters. (2) Increasing the number of particles from $N=100$
to $1000$ improves the performance of the PHS and the refined PGDA samplers significantly.
They require at least  $N=1000$ particles to obtain comparable values of $\widehat{\textrm{IACT}}_{\textrm{MEAN}}$
and $\widehat{\textrm{IACT}}_{\textrm{MAX}}$ as the \corrPMMHPG{}.
(3) The \corrPMMHPG{} with $N=100$ particles is much more efficient than the PGBS sampler in terms of both
$\widehat{\textrm{RTNV}}_{\textrm{MEAN}}$ and $\widehat{\textrm{RTNV}}_{\textrm{MAX}}$ for all number of particles $N$.
(4) The \corrPMMHPG{} is much more efficient than the PHS and the refined PGDA samplers for $N=100$ and $250$,
but only slightly more efficient for $N=500$ and $1000$.

We now compare the performance of the \corrPMMHPG{} with
$N=100$ when the number of observations $T=3000$ is large,
with PGBS, the refined PGDA, and the \PHS{}  with $N\in\left\{ 500,1000,2000\right\} $. 
Tables~\ref{tab:Inefficiency-factor-(IACT) N500 T3000} to \ref{tab:Inefficiency-factor-(IACT) N2000 T3001} in section~\ref{S: further empirical results for factor SV model} of the online supplement
show the mean and maximum IACT values for each parameter in
the factor SV model with leverage for $T=3000$ observations, $S=26$ stock returns, $K=4$ factors,
for the \corrPMMHPG{}, the 
\PHS{}, the PGBS and the refined PGDA samplers. 
Similarly to the $T=1000$ case,
the tables show that the factor loading matrix $\beta$ is sampled efficiently by all samplers, with comparable IACT
values. The \corrPMMHPG{} has much smaller IACT values
than the PGBS sampler, the refined PGDA, and the \PHS{} sampler for the $\tau^{2}$, $\phi$,
and $\rho$ parameters for all cases. Note that even with $N=2000$,
the \PHS{}, PGBS, and refined PGDA samplers do not give comparable performance to the \corrPMMHPG{}. 
Table~\ref{tab:Comparing-Sampler-I:Empirical results2} summarises the estimation result for a factor SV model with leverage with $T=3000$ observations,
$S=26$ stock returns, and $K=4$ factors. The table shows that:
(1)~It is necessary to greatly increase the number of 
particles  for both the \PHS{} and refined PGDA samplers to avoid the Markov chains from getting stuck.
It is clear, however, that the \corrPMMHPG{} works with only $N=100$ particles.
This suggests that the \corrPMMHPG{}  scales up better 
than the other two samplers in both the number of observations and in the number of parameters.
(2)~The table clearly shows that the \corrPMMHPG{} with $N=100$ particles is much more efficient than the PHS, the PGBS,
and the refined PGDA samplers with $N=500, 1000, 2000$ particles in terms of both $\widehat{\textrm{RTNV}}_{\textrm{MEAN}}$ and $\widehat{\textrm{RTNV}}_{\textrm{MAX}}$.


\subsection*{GARCH Diffusion Models}
This section considers the factor stochastic volatility models described
in section \ref{Multivariate example}, where the idiosyncratic log-volatilities follow a GARCH
diffusion continuous volatility time process which do not have closed form state transition densities. The following samplers are compared: 
(I)~The CPHS $\left(\tau_{\epsilon}^{2},\tau_{f}^{2},\mu_{\epsilon},\alpha_{\epsilon};\phi_{f},f_{1:T},\beta\right)$,
(II)~the $\textrm{PG}\left(\tau_{\epsilon}^{2},\tau_{f}^{2},\mu_{\epsilon},\alpha_{\epsilon},\phi_{f},f_{1:T},\beta\right)$, 
and (III) the \PHS $( \tau_{\epsilon}^{2},\tau_{f}^{2},\mu_{\epsilon},\alpha_{\epsilon};   \phi_{f},f_{1:T},\beta)$.
We compare the PHS with $N=100$ with the PG and the PHS
with $N\in\left\{ 500,1000\right\} $. 

Tables \ref{tab:Inefficiency-factor-(IACT) N500 T3000-diffusion} 
and \ref{tab:Inefficiency-factor-(IACT) N1000 T3000-diffusion-1}
in section~\ref{S: further empirical results for factor SV 
model} of the online supplement report the IACT estimates for 
all the parameters for the factor SV model with the 
idiosyncratic log-volatilities following GARCH diffusion models.
The table clearly shows that the PG sampler has large IACT 
values for the GARCH parameters for both $N=500$ and $1000$ 
particles. Putting those three parameters of the GARCH diffusion model in the PMMH step of the PHS sampler with $N=1000$ particles and in the MWG step of the CPHS sampler with $N=100$ particles  improves the mixing significantly. 
Table~\ref{tab:Comparing-Sampler-I:Diffusion} summarises the estimation results and shows that in terms of $\textrm{TNV}_{\textrm{MAX}}$, the CPHS is $93.90$ and $289.10$ times better than the PG sampler with $N=500$ and $1000$, respectively, and in terms of  $\textrm{TNV}_{\textrm{MEAN}}$, the CPHS is $24.25$ and $61.19$ times better than the PG sampler with $N=500$ and $1000$, respectively. The table also shows that in terms of $\textrm{TNV}_{\textrm{MAX}}$, the CPHS is $8.35$ and $12.07$ times better than the PHS sampler with $N=500$ and $1000$, respectively, and in terms of  $\textrm{TNV}_{\textrm{MEAN}}$, the CPHS is $4.31$ and $7.64$ times better than the PHS sampler with $N=500$ and $1000$, respectively.

\begin{sidewaystable}[H]
\caption{Factor SV model with leverage for US stock return data with ${T=1000}$ observations, $S=26$ stock returns, and $K=1$ factors. Comparing Sampler I: $\textrm{\corrPMMHPG}\left({\rho}_{\epsilon},{\tau}_{\epsilon}^{2},{\tau}_{f}^{2};{f}_{1:T},{\beta},{\mu}_{\epsilon},\phi_{\epsilon},\phi_{f}\right)$;
Sampler II: $\textrm{PHS \ensuremath{\left(\tau_{f}^{2},\tau_{\epsilon}^{2},\rho_{\epsilon};{f}_{1:T},\mu_{\epsilon},\phi_{\epsilon},\phi_{f},\beta\right)}}$;
Sampler III: $\textrm{PGBS}\left({f}_{1:T},\rho_{\epsilon},\mu_{\epsilon},\phi_{\epsilon},\phi_{f},\beta,\tau_{f}^{2},\tau_{\epsilon}^{2}\right)$; Sampler IV: $\textrm{PGDA}\left({f}_{1:T},\rho_{\epsilon},\mu_{\epsilon},\phi_{\epsilon},\phi_{f},\beta,\tau_{f}^{2},\tau_{\epsilon}^{2}\right)$
in terms of Time Normalised Variance. Time
denotes the time taken in seconds per iteration of the method. The
table shows the $\widehat{IACT}_{MAX}$, $\widehat{TNV}_{MAX}$, ...,
$\widehat{RTNV}_{MEAN}$. The IACT, TNV, and RTNV are defined in section \ref{SS: preliminaries}. \label{tab:Comparing-Sampler-I: Empirical Results1}}
 
\centering{}%
\begin{tabular}{|c|c|cccc|cccc|cccc|}
\hline 
 & I  & \multicolumn{4}{c|}{II} & \multicolumn{4}{c|}{III} & \multicolumn{4}{c|}{IV}\tabularnewline
\hline 
$N$ & 100 & 100 & 250 & 500 & 1000 & 100 & 250 & 500 & 1000 & 100 & 250 & 500 & 1000\tabularnewline
$\widehat{\textrm{IACT}}_{\textrm{MAX}}(\theta)$ & 69.00 & 4454.61 & 493.83 & 88.33 & 75.54 & 1966.36 & 2158.37 & 2016.98 & 1781.30 & 6052.27 & 1111.59 & 192.41 & 98.03\tabularnewline
$\widehat{\textrm{TNV}}_{\textrm{MAX}}(\theta)$ & 55.20 & 2316.40 & 340.74 & 83.03 & 109.53 & 845.53 & 1122.35 & 1311.04 & 1799.11 & 2602.48 & 578.03 & 126.99 & 99.99\tabularnewline
$\widehat{\textrm{RTNV}}_{\textrm{MAX}}(\theta)$ & 1 & 41.96 & 6.17 & 1.50 & 1.98 & 15.32 & 20.13 & 23.75 & 32.59 & 47.15 & 10.47 & 2.30 & 1.81\tabularnewline
$\widehat{\textrm{IACT}}_{\textrm{MEAN}}(\theta)$ & 19.47 & 140.20 & 32.78 & 22.32 & 19.98 & 237.77 & 259.93 & 237.88 & 245.48 & 545.18 & 81.01 & 29.25 & 17.78\tabularnewline
$\widehat{\textrm{TNV}}_{\textrm{MEAN}}(\theta)$ & 15.58 & 72.90 & 22.62 & 20.98 & 28.97 & 102.24 & 135.16 & 154.62 & 247.93 & 234.43 & 42.13 & 19.30 & 18.14\tabularnewline
$\widehat{\textrm{RTNV}}_{\textrm{MEAN}}(\theta)$ & 1 & 4.68 & 1.45 & 1.35 & 1.86 & 6.56 & 8.68 & 9.92 & 15.91 & 15.05 & 2.70 & 1.24 & 1.16\tabularnewline
Time & 0.80 & 0.52 & 0.69 & 0.94 & 1.45 & 0.43 & 0.52 & 0.65 & 1.01 & 0.43 & 0.52 & 0.66 & 1.02\tabularnewline
\hline 
\end{tabular}
\end{sidewaystable}

\begin{sidewaystable}[H]
\caption{Factor SV model with leverage for US stock return data with ${T=3001}$ observations, $S=26$ stock returns, and $K=4$ factors. Comparing Sampler I: $\textrm{\corrPMMHPG}\left({\rho}_{\epsilon},{\tau}_{\epsilon}^{2},{\tau}_{f}^{2};{f}_{1:T},{\beta},{\mu}_{\epsilon},\phi_{\epsilon},\phi_{f}\right)$;
Sampler II: $\textrm{PHS \ensuremath{\left(\tau_{f}^{2},\tau_{\epsilon}^{2},\rho_{\epsilon};{f}_{1:T},\mu_{\epsilon},\phi_{\epsilon},\phi_{f},\beta\right)}}$;
Sampler III: $\textrm{PGBS}\left({f}_{1:T},\rho_{\epsilon},\mu_{\epsilon},\phi_{\epsilon},\phi_{f},\beta,\tau_{f}^{2},\tau_{\epsilon}^{2}\right)$; Sampler IV: $\textrm{PGDA}\left({f}_{1:T},\rho_{\epsilon},\mu_{\epsilon},\phi_{\epsilon},\phi_{f},\beta,\tau_{f}^{2},\tau_{\epsilon}^{2}\right)$
in terms of Time Normalised Variance. Time
denotes the time taken in seconds per iteration of the method. The
table shows the $\widehat{IACT}_{MAX}$, $\widehat{TNV}_{MAX}$, ...,
$\widehat{RTNV}_{MEAN}$. The IACT, TNV, and RTNV are defined in section \ref{SS: preliminaries}. \label{tab:Comparing-Sampler-I:Empirical results2}}

\centering{}%
\begin{tabular}{|c|c|ccc|ccc|ccc|}
\hline 
 & I  & \multicolumn{3}{c|}{II} & \multicolumn{3}{c|}{III} & \multicolumn{3}{c|}{IV}\tabularnewline
\hline 
$N$ & 100 & 500 & 1000 & 2000 & 500 & 1000 & 2000 & 500 & 1000 & 2000\tabularnewline
$\widehat{\textrm{IACT}}_{\textrm{MAX}}(\theta)$ & 119.37 & 5003.03 & 2467.29 & 908.42 & 2517.54 & 2233.80 & 1914.02 & 4815.88 & 2269.26 & 1379.56\tabularnewline
$\widehat{\textrm{TNV}}_{\textrm{MAX}}(\theta)$ & 358.11 & 14808.97 & 11349.53 & 7721.57 & 7049.11 & 6768.41 & 9321.28 & 13725.26 & 7261.63 & 7008.06\tabularnewline
$\widehat{\textrm{RTNV}}_{\textrm{MAX}}(\theta)$ & 1 & 41.35 & 31.69 & 21.56 & 19.68 & 18.90 & 26.03 & 38.33 & 20.28 & 19.57\tabularnewline
$\widehat{\textrm{IACT}}_{\textrm{MEAN}}(\theta)$ & 24.02 & 138.21 & 85.71 & 41.14 & 217.58 & 215.24 & 247.47 & 395.05 & 140.35 & 72.56\tabularnewline
$\widehat{\textrm{TNV}}_{\textrm{MEAN}}(\theta)$ & 72.06 & 409.10 & 394.27 & 349.69 & 609.22 & 652.18 & 1205.18 & 1125.89 & 449.12 & 368.60\tabularnewline
$\widehat{\textrm{RTNV}}_{\textrm{MEAN}}(\theta)$ & 1 & 5.68 & 5.47 & 4.85 & 8.45 & 9.05 & 16.72 & 15.62 & 6.23 & 5.12\tabularnewline
Time & 3.00 & 2.96 & 4.60 & 8.50 & 2.80 & 3.03 & 4.87 & 2.85 & 3.20 & 5.08\tabularnewline
\hline 
\end{tabular}
\end{sidewaystable}

\begin{table}
\caption{Factor SV model with GARCH diffusion processes for the idiosyncratic
volatility for US stock return data with $T=3001$ observations, $S=26$
stocks, and $K=4$ factors. Comparing Sampler I: $\textrm{CPHS \ensuremath{\left(\tau_{\epsilon}^{2},\tau_{f}^{2},\mu_{\epsilon},\alpha_{\epsilon};\phi_{f},f_{1:T},\beta\right)}}$,
Sampler II: $\textrm{\textrm{PG}\ensuremath{\left(\tau_{\epsilon}^{2},\tau_{f}^{2},\mu_{\epsilon},\alpha_{\epsilon},\phi_{f},f_{1:T},\beta\right)}}$,
and Sampler III: the $\textrm{PHS}\left(\tau_{\epsilon}^{2},\tau_{f}^{2},\mu_{\epsilon},\alpha_{\epsilon};\phi_{f},f_{1:T},\beta\right)$
in terms of Time Normalised Variance. Time denotes the time taken
in seconds per iteration of the method. The table shows the $\widehat{IACT}_{MAX}$,
$\widehat{TNV}_{MAX}$, ..., $\widehat{RTNV}_{MEAN}$. The IACT, TNV, and RTNV are defined in section \ref{SS: preliminaries}.\label{tab:Comparing-Sampler-I:Diffusion}}

\centering{}%
\begin{tabular}{c|c|cc|cc}
\hline 
 & I & \multicolumn{2}{c|}{II} & \multicolumn{2}{c}{III}\tabularnewline
\hline 
$N$ & 100 & 500 & 1000 & 500 & 1000\tabularnewline
$\widehat{\textrm{IACT}}_{\textrm{MAX}}(\theta)$ & 77.27 & 6610.81 & 9535.97 & 282.65 & 175.38\tabularnewline
$\widehat{\textrm{TNV}}_{\textrm{MAX}}(\theta)$ & 554.80 & 52093.18 & 160395.02 & 4632.63 & 6696.01\tabularnewline
$\widehat{\textrm{RTNV}}_{\textrm{MAX}}(\theta)$ & 1 & 93.90 & 289.10 & 8.35 & 12.07\tabularnewline
$\widehat{\textrm{IACT}}_{\textrm{MEAN}}(\theta)$ & 19.76 & 436.58 & 516.13 & 37.31 & 28.38\tabularnewline
$\widehat{\textrm{TNV}}_{\textrm{MEAN}}(\theta)$ & 141.88 & 3440.25 & 8681.31 & 611.51 & 1083.55\tabularnewline
$\widehat{\textrm{RTNV}}_{\textrm{MEAN}}(\theta)$ & 1 & 24.25 & 61.19 & 4.31 & 7.64\tabularnewline
Time & 7.18 & 7.88 & 16.82 & 16.39 & 38.18\tabularnewline
\hline 
\end{tabular}
\end{table}

\section{Discussion \label{sec:discussion}}
Our article shows how to scale up 
the particle MCMC in terms of the
number of parameters and the number of observations by 
expressing the target density of the PMCMC in terms of the 
basic uniform or standard normal random numbers used in the sequential Monte Carlo algorithm,
rather than in terms of state particles. The parameters that can be drawn efficiently conditional on the particles are generated by a particle Gibbs step(s); all the other parameters are drawn in a Metropolis-within-Gibbs step(s) by conditioning on the basic uniform or standard normal random variables; e.g., parameters that are highly correlated with states, or parameters whose generation is expensive when conditioning on the states. The empirical results show that 
the \corrPMMHPG{} is scalable in the number of parameters and the number of observations and is much more efficient
than the competing MCMC methods. We also show that the particle Gibbs (PG) of \citet{andrieuetal2010} and the correlated pseudo marginal Metropolis-Hastings (CPMMH) algorithm of \citet{Deligiannidis2018} are special cases of the \corrPMMHPG{}. 

The \corrPMMHPG{} is useful when it is necessary  estimate complex statistical models that are very difficult to estimate by existing methods.
For example: complex factor stochastic volatility models where the factors and idiosyncratic errors follow mixture of normal distributions; multivariate financial time series model with recurrent neural network type architectures,
e.g. the Long Short-term Memory Model of \citet{Hochreiter1997} and the statistical recurrent unit of \citet{Oliva2017};
stochastic differential equation mixed effect models, which are used in psychology \citep{Oravecz2011} and biomedical work \citep{Leander2015}.

\section{Acknowledgement} The research of Robert Kohn and David Gunawan was partially supported by an ARC Center of Excellence grant CE140100049.


\bibliographystyle{apalike}
\bibliography{references_v1,pfmcmc}



\pagebreak
\makeatletter
\def\renewtheorem#1{%
  \expandafter\let\csname#1\endcsname\relax
  \expandafter\let\csname c@#1\endcsname\relax
  \gdef\renewtheorem@envname{#1}
  \renewtheorem@secpar
}
\def\renewtheorem@secpar{\@ifnextchar[{\renewtheorem@numberedlike}{\renewtheorem@nonumberedlike}}
\def\renewtheorem@numberedlike[#1]#2{\newtheorem{\renewtheorem@envname}[#1]{#2}}
\def\renewtheorem@nonumberedlike#1{
\def\renewtheorem@caption{#1}
\edef\renewtheorem@nowithin{\noexpand\newtheorem{\renewtheorem@envname}{\renewtheorem@caption}}
\renewtheorem@thirdpar
}
\def\renewtheorem@thirdpar{\@ifnextchar[{\renewtheorem@within}{\renewtheorem@nowithin}}
\def\renewtheorem@within[#1]{\renewtheorem@nowithin[#1]}
\makeatother
\renewtheorem{proposition}{Proposition}
\renewcommand{\thesscheme}{S\arabic{sscheme}}
\renewcommand{\thealgorithm}{S\arabic{algorithm}}
\renewcommand{\theremark}{S\arabic{remark}}
\renewcommand{\theequation}{S\arabic{equation}}
\renewcommand{\thetheorem}{S\arabic{theorem}}
\renewcommand{\thesection}{S\arabic{section}}
\renewcommand{\thepage}{S\arabic{page}}
\renewcommand{\thetable}{S\arabic{table}}
\renewcommand{\thefigure}{S\arabic{figure}}
\renewcommand{\theassumption}{S\arabic{assumption}}
\renewcommand{\theproposition}{S\arabic{proposition}}
\renewcommand{\thelemma}{S\arabic{lemma}}
\setcounter{page}{1}
\setcounter{section}{0}
\setcounter{equation}{0}
\setcounter{algorithm}{0}
\setcounter{proposition}{0}
\setcounter{table}{0}
\setcounter{figure}{0}
\setcounter{assumption}{0}
\def\d{{\rm d}}
\section*{Online Supplementary material}
We use the following notation in the supplement. Eq.~(1), algorithm~1,
and sampling scheme~1, etc, refer to the main paper, while Eq.~(S1),
algorithm~S1, and sampling scheme~S1, etc, refer to the supplement.

\section{Proofs\label{S: proofs}}

Lemma~\ref{lemma: prelim1} is used to prove Theorem~\ref{lemma: target distn}.

\begin{lemma}\label{lemma: prelim1}
\begin{enumerate}
\item [(i)]
\begin{align*}
m_1^\theta(\d x_1^i) & = \int_{ \{v_{x1}^i: \mathfrak{X}(v_{x1}^i; \theta,\cdot) \in \d x_1^i\} } \psi(\d v_{x1}^i). 
\end{align*}
\item [(ii)]
For $t \geq 2$,
\begin{align*}
m_t^\theta(\d x_t^i|x_{t-1}^{a_{t-1}^i}) & = \int_{\big  \{v_{xt}^i: \mathfrak{X}\big (v_{xt}^i; \theta, x_{t-1}^{a_{t-1}^i}\big ) \in  \d x_t^i\big \} } \psi(\d v_{xt}^i).
\end{align*}
\item [(iii)] For $ t \geq 2 $,
\begin{align*}
\ov w_{t-1}^j & = \Pr(A_{t-1}^k=j |\ov w_{t-1}^{1:N}) = \int_{ \bigg \{v_{A,t-1}^{i}:\bigg ( \mathfrak{A}\(v_{A,t-1}^{i}; \ov w_{t-1}^{1:N} , x_{t-1}^{1:N} \)\bigg )^k =j\bigg \}}  \psi(\d v_{A,t-1}^{i} )
\end{align*}
\end{enumerate}
\begin{proof}
The proof of parts~(i) and (ii) is straightforward. 
The proof of part~(iii) follows from
assumption~\ref{ass: resampling scheme}.
\end{proof}
\end{lemma}

\begin{proof} [Proof of Theorem~\ref{lemma: target distn}]
To  prove the theorem we carry out the marginalisation by building on the proof of Theorem~3 of \citet{olssonryden2011}.
Let $V_{x,1:T}^{\left(-j_{1:T}\right)}=\left\{ V_{x1}^{\left(-j_{1}\right)},...,V_{xT}^{\left(-j_{T}\right)}\right\} $.
The marginal of $\widetilde{\pi}^{N}\left(dx_{1:T},j_{1:T},d\theta\right)$
is obtained by integrating it over $\left(v_{A,1:T-1},v_{x,1:T}^{\left(-j_{1:T}\right)}\right)$.
We start by integrating over $v_{xT}^{\left(-J_{T}\right)}$ to get
\begin{align*}
\widetilde{\pi}^{N}\left(\d v_{x,1:T-1}^{1:N},\d v_{x,T}^{J_{T}},dv_{A,1:T-1}^{1:N},j_{1:T},\d\theta\right)\\
=\frac{p\left(\d x_{1:T}^{j_{1:T}},\d\theta|y_{1:T}\right)}{N^{T}}\times & \frac{\psi\left(\d v_{x,1:T-1}^{1:N},\d v_{x,T}^{j_{T}},\d v_{A,1:T-1}^{1:N}\right)}{m_{1}^{\theta}\left(\d x_{1}^{j_{1}}\right)
\prod_{t=2}^{T}\overline{w}_{t-1}^{a_{t-1}^{j_{t}}}m_{t}^{\theta}\left(\d x_{t}^{j_{t}}|x_{t-1}^{a_{t-1}^{j_{t}}}\right)}\\
 & \prod_{t=2}^{T}\frac{w_{t-1}^{a_{t-1}^{j_{t}}}f_{t}^{\theta}\left(x_{t}^{j_{t}}|x_{t-1}^{a_{t-1}^{j_{t}}}\right)}
 {\sum_{l=1}^{N}w_{t-1}^{l}f_{t}^{\theta}\left(x_{t}^{j_{t}}|x_{t-1}^{l}\right)}.
\end{align*}
Now, integrate over $\left\{ v_{xT}^{j_{T}}:\mathfrak{X}\(v^{j_T}_{xT}; \theta, x_{T-1}^{a_{T-1}^{j_T}}
\) \in \d x_{T}^{j_{T}}\right\} $ using part~(ii) of 
lemma~\ref{lemma: prelim1}
to obtain,
\begin{align*}
\widetilde{\pi}^{N}\left(\d v_{x,1:T-1}^{1:N},dx_{T}^{j_{T}},\d v_{A,1:T-1}^{1:N},j_{1:T},\d\theta\right)\\
=\frac{p\left(\d x_{1:T}^{j_{1:T}},\d\theta|y_{1:T}\right)}{N^{T}}\times & \frac{\psi\left(\d v_{x,1:T-1}^{1:N},\d v_{A,1:T-1}^{1:N}\right)}{m_{1}^{\theta}\left(\d x_{1}^{j_{1}}\right)
\prod_{t=2}^{T-1}\overline{w}_{t-1}^{a_{t-1}^{j_{t}}}m_{t}^{\theta}\left(\d x_{t}^{j_{t}}|x_{t-1}^{a_{t-1}^{j_{t}}}\right)\overline{w}_{T-1}^{a_{T-1}^{j_{T}}}}\\
 & \prod_{t=2}^{T}\frac{w_{t-1}^{a_{t-1}^{j_{t}}}f_{t}^{\theta}\left(x_{t}^{j_{t}}|x_{t-1}^{a_{t-1}^{j_{t}}}\right)}{\sum_{l=1}^{N}w_{t-1}^{l}
 f_{t}^{\theta}\left(x_{t}^{j_{t}}|x_{t-1}^{l}\right)}.
\end{align*}
Then, integrate over $\left\{ v_{A,T-1}^{i}: \( \mathfrak{A}(v_{A,T-1}^{i}; \ov w_{T-1}^{1:N},x_{T-1}^{1:N})\)^{j_T}=a_{T-1}^{j_{T}}\right\} $ using part~(iii) of lemma~\ref{lemma: prelim1} and then sum over $a_{T-1}^{j_T}$
to obtain
\begin{align*}
\widetilde{\pi}^{N}\left(\d v_{x,1:T-1}^{1:N},\d x_{T}^{j_{T}},dv_{A,1:T-2}^{1:N},j_{1:T},\d\theta\right)\\
=\frac{p\left(\d x_{1:T}^{j_{1:T}},\d\theta|y_{1:T}\right)}{N^{T}}\times & \frac{\psi\left(\d v_{x,1:T-1}^{1:N},\d v_{A,1:T-2}^{1:N}\right)}{m_{1}^{\theta}\left(\d x_{1}^{j_{1}}\right)
\prod_{t=2}^{T-1}\overline{w}_{t-1}^{a_{t-1}^{j_{t}}}m_{t}^{\theta}\left(\d x_{t}^{j_{t}}|x_{t-1}^{a_{t-1}^{j_{t}}}\right)}\\
 & \prod_{t=2}^{T-1}\frac{w_{t-1}^{a_{t-1}^{j_{t}}}f_{t}^{\theta}
 \left(x_{t}^{j_{t}}|x_{t-1}^{a_{t-1}^{j_{t}}}\right)}{\sum_{l=1}^{N}w_{t-1}^{l}f_{t}^{\theta}\left(x_{t}^{j_{t}}|x_{t-1}^{l}\right)}.
\end{align*}
We repeat this  for $t=T-2,...,2$, to obtain, 
\begin{align*}
\widetilde{\pi}^{N}\left(\d v_{x1}^{1:N},dx_{2:T}^{j_{2:T}},j_{1:T},\d\theta\right)\\
=\frac{p\left(\d x_{1:T}^{j_{1:T}},\d\theta|y_{1:T}\right)}{N^{T}}\times & \frac{\psi\left(\d v_{x1}^{1:N}\right)}{m_{1}^{\theta}\left(\d x_{1}^{j_{1}}\right)
}.
\end{align*}
Finally, integrate over $v_{x1}^{(-j_1)} $ and then integrate over
$\left\{ v_{x1}^{j_{1}}:\mathfrak{X}\(v^{j_1}_{x1}; \theta, x_{T-1}^{a_{T-1}^{j_T}} \) \in \d x_{1}^{j_{1}}\right\} $
using part~(i) of lemma~\ref{lemma: prelim1} to obtain the result. 
\end{proof}

\begin{proof}[Proof of theorem~\ref{lemma: alt expression target distn}]
We have that,
\begin{align*}
& \frac{ p(\d x_{1:T}^{j_{1:T}},\d\theta|y_{1:T})}{m_1^\theta(\d x_1^{j_1})\prod_{t=2}^T {\ov w}_{t-1} ^{a_{t-1}^{j_t} } m_t^\theta(\d x_t^{j_t}|x_{t-1}^{a_{t-1}^{j_t}})} \times \frac{p(y_{1:T})}{p(\d\theta)} \times
\prod_{t=2}^{T}\frac{w_{t-1}^{a_{t-1}^{j_{t}}}f_{t}^{\theta}\left(x_{t}^{j_{t}}|x_{t-1}^{a_{t-1}^{j_{t}}}\right)}
{\sum_{l=1}^{N}w_{t-1}^{l}f_{t}^{\theta}\left(x_{t}^{j_{t}}|x_{t-1}^{l}\right)}\\
& =\frac{g_1^\theta(y_1|x_1^{j_1}) f_1^\theta(\d x_1^{j_1})\prod_{t=2}^T g_t^\theta(y_t|x_t^{j_t})
f_t^\theta(\d x_t^{j_t}|x_{t-1}^{j_{t-1}}) }
{m_1^\theta(\d x_1^{j_1})\prod_{t=2}^T {\ov w}_{t-1} ^{a_{t-1}^{j_t} } m_t^\theta(\d x_t^{j_t}|x_{t-1}^{a_{t-1}^{j_t}})}
\times
\prod_{t=2}^{T}\frac{w_{t-1}^{a_{t-1}^{j_{t}}}f_{t}^{\theta}\left(x_{t}^{j_{t}}|x_{t-1}^{a_{t-1}^{j_{t}}}\right)}
{\sum_{l=1}^{N}w_{t-1}^{l}f_{t}^{\theta}\left(x_{t}^{j_{t}}|x_{t-1}^{l}\right)}\\
& = w_1^{j_1}\prod_{t=1}^{T-1}\(\sum_{i=1}^N w_t^i \)
\prod_{t=2}^T \frac{ g_t^\theta(y_t|x_t^{j_t})
f_t^\theta(\d x_t^{j_t}|x_{t-1}^{a^{j_t}_{t-1}}) }{ m_t^\theta(
\d x_t^{j_t}|x_{t-1}^{a_{t-1}^{j_t}})}
\prod_{t=2}^{T}\frac{f_{t}^{\theta}\left(x_{t}^{j_{t}}|x_{t-1}^{j_{t-1}}\right)}
{\sum_{l=1}^{N}w_{t-1}^{l}f_{t}^{\theta}\left(x_{t}^{j_{t}}|x_{t-1}^{l}\right)}\\
& = \prod_{t=1}^{T}\(\sum_{i=1}^N w_t^i \)\prod_{t=2}^{T}\frac{w_{t-1}^{j_{t-1}}f_{t}^{\theta}\left(x_{t}^{j_{t}}|x_{t-1}^{j_{t-1}}\right)}
{\sum_{l=1}^{N}w_{t-1}^{l}f_{t}^{\theta}\left(x_{t}^{j_{t}}|x_{t-1}^{l}\right)}
{\ov w}_T^{j_T}.
\end{align*}
\end{proof}

\begin{proof}[Proof of corollary~\ref{corr: target with j integrated out}]
The proof follows from theorem~\ref{lemma: alt expression target distn} by summing
 the terms in the target distribution \eqref{eq:Targetdistribution2}
that include $j_{1:T}$, i.e.,
 \begin{align*}
\overline{w}_{T}^{j_{T}}\prod_{t=2}^{T}\frac{w_{t-1}^{j_{t-1}}f_{t}^{\theta}
\left(x_{t}^{j_{t}}|x_{t-1}^{j_{t-1}}\right)}{\sum_{l=1}^{N}w_{t-1}^{l}f_{t}^{\theta}\left(x_{t}^{j_{t}}|x_{t-1}^{l}\right)}.
\end{align*}
\end{proof}
\begin{proof}[Proof of corollary~\ref{cor: j cond on rest}]
The proof follows from theorem~\ref{lemma: alt expression target distn} and corollary~\ref{corr: target with j integrated out}. 
\end{proof}

\section{Assumptions \label{assumptionSMC}}
This section outlines the assumptions required for the \PMCMC{} algorithms.
For $t\geq1$, we define $\pi_t(x_{1:t}|\theta) := p(x_{1:t}|y_{1:t},\theta)$, $\mathcal{S}_{t}^{\theta} :=\left\{x_{1:t}\in\boldsymbol{\chi}^t: \pi_t(x_{1:t}|\theta)  >0\right\}$ and \newline
$\mathcal{Q}_{t}^{ \theta} :=\left\{ x_{1:t}\in\boldsymbol{\chi}^t: \pi_{t-1}(x_{1:t-1}|\theta) m_{t}^{ \theta }\left( x_{t}| x_{1:t-1}\right)>0\right\}$.
We follow \cite{andrieuetal2010} and assume that
\begin{assumption}
 $\mathcal{S}_{t}^{ \theta} \subseteq \mathcal{Q}_{t}^{\theta }$
for any $ \theta\in \Theta$ and $t=1,...,T$.
\label{assu:propstatespace}
\end{assumption}

Assumption \ref{assu:propstatespace} ensures that the proposal densities  $m_{t}^{\theta }\left( x_{t}| x_{1:t-1}\right)$
can be used to approximate $\pi_t\left(x_{t}| x_{1:t-1}, \theta \right)$
for $t\geq1$. If $m_t^\theta(\cdot)$ is a mixture of  some general proposal density $\wt m_{t}^{\theta }\left( x_{t}| x_{1:t-1}\right)$ and 
$f_t^\theta(\cdot)$, with $f_t^\theta(\cdot)$ having nonzero weight, 
and $g_t^\theta(y_t|x_t)> 0 $ for all $\theta$, then 
assumption~\ref{assu:propstatespace} is satisfied.
In particular, if we use the bootstrap filter then
$m_t^\theta(\cdot) = f_t^\theta(\cdot)$.
Furthermore, $g_t^\theta(y_t|x_t)> 0 $ for all $\theta$ for the univariate stochastic volatility model
in section~\ref{SV with leverage}. We follow \cite{andrieuetal2010} and assume that
\begin{assumption}\label{ass: resampling scheme}
For any $k=1,...,N$ and $t=2,..,T$, the resampling
scheme $\mathcal{M}\left(a_{t-1}^{1:N}|\bar{w}_{t-1}^{1:N}, x_{t-1}^{1:N}\right)$
satisfies $\Pr\left(A_{t-1}^{k}=j|\bar{w}_{t-1}^{1:N}\right) = \bar{w}_{t-1}^{j}$.\label{assu:resampling}
\end{assumption}
Assumption \ref{assu:resampling} is used in section~\ref{sub:Target-Distributions univariate} to prove theorem~\ref{lemma: target distn} 
and is satisfied by all the popular resampling schemes, e.g.,
multinomial, systematic and  residual resampling. We refer to
\citet{doucetal2005} for a comparison between resampling schemes
and \citet{doucetetal2000,merveetal2001,Scharth:2016} for the choice
of proposal densities.

\section{Algorithms\label{S: algorithms}}

\subsection{The SMC algorithms\label{SS: SMC algorithms}}

This section describes how to implement the SMC algorithm described in sections~\ref{SS: smc} and \ref{sub:Flexible-Correlated-PMMH+PG sampling scheme}.
The SMC algorithm takes the number of particles $N$, the parameters $\theta$, the random variables used to propagate state particles $V_{x,1:T}^{1:N}$,
and the random numbers used in the resampling steps $V_{A,1:T-1}^{1:N}$ as the inputs; it outputs the set of particles $x_{1:T}^{1:N}$, ancestor indices $a_{1:T-1}^{1:N}$, and weights $w_{1:T}^{1:N}$. At $t=1$,
we obtain the particles $x_{1}^{1:N}$ as a function of the basic random numbers $v_{x1}^{1:N}$ using \eqref{V_x example SV1}; we then compute the weights of all particles in step (2). 

Step (3a) sorts the particles from smallest to largest using the Euclidean sorting procedure of \citet{Choppala2016} to obtain the sorted particles and weights.
Algorithm \ref{alg:Multinomial-Resampling-Algorithm} resamples the particles using multinomial sampling 
to obtain the ancestor index $A_{1:T-1}^{1:N}$ in the original order of the particles in steps~(3b) and (3c).
Steps~(3a) - (3c) define the mapping $A_{t-1}^{1:N}=\mathfrak{ A}\(v_{A,t-1}^{1:N} ; \overline{w}_{t-1}^{1:N},x_{t-1}^{1:N}\)$.
Steps~(3d) generates the particles $x_{t}^{1:N}$ as a function of the basic random numbers $v_{xt}^{1:N}$ using \eqref{V_x example SVt} for the univariate stochastic volatility with leverage
and then computes the weights of all particles in step~(3e).

\begin{algorithm}[H]
\caption{The sequential Monte Carlo algorithm \label{alg:The-Sequential-Monte carlo algorithm}}
Inputs: $N, \theta$, $V_{x,1:T}^{1:N}$ and $V_{A,1:T-1}^{1:N}$.\\  
Outputs: $x_{1:T}^{1:N}, a_{1:T-1}^{1:N},w_{1:T}^{1:N}$.
\begin{enumerate}
\item
For $t=1$, set $X_{1}^{i}=x_{1}^i = \mathfrak{X}\left(v_{x1}^{i}; \theta, \cdot \right)$ for $i=1, \dots, N$.
\item  Compute the importance weights \begin{align*}
    w_{1}^{i}=\frac{f_{1}^{\theta}\left(x_{1}^{i}\right)g_{1}^{\theta}\left(y_{1}|x_{1}^{i}\right)}{m_{1}^{\theta}\left(x_{1}^{i}\right)},
\quad \text{for} \quad i=1,...,N.
\end{align*}
and  normalize  $\overline{w}_{1}^{i}=w_{1}^{i}/\sum_{j=1}^{N}w_{1}^{j}$
for $i=1,...N$.

\item

For $t = 2, \dots, T$,
\begin{enumerate}
\item Sort the particles $x_{t-1}^{i}$ using the Euclidean sorting of \citet{Choppala2016} and obtain the sorted index $\zeta_{i}$ for $i=1,...,N$ and the sorted particles and weights $\wt x_{t-1}^i = x_{t-1}^{\zeta_i} $ and $\wt {\ov w}^i_{t-1} = \ov w_{t-1}^{\zeta_i}$, for $i=1, \dots, N$.

\item Obtain the ancestor index based on the sorted particles $\wt A_{t-1}^{1:N} = \wt a_{t-1}^{1:N}$ using a resampling scheme
${\cal{M}}(\wt a_{t-1}^{1:N} |\wt x_{t-1}^{1:N}, {\wt {\ov w}}^i_{t-1}),$  e.g. the multinomial resampling in Algorithm~\ref{alg:Multinomial-Resampling-Algorithm}.

\item Obtain the ancestor index based on the original order of the particles $A_{t-1}^{i}$ for $i=1,...,N$.

\item  Generate $V_{xt}^i \sim \psi_{xt}(\cdot)$ and
set $ X_{t}^{i}= x_{t}^i = \mathfrak{X} \left(v_{xt}^{i}; \theta, x_{t-1}^{a_{t-1}^i} \right)$
for $i=1,...,N$.

\item   Compute the importance weights
\begin{align*}
w_{t}^{i}=\frac{f_{t}^{\theta}\left(x_{t}^{i}|x_{t-1}^{a_{t-1}^i},y_{t-1}\right)g_{t}^{\theta}\left(y_{t}|x_{t}^{i}\right)}{m_{t}^{\theta}\left(x_{t}^{i}|x_{t-1}^{a_{t-1}^{i}}\right)},
\end{align*}
for $i=1,...,N$
and normalize to obtain $\overline{w}_{t}^{i}$
for $i=1,...N$.
\end{enumerate}
\end{enumerate}
\end{algorithm}

\begin{algorithm}[H]
\caption{Multinomial Resampling Algorithm \label{alg:Multinomial-Resampling-Algorithm}}

Input: $v_{At-1}^{1:N}$, $\widetilde{x}_{t-1}^{1:N}$, and $\widetilde{\overline{w}}_{t-1}^{1:N}$

Output: $\widetilde{A}_{t-1}^{1:N}$
\begin{enumerate}
\item Compute the cumulative weights based on the sorted particles $\left\{ \widetilde{x}_{t-1}^{1:N},\widetilde{\overline{w}}_{t-1}^{1:N}\right\} $
\[
\widehat{F}_{t-1}^{N}\left(j\right)=\sum_{i=1}^{j}\widetilde{\overline{w}}_{t-1}^{i}.
\]
\item Set $\widetilde{A}_{t-1}^{i}=\underset{j}{\min}\,\, \widehat{F}_{t-1}^{N}\left(j\right)\geq v_{At-1}^{i}$
for $i=1,...N$, and note that $\widetilde{A}_{t-1}^{i}$ for $i=1,...,N$
is the ancestor index based on the sorted particles.
\end{enumerate}
\end{algorithm}

\subsection{The backward simulation algorithm\label{SS: backward simulation1}}

\begin{algorithm}[H]
\caption{The Backward simulation algorithm \label{alg:The-backward simulation algorithm} }

\begin{enumerate}
\item Sample $J_{T}=j_{T}$ conditional on $\left(V_{x,1:T}^{1:N},V_{A,1:T-1}^{1:N},\theta\right)$,
with probability proportional to $w_{T}^{j_{T}}$, and choose $x_{T}^{j_{T}}$;
\item For $t=T-1,...,1$, sample $J_{t}=j_{t}$, conditional on \newline
$\left(V_{x,1:t}^{1:N},V_{A,1:t}^{1:N},j_{t+1:T},x_{t+1}^{j_{t+1}},...,x_{T}^{j_{T}}\right)$,
and  choose $J_t=l$ with probability proportional to $w_{t}^{l}f_{\theta}\left(x_{t+1}^{j_{t+1}}|x_{t}^{l}\right)$.
\end{enumerate}
\end{algorithm}

\subsection{Multidimensional Euclidean Sorting Algorithm \label{Multidimensional Sorting}}

This section discusses the multidimensional Euclidean sorting algorithm
used in Algorithms~\ref{alg:The-conditional Sequential-Monte carlo algorithm} and \ref{alg:The-Sequential-Monte carlo algorithm}. 
Let $x_{t}^{i}$ be the $n_{x}$- dimensional particle at a time $t$, $x_{t}^{i}=\left(x_{t,1}^{i},...,x_{t,n_{x}}^{i}\right)^{\top}$.
Let $d\left(x_{t}^{j},x_{t}^{i}\right)$ be the Euclidean distance
between two multidimensional particles. 
Algorithm~\ref{alg:Multidimensional-Euclidean-Sorting}
is the multidimensional Euclidean sorting algorithm that generates the 
sorted particles and weights with associated sorted indices. 
The first sorted index is the index of the particle 
having the smallest value
along its first dimension. The other particles are chosen in a way
that minimises the Euclidean distance between the last selected particle
and the set of all remaining particles. 

\begin{algorithm}[H]

\caption{Multidimensional Euclidean sorting algorithm \label{alg:Multidimensional-Euclidean-Sorting}}

Input: $x_{t}^{1:N}$, $\overline{w}_{t}^{1:N}$ 

Output: sorted particles $\widetilde{x}_{t}^{1:N}$, sorted weights
$\widetilde{\overline{w}}_{t}^{1:N}$, sorted indices $\zeta_{1:N}$

Let $\chi^{j}=\left\{ 1,...,N\right\} $ be the index set. 

When $j=1$,
\begin{itemize}
\item Obtain the index $\zeta_{1}=\min_{i}$$x_{t,1}^{i}$, for all $ i\in\chi^{j}$.
\end{itemize}
For $j=2,...,N$
\begin{itemize}
\item Set $x_{t}^{*}=x_{t}^{j-1}$.
\item Update the index $\chi^{j}$ by removing $\zeta_{j-1}$ from the index
set.
\item Obtain the index $\zeta_{j}=\min_{i}d\left(x_{t}^{*},x_{t}^{i}\right),$$\forall i\in\chi^{j}$.
\end{itemize}
Sort the particles and weights according to the indices $\zeta_{1:N}$ to obtain the sorted particles $\widetilde{x}_{t}^{1:N}$, and sorted weights $\widetilde{\overline{w}}_{t}^{1:N}$. 
\end{algorithm}

\subsection{Constrained conditional multinomial resampling algorithm\label{constrainedresampling}}

Algorithm \ref{alg:Multinomial-Resampling-Algorithm for CCSMC} describes the constrained multinomial resampling algorithm used in CCSMC.
It takes the sorted particles $\widetilde{x}_{t-1}^{1:N}$ and weights 
$\widetilde{\overline{w}}_{t-1}^{1:N}$ as the inputs and produces the basic
random numbers $V_{A,1:T-1}^{1:N}$
and the ancestor indices based on the sorted particles and weights $\widetilde{A}_{t-1}^{1:N}$.
The first step computes the cumulative weights based on the sorted particles; 
the second step generates the random numbers $V_{A,1:T-1}^{1:N}$ and ancestor 
indices $\widetilde{A}_{t-1}^{1:N}$. 

\begin{algorithm}[]
\caption{Constrained multinomial resampling algorithm for CCSMC\label{alg:Multinomial-Resampling-Algorithm for CCSMC}}

Input: $\widetilde{x}_{t-1}^{1:N}$, and $\widetilde{\overline{w}}_{t-1}^{1:N}$

Output: $V_{A,1:T-1}^{1:N}$ and $\widetilde{A}_{t-1}^{1:N}$.
\begin{enumerate}
\item Compute the cumulative weights based on the sorted particles $\left\{ \widetilde{x}_{t-1}^{1:N},\widetilde{\overline{w}}_{t-1}^{1:N}\right\} $,
\[
\widehat{F}_{t-1}^{N}\left(j\right)=\sum_{i=1}^{j}\widetilde{\overline{w}}_{t-1}^{i}.
\]
\item
Generate $N-1$ uniform $(0,1)$ random numbers $v_{At-1}^i\sim\psi_{A,t-1}\left(\cdot\right)$ 
for $i=1,...,N$, such that $i\neq j_{t}$, and set $\widetilde{A}_{t-1}^{i}=\underset{j}{\min}\,\, \widehat{F}_{t-1}^{N}\left(j\right)\geq v_{At-1}^{i}$.
For $i=j_{t}$,
\begin{align*}
v_{At-1}^{j_{t}}\sim U\left(\widehat{F}_{t-1}^{N}\left(j_{t}-1\right),\widehat{F}_{t-1}^{N}\left(j_{t}\right)\right),
\quad
\text{where} \\
\widehat{F}_{t-1}^{N}\left(j_{t}-1\right)=\sum_{i=1}^{j_{t}-1}\widetilde{\overline{w}}_{t-1}^{i}\quad
\text{and} \quad
\widehat{F}_{t-1}^{N}\left(j_{t}\right)=\sum_{i=1}^{j_{t}}\widetilde{\overline{w}}_{t-1}^{i}.
\end{align*}
\end{enumerate}
\end{algorithm}

\subsection{The conditional sequential Monte Carlo algorithm\label{CSMC_algorithm}}

\begin{algorithm}[H]
\caption{The conditional sequential Monte Carlo algorithm \label{alg:The-conditional Sequential-Monte carlo algorithm standard PG} }

Inputs: $N$, $\theta$, $x_{1:T}^{j_{1:T}}$, and $j_{1:T}$

Outputs: $x_{1:T}^{1:N}$, $a_{1:T-1}^{1:N}$, $w_{1:T}^{1:N}$.

Fix $X_{1:T}^{j_{1:T}}=x_{1:T}^{j_{1:T}}$, $A_{1:T-1}^{J}=j_{1:T-1}$,
and $J_{T}=j_{T}$.
\begin{enumerate}
\item For $t=1$

\begin{enumerate}
\item Sample $v_{x1}^{i}\sim\psi_{x1}\left(\cdot\right)$ and set $x_1^{i} = \mathfrak{X}\left(v_{x1}^{i}; \theta, \cdot\right)$
for $i=1,...,N\setminus\left\{ j_{1}\right\} $.
\item Compute the importance  weights $w_{1}^{i}=\nicefrac{f_{1}^{\theta}\left(x_{1}^{i}\right)g_{1}^{\theta}\left(y_{1}|x_{1}^{i}\right)}{m_{1}^{\theta}\left(x_{1}^{i}\right)}$,
for $i=1,...,N$,
and normalize $\overline{w}_{1}^{i}= w_{1}^{i}/\sum_{j=1}^{N}w_{1}^{j}$.
\end{enumerate}

\item For $t\geq2$

\begin{enumerate}
\item Use a conditional multinomial sampler in 
algorithm \ref{alg:Multinomial-Resampling-Algorithm for standard CSMC} and obtain the ancestor indices of the particles $A_{t-1}^{1:N}$.
\item Sample $v_{xt}^{i}\sim\psi_{xt}\left(\cdot\right)$ for $i=1,...,N\setminus\left\{ j_{t}\right\} $.
\item Set $x_{t}^{i}=\mathfrak{X}\left(v_{xt}^{i};\theta,x_{t-1}^{a_{t-1}^{i}}\right)$
for $i=1,...,N\setminus\left\{ j_{t}\right\} $.
\item Compute the importance weights,
\begin{align*}
    w_{t}^{i} & =\frac{f_{t}^{\theta}\left(x_{t}^{i}|x_{t-1}^{a_{t-1}^{i}},y_{t-1}\right)g_{t}^{\theta}
   \left(y_{t}|x_{t}^{i}\right)}{m_{t}^{\theta}\left(x_{t}^{i}|x_{t-1}^{a_{t-1}^{i}}\right)},
\quad \text{for} \quad i=1,...,N,
\end{align*}
and  normalize the $\overline{w}_{t}^{i}$.
\end{enumerate}
\end{enumerate}
\end{algorithm}

\begin{algorithm}[H]
\caption{Conditional multinomial resampling algorithm\label{alg:Multinomial-Resampling-Algorithm for standard CSMC}}

Input: ${x}_{t-1}^{1:N}$, and ${\overline{w}}_{t-1}^{1:N}$

Output: ${A}_{t-1}^{1:N}$.
\begin{enumerate}
\item Compute the cumulative weights based on the particles $\left\{{x}_{t-1}^{1:N},{\overline{w}}_{t-1}^{1:N}\right\} $,
\[
\widehat{F}_{t-1}^{N}\left(j\right)=\sum_{i=1}^{j}{\overline{w}}_{t-1}^{i}.
\]
\item
Generate $N-1$ uniform $(0,1)$ random numbers $v_{At-1}^i\sim\psi_{A,t-1}\left(\cdot\right)$ 
for $i=1,...,N$, such that $i\neq j_{t}$, and set ${A}_{t-1}^{i}=\underset{j}{\min}\,\, \widehat{F}_{t-1}^{N}\left(j\right)\geq v_{At-1}^{i}$.
For $i=j_{t}$, ${A}_{t-1}^{i}=j_{t-1}$

\end{enumerate}
\end{algorithm}

\section{The \corrPMMHPG{} with a general proposal in 
part~1\label{SS: The full correlated PMMH+PG}}

Algorithm~\ref{alg:Sampling-Scheme:-The full correlated PMMH+PG} is a version of the \corrPMMHPG{} with a general Metropolis-within-Gibbs proposal for the parameter and the basic uniform and standard normal variables;
part~1 includes the correlated pseudo-marginal proposal in \citet{Deligiannidis2018}, the block pseudo-marginal in \citet{Tran:2016}, the standard pseudo marginal of \citet{Andrieu:2009},
and the proposal in part~1 of Algorithm~\ref{alg:Sampling-Scheme:-The correlated PMMH+PG} as special cases.
It is unnecessary to move the basic uniform and standard normal variables in part~1 to obtain a valid sampling scheme
because they are generated in a Gibbs step in part~4 of 
Algorithms~\ref{alg:Sampling-Scheme:-The correlated PMMH+PG} 
and \ref{alg:Sampling-Scheme:-The full correlated PMMH+PG}.

\begin{algorithm}[H]
\caption{ The \corrPMMHPG{} with a general proposal in 
part~1.\label{alg:Sampling-Scheme:-The full correlated
PMMH+PG}}

Given initial values for $V_{x,1:T}^{1:N}$, $V_{A,1:T-1}^{1:N}$, $J_{1:T}$,
and $\theta$
\begin{enumerate}
\item [Part 1:] General Metropolis-within-Gibbs move. 
\begin{enumerate}
\item Sample $\left(\theta_{1}^{*}, v_{x,1:T}^{*,1:N},v_{A,1:T-1}^{*,1:N}\right)
\sim q_{1}\left(\cdot, \d v_{x,1:T}^{*,1:N},\d v_{A,1:T-1}^{*,1:N}|v_{x,1:T}^{1:N},v_{A,1:T-1}^{1:N},\theta_{2},\theta_{1}\right)$.
\item Run the sequential Monte Carlo algorithm and evaluate $\widehat{Z}\left(v^{*,1:N}_{x,1:T},v_{A,1:T-1}^{*,1:N},\theta_{1}^{*},\theta_{2}\right)$.
\item Accept the proposed values $\left(\theta_{1}^{*}, v_{x,1:T}^{*,1:N},v_{A,1:T-1}^{*,1:N}\right)$ with probability
\begin{align}\label{eq: MH ratio11}
\alpha\left(\theta_{1},v_{x,1:T}^{1:N},v_{A,1:T-1}^{1:N};\theta_{1}^{*},v_{x,1:T}^{*,1:N},v_{A,1:T-1}^{*,1:N}|\theta_{2}\right) =\notag \\
1\land\frac{\widehat{Z}\left(v_{x,1:T}^{*,1:N},v_{A,1:T-1}^{*,1:N},\theta_{1}^{*},\theta_{2}\right)p\left(\theta_{1}^{*}|\theta_{2}\right)\psi\left(\d v_{x,1:T}^{*,1:N},dv_{A,1:T-1}^{*,1:N} \right)}
{\widehat{Z}\left(v_{x,1:T}^{1:N},v_{A,1:T-1}^{1:N},\theta_{1},\theta_{2}\right)p\left(\theta_{1}|\theta_{2}\right) \psi\left(\d v_{x,1:T}^{1:N},dv_{A,1:T-1}^{1:N} \right)} \notag \\
\times\frac{q_{1}\left(\theta_{1},\d v_{x,1:T}^{1:N},\d v_{A,1:T-1}^{1:N}|v_{x,1:T}^{*,1:N},v_{A,1:T-1}^{*,1:N},\theta_{2},\theta_{1}^{*}\right)}
{q_{1}\left(\theta_{1}^{*},\d v_{x,1:T}^{*,1:N},\d v_{A,1:T-1}^{*,1:N}|v_{x,1:T}^{1:N},v_{A,1:T-1}^{1:N},\theta_{2},\theta_{1}\right)}.
\end{align}

\end{enumerate}
\item [Part 2:] Sample $J_{1:T}=j_{1:T}\sim\widetilde{\pi}^{N}\left(\cdot|v_{x,1:T}^{1:N},v_{A,1:T-1}^{1:N},\theta\right)$
given in \eqref{eq:conditional distribution j} using the
backward simulation algorithm (Algorithm~\ref{alg:The-backward simulation algorithm}).

\item [Part 3:] PG sampling. 
\begin{enumerate}
\item Sample $\theta_{2}^{*}\sim q_{2}\left(\cdot|x_{1:T}^{j_{1:T}},j_{1:T},\theta_{1},\theta_{2}\right)$.
\item Accept the proposed values $\theta_{2}^{*}$ with probability
\begin{align}
\alpha\left(\theta_{2};\theta_{2}^{*}|x_{1:T}^{j_{1:T}},j_{1:T},\theta_{1}\right)\nonumber \\
= & 1\land\frac{\widetilde{\pi}^{N}\left(\theta_{2}^{*}|x_{1:T}^{j_{1:T}},j_{1:T},\theta_{1}\right)}{\widetilde{\pi}^{N}\left(\theta_{2}|x_{1:T}^{j_{1:T}},j_{1:T},
\theta_{1}\right)}\frac{q_{2}\left(\theta_{2}|x_{1:T}^{j_{1:T}},j_{1:T},\theta_{1},\theta_{2}^{*}\right)}
{q_{2}\left(\theta_{2}^{*}|x_{1:T}^{j_{1:T}},j_{1:T},\theta_{1},\theta_{2}\right)}. \label{eq: MH ratio 21}
\end{align}
\end{enumerate}
\item [Part 4:] Sample $\left(V_{x,1:T}^{1:N},V_{A,1:T-1}^{1:N}\right)$ from $\widetilde{\pi}^{N}\left(\cdot|x_{1:T}^{j_{1:T}},j_{1:T},\theta\right)$
using constrained conditional sequential Monte Carlo (Algorithm \ref{alg:The-conditional Sequential-Monte carlo algorithm})
and evaluate $\widehat{Z}\left(v_{x,1:T}^{1:N},v_{A,1:T-1}^{1:N},\theta\right)$.
\end{enumerate}
\end{algorithm}

\section{Further empirical results for the univariate SV model with leverage\label{additionalunivariateSV}}
This section gives further empirical results for the univariate SV model with leverage.
Table~\ref{tab:Inefficiency-factor-of univariate SV leverage
example-T6000} shows
the IACT, TNV and RTNV estimates for the parameters in the univariate SV model with leverage estimated
using the five different samplers for the simulated data with $T=6000$ observations. The table shows that:
(1) The PGDA sampler requires around $N=5000$ particles to ensure that the Markov chains do not get stuck.
(2) The PHS sampler requires at least $N=5000$ particles to obtain similar IACT values for the parameters compared to the \corrPMMHPG{} and the CPMMH sampler with
$ N=100$ particles.
(3) In terms of $\textrm{TNV}_{\textrm{MAX}}$, the \corrPMMHPG{} with only $N=100$ particles
is 18.82 times better than PGBS sampler with $N=1000$ particles, is 15.70 times better than the PGDA sampler with $N=5000$ particles,
and is 670.57, 49.12, and 19.32 times better than the PHS sampler with $N=1000,2000,5000$ particles for estimating the parameters of the univariate SV model.
Similar conclusions hold for $\textrm{TNV}_{\textrm{MEAN}}$.
(4)~In terms of $\textrm{TNV}_{\textrm{MAX}}$, the \corrPMMHPG{} with $N=100$ particles is 39.37 times better than the PGDA sampler with $N=5000$ particles, and is 12.56 times better than the CPMMH sampler for estimating the latent volatilities of the univariate SV model. The PGBS sampler is 1.72 times better than \corrPMMHPG{} for estimating the latent volatilities of the univariate SV model.


\begin{table}[H]
\caption{Univariate SV model with leverage for the five samplers. Sampler I:
$\textrm{CPHS}\left(\rho,\tau^{2};\mu,\phi\right)$, Sampler II: $\textrm{PGBS}\left(\mu,\tau^{2},\phi,\rho\right)$,
Sampler III: $\textrm{PGDA}\left(\mu,\tau^{2},\phi,\rho\right)$,
Sampler IV: $\textrm{ CPMMH}\left(\mu,\tau^{2},\phi,\rho\right)$,
and Sampler V: $\textrm{PHS}\left(\tau^{2},\rho;\mu,\phi\right)$
for simulated data with $T=6000$. The symbol ``NA'' means the Markov
chain gets stuck and does not converge. Time is the time in seconds for one iteration of the algorithm. The IACT, TNV, and RTNV are defined in section \ref{SS: preliminaries}. \label{tab:Inefficiency-factor-of univariate SV leverage example-T6000}}

\centering{}%
\begin{tabular}{c|c|c|c|c|c|c|c|c|c|}
\hline 
Param & \multicolumn{1}{c|}{I} & \multicolumn{1}{c|}{II} & \multicolumn{3}{c|}{III} & IV & \multicolumn{3}{c|}{V}\tabularnewline
\hline 
$N$ & 100 & 1000 & 1000 & 2000 & 5000 & 100 & 1000 & 2000 & 5000\tabularnewline
\hline 
$\widehat{\textrm{IACT}}(\phi)$ & 4.68 & 25.02 & NA & NA & 36.52 & 23.61 & 94.06 & 26.44 & 5.69\tabularnewline
$\widehat{\textrm{IACT}}(\mu)$ & 1.41 & 3.97 & NA & NA & 34.09 & 17.80 & 1.11 & 4.06 & 1.99\tabularnewline
$\widehat{\textrm{IACT}}(\tau^{2})$ & 11.80 & 156.59 & NA & NA & 28.73 & 18.61 & 1931.35 & 79.89 & 14.94\tabularnewline
$\widehat{\textrm{IACT}}(\rho)$ & 14.24 & 221.50 & NA & NA & 33.49 & 20.67 & 4441.36 & 163.42 & 35.92\tabularnewline
\hline 
$\widehat{\textrm{IACT}}_{\textrm{MAX}}(\theta)$ & 14.24 & 221.50 & NA & NA & 36.52 & 23.61 & 4441.36 & 163.42 & 35.92\tabularnewline
$\widehat{\textrm{TNV}}_{\textrm{MAX}}(\theta)$ & 14.24 & 268.01 & NA & NA & 223.50 & 12.28 & 9548.92 & 699.44 & 275.15\tabularnewline
$\widehat{\textrm{RTNV}}_{\textrm{MAX}}(\theta)$ & 1 & 18.82 & NA & NA & 15.70 & 0.86 & 670.57 & 49.12 & 19.32\tabularnewline
$\widehat{\textrm{IACT}}_{\textrm{MEAN}}(\theta)$ & 8.03 & 101.77 & NA & NA & 33.21 & 20.17 & 1616.97 & 68.45 & 14.64\tabularnewline
$\textrm{TNV}_{\textrm{MEAN}}(\theta)$ & 8.03 & 123.14 & NA & NA & 202.69 & 10.49 & 3476.49 & 292.97 & 112.14\tabularnewline
$\textrm{RTNV}_{\textrm{MEAN}}(\theta)$ & 1 & 15.33 & NA & NA & 25.24 & 1.31 & 432.94 & 36.48 & 13.97\tabularnewline
\hline 
$\widehat{\textrm{IACT}}_{\textrm{MAX}}$$\left(x_{1:T}\right)$ & 3.26 & 1.56 & NA & NA & 20.97 & 78.73 & 1.31 & 2.26 & 1.25\tabularnewline
$\textrm{TNV}_{\textrm{MAX}}$$\left(x_{1:T}\right)$ & 3.26 & 1.89 & NA & NA & 128.34 & 40.94 & 2.82 & 9.67 & 9.57\tabularnewline
$\textrm{RTNV}_{\textrm{MAX}}$$\left(x_{1:T}\right)$ & 1 & 0.58 & NA & NA & 39.37 & 12.56 & 0.87 & 2.97 & 2.94\tabularnewline
$\widehat{\textrm{IACT}}_{\textrm{MEAN}}$$\left(x_{1:T}\right)$ & 1.19 & 1.02 & NA & NA & 9.82 & 10.76 & 1.01 & 1.02 & 1.01\tabularnewline
$\textrm{TNV}_{\textrm{MEAN}}$$\left(x_{1:T}\right)$ & 1.19 & 1.23 & NA & NA & 60.10 & 5.60 & 2.17 & 4.37 & 7.74\tabularnewline
$\textrm{RTNV}_{\textrm{MEAN}}$$\left(x_{1:T}\right)$ & 1 & 1.03 & NA & NA & 50.50 & 4.71 & 1.82 & 3.67 & 6.50\tabularnewline
\hline 
Time & 1.00 & 1.21 & NA & NA & 6.12 & 0.52 & 2.15 & 4.28 & 7.66\tabularnewline
\hline 
\end{tabular}
\end{table}

Figure~\ref{fig:Left:-The-plots figure RTNV mean and RTNV max} shows the $\textrm{RTNV}_{\textrm{MAX}}$ and $\textrm{RTNV}_{\textrm{MEAN}}$ of the parameters and the latent volatilities estimated using
\corrPMMHPG{} and the CPMMH with $N=100$ particles for the 
univariate SV model. The RTNV is the TNV of the CPMMH relative to the TNV of the CPHS. 

The data is generated from the univariate SV model with 
$T=2000,4000,...,20000$ observations, $\phi=0.98$, $\tau^2=0.5$, 
$\rho=-0.45$, and $\mu=-0.42$. 
The figure also shows the variance of log of estimated likelihood 
evaluated at the true values of the parameters,  
with the variance increasing with the length of the time series. The figure suggests that the \corrPMMHPG{} is much better than the CPMMH sampler for estimating both the parameters and the latent volatilities. 
This suggests that \corrPMMHPG{} performs better
than CPMMH at estimating an SV model with a large number of observations. 


\begin{figure}[H]
\caption{Comparing the performance of CPHS and CPMMH  for estimating a 
univariate stochastic volatility with leverage for different number of 
observations $T$. Left: The plots of $\textrm{RTNV}_{\textrm{MAX}}$ and $\textrm{RTNV}_{\textrm{MEAN}}$
for the parameters $\theta$; Middle: The plots of $\textrm{RTNV}_{\textrm{MAX}}$
and $\textrm{RTNV}_{\textrm{MEAN}}$ for the latent volatilities $x_{1:T}$;
Right: The plot of the variance of log of estimated likelihood $\widehat{Z}$
evaluated at the true values of the parameters with $N=100$ particles. The RTNV is the TNV of the CPMMH relative to the TNV for the \corrPMMHPG{}. The TNV and RTNV are defined in section \ref{SS: preliminaries}.
\label{fig:Left:-The-plots figure RTNV mean and RTNV max}}

\begin{centering}
\includegraphics[width=15cm,height=8cm]{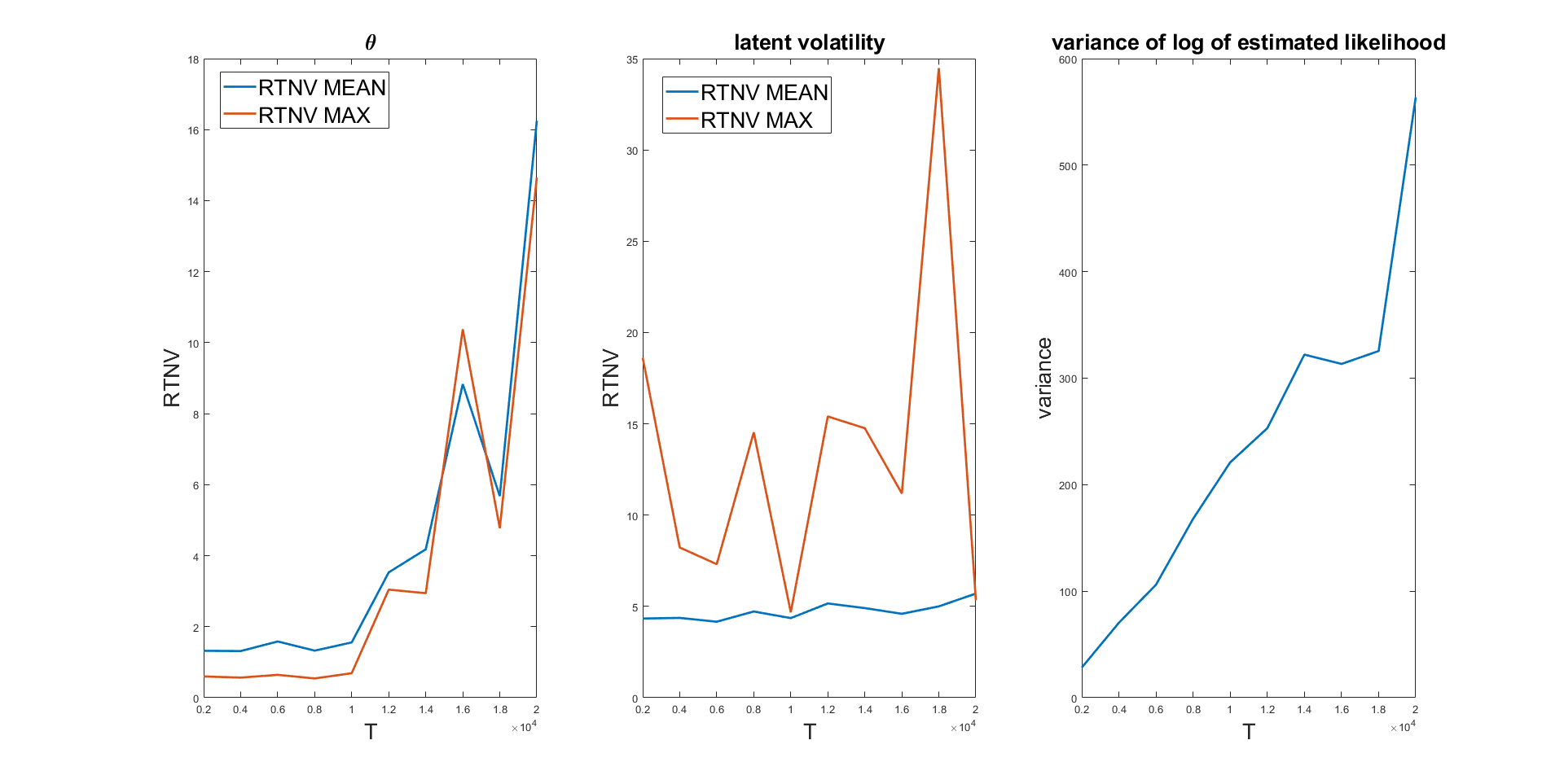}
\par\end{centering}
\end{figure}

\section{Target distribution for the factor SV model\label{subsec:Target-Distributions-for factor SV}}

This section defines the target density for the factor SV
model in section~\ref{S: factor SV model}; 
it includes all the random variables produced by
$K+S$ univariate SMC methods that generate the factor log-volatilities
${\lambda}_{k,1:T}$ for $k=1,...,K$ and the idiosyncratic
log-volatilities ${h}_{s,1:T}$ for $s=1,...,S$, as well
as the latent factors ${f}_{1:T}$, the parameters of the individual idiosyncratic error SV's ${\theta}_{\epsilon,1:S}$,
the parameters of the factor SV's ${\theta}_{f,1:K}$,
and the factor loading matrix ${\beta}$.
We define 
${\theta} :=\left({f}_{1:T},{\theta}_{\epsilon,1:S},{\theta}_{f,1:K},
{\beta}\right)$.
We use equations  \eqref{eq: lambda eqn for factor sv}
to specify the univariate particle filters that generate the 
factor log-volatilities ${\lambda}_{k,1:T}$ for $k=1,...,K$ without leverage and
\eqref{eq: h eqn for factor sv} to specify the particle filters for the
idiosyncratic SV log-volatilities with leverage ${h}_{s,1:T}$
for $s=1,...,S$.
We denote the $N$ weighted samples at time $t$ for the factor log volatilities
by $\left({\lambda}_{kt}^{1:N},\overline{w}_{fkt}^{1:N}\right)$
and $\left({h}_{st}^{1:N},\overline{w}_{\epsilon st}^{1:N}\right)$ for the idiosyncratic errors log volatilities.
The corresponding proposal densities are $m_{fk1}^{{\theta_{fk}}}\left(\lambda_{k1}\right)$,
$m_{fkt}^{{\theta_{fk}}}\left(\lambda_{kt}|\lambda_{kt-1}\right)$,
$m^{\theta_{\epsilon s}}_{\epsilon s1}\left(h_{s1}\right)$, and $m^{\theta_{\epsilon s}}_{\epsilon st}\left(h_{st}|h_{st-1}\right)$
for $t=2,...,T$ and are the same as in the univariate case.

We define $V_{\epsilon st}^{i}$ as a random
variable used to generate $h_{st}^{i}$ from the proposal $m_{\epsilon st}\left(h_{s1}\right)$ and $m_{\epsilon st}\left(h_{st}|h_{st-1}\right)$ for $ t \geq 2$,
and write $h_{st}^{i}=\mathfrak{X}\left(v_{\epsilon st}^{i}; \theta_{\epsilon s}, h_{s,t-1}\right)$
for $s=1,...,S$. The distribution of $v_{\epsilon st}^{i}$ is denoted
as $\psi_{\epsilon st}\left(\cdot\right)$ and is the standard normal
distribution $N\left(0,1\right)$. The random variable $V_{fkt}^{i}$
and its distribution $\psi_{fkt}\left(\cdot\right)$ are defined similarly.

We define $V_{A\epsilon st-1}^{i}$ for $s=1,...,S$, as the vector
of random variable used to generate the ancestor indices $A_{\epsilon st-1}^{i}$
for $i=1,...,N$, such that $A_{\epsilon st-1}^{1:N}$ is generated
using $\M_{\epsilon}\left(a_{\epsilon st-1}^{1:N}|\overline{w}_{\epsilon st-1}^{1:N},h_{\epsilon st-1}^{1:N}\right)$,
where each ancestor index $a_{\epsilon st-1}^{i}=j$ indexes a particle
in $\left({h}_{st-1}^{1:N},\overline{w}_{\epsilon st-1}^{1:N}\right)$
and is sampled with probability $\overline{w}_{\epsilon st-1}^{j}$.
We write the mapping $A_{\epsilon st-1}^{1:N}=\mathfrak{A}
\left(V_{A\epsilon st-1}^{1:N}; \overline{w}_{\epsilon st-1}^{1:N}, h_{st-1}^{1:N}
\right)$
and denote the distribution $V_{A\epsilon st}^{i}$  as $\psi_{A\epsilon st}\left(\cdot\right)$
(which is standard $U\left(0,1\right)$). The random variable $V_{Afkt-1}^{i}$,
its distribution $\psi_{Afkt}\left(\cdot\right)$, the resampling
scheme $\M_{f}\left(a_{fkt-1}^{1:N}|\overline{w}_{fkt-1}^{1:N},\lambda_{fkt-1}^{1:N}\right)$
and the mapping $A_{ f kt-1}^{1:N}=\mathfrak{A}
\left(V_{Af kt-1}^{1:N}; \overline{w}_{f kt-1}^{1:N}, \lambda_{kt-1}^{1:N}
\right)$  are defined similarly for $k=1,...,K$.

The joint distribution of the random variables $\left(V_{\epsilon s,1:T}^{1:N},V_{A\epsilon s,1:T-1}^{1:N}\right)$
is
\begin{align}
\psi\left(dV_{\epsilon s,1:T}^{1:N},dV_{A\epsilon s,1:T-1}^{1:N}\right) & =\prod_{t=1}^{T}\prod_{i=1}^{N}\psi_{\epsilon st}\left(dV_{\epsilon st}^{i}\right)\prod_{t=1}^{T-1}\prod_{i=1}^{N}\psi_{A\epsilon st}\left(dV_{A\epsilon st}^{i}\right),
\end{align}
for $s=1,...,S$ and the joint distribution of the random variables
$\left(V_{fk,1:T}^{1:N},V_{Afk,1:T-1}^{1:N}\right)$ is
\begin{align}
\psi\left(dV_{fk,1:T}^{1:N},dV_{Afk,1:T-1}^{1:N}\right) & =\prod_{t=1}^{T}\prod_{i=1}^{N}\psi_{fkt}\left(dV_{fkt}^{i}\right)\prod_{t=1}^{T-1}\prod_{i=1}^{N}\psi_{Afkt}\left(dV_{Afkt}^{i}\right),
\end{align}
for $k=1,...,K$.

We next define the indices $J_{\epsilon s,1:T}$ for
$s=1,...,S$, the selected particle trajectories ${h}_{s,1:T}^{j_{\epsilon s,1:T}}=\left(h_{s1}^{j_{\epsilon s1}},...,h_{sT}^{j_{\epsilon sT}}\right)$,
the indices $J_{fk,1:T}$ for $k=1,...,K$ and the selected particle trajectories
${\lambda}_{k,1:T}^{j_{fk,1:T}}=\left(\lambda_{k1}^{j_{fk1}},...,\lambda_{kT}^{j_{fkT}}\right)$.

The augmented target density in this case consists of all of the particle
filter variables $\left(V_{\epsilon s,1:T}^{1:N},V_{A\epsilon s,1:T-1}^{1:N}\right)$
and $\left(V_{fk,1:T}^{1:N},V_{Afk,1:T-1}^{1:N}\right)$, the sampled trajectory
$\left({\lambda}_{k,1:T}^{j_{fk,1:T}},{h}_{s,1:T}^{j_{\epsilon s,1:T}}\right)$
and indices $\left(J_{fk,1:T},J_{\epsilon s,1:T}\right)$ for all
$s=1,...,S$ and for $k=1,..,K$ and is
\begin{multline}
\widetilde{\pi}^{N}\left(dV_{\epsilon,1:S,1:T}^{1:N},dV_{A\epsilon,1:S,1:T-1}^{1:N},J_{\epsilon,1:S,1:T},dV_{f,1:K,1:T}^{1:N},dV_{Af,1:K,1:T-1}^{1:N},J_{f,1:K,1:T},d\theta\right)\coloneqq\\
\frac{\pi\left(d{\lambda}_{1:K,1:T}^{{J}_{f,1:T}},d{h}_{1:S,1:T}^{{J}_{S,1:T}},d{\theta}\right)}{N^{T\left(S+K\right)}}\times\prod_{s=1}^{S}\frac{\psi\left(dV_{\epsilon s,1:T}^{1:N},dV_{A\epsilon s,1:T-1}^{1:N}\right)}{m_{\epsilon s1}^{\theta}\left(dh_{s1}^{j_{\epsilon s1}}\right)\prod_{t=2}^{T}\overline{w}_{\epsilon st-1}^{a_{\epsilon st-1}^{j_{\epsilon st}}}m_{\epsilon st}^{\theta}\left(dh_{st}^{j_{\epsilon st}}|h_{st-1}^{a_{\epsilon st-1}^{j_{\epsilon st}}}\right)}\\
\times\prod_{t=2}^{T}\frac{w_{\epsilon st-1}^{a_{\epsilon st-1}^{j_{\epsilon st}}}f_{st}^{\theta}\left(h_{ st}^{j_{\epsilon s t}}|h_{ st-1}^{a_{\epsilon st-1}^{j_{\epsilon st}}}\right)}{\sum_{l=1}^{N}w_{\epsilon st-1}^{l}f_{st}^{\theta}\left(h_{st}^{j_{\epsilon s t}}|h_{st-1}^{l}\right)}\\ 
\times\prod_{k=1}^{K}\frac{\psi\left(dV_{fk,1:T}^{1:N},dV_{Afk,1:T-1}^{1:N}\right)}{m_{fk1}^{\theta}\left(d\lambda_{fk1}^{j_{fk1}}\right)\prod_{t=2}^{T}
\overline{w}_{fkt-1}^{a_{fkt-1}^{j_{fkt}}}
m_{ft}^{\theta}\left(d\lambda_{fkt}^{j_{fkt}}|\lambda_{fkt-1}^{a_{fkt-1}^{j_{fkt}}}\right)}\times\prod_{t=2}^{T}\frac{w_{fkt-1}^{a_{fkt-1}^{j_{fkt}}}f_{kt}^{\theta}\left(\lambda_{fkt}^{j_{fkt}}
|\lambda_{fkt-1}^{a_{fkt-1}^{j_{fkt}}}\right)}{\sum_{l=1}^{N}w_{fkt-1}^{l}f_{kt}^{\theta}\left(\lambda_{fkt}^{j_{fkt}}|\lambda_{fkt-1}^{l}\right)}.
\label{eq:Target distribution factor model}
\end{multline}

Similarly to the proof of Theorem~\ref{lemma: target distn}, we can show that the marginal distribution of $\widetilde{\pi}^{N}\left(dV_{\epsilon,1:S,1:T}^{1:N},dV_{A\epsilon,1:S,1:T-1}^{1:N},J_{\epsilon,1:S,1:T},dV_{f,1:K,1:T}^{1:N},dV_{Af,1:K,1:T-1}^{1:N},J_{f,1:K,1:T},d\theta\right)$
is
$
\( N^{T\left(S+K\right)}\)^{-1}\pi\left(d{\lambda}_{1:K,1:T}^{{J}_{f,1:T}},d{h}_{1:S,1:T}^{{J}_{S,1:T}},d{\theta}\right)
$

\section{The \corrPMMHPG{} for the factor SV model\label{subsec:Correlated-PMMH+PG-sampling}}

We illustrate our methods using the $\textrm{\corrPMMHPG}\left({\rho}_{\epsilon},{\tau}_{\epsilon}^{2},{\tau}_{f}^{2};{f}_{1:T},{\beta},{\mu}_{\epsilon},\phi_{\epsilon},\phi_{f}\right)$
for the factor SV with leverage, which 
performed well in the empirical studies 
in section~\ref{results multivariate examples}.
It is straightforward to modify the sampling scheme 
for other choices
of which parameters to sample with a PG step and which parameters
to sample with a MWG step.
Algorithm~\ref{alg:Sampling-Scheme:-The correlated PMMH+PG for 
factor SV model }
outlines the sampling scheme.


\begin{algorithm}[h!]
\caption{The \corrPMMHPG{} for the factor SV model with leverage:
$\textrm{\corrPMMHPG}\left({\rho}_{\epsilon},{\tau}_{\epsilon}^{2},{\tau}_{f}^{2};{f}_{1:T},{\beta},{\mu}_{\epsilon},\phi_{\epsilon},\phi_{f}\right)$
\label{alg:Sampling-Scheme:-The correlated PMMH+PG for factor SV model }}

Given initial values for $V_{\epsilon,1:S,1:T}^{1:N}$, $V_{A\epsilon,1:S,1:T-1}^{1:N}$,$V_{f,1:K,1:T}^{1:N}$,
$V_{Af,1:K,1:T-1}^{1:N}$ , $J_{\epsilon1:T}$, $J_{f1:T}$, and $\theta$

\begin{enumerate}
\item[Part 1:] MWG sampling,
For $k=1,...,K$
\begin{enumerate}
\item Sample $\tau_{fk}^{2*}\sim q_{\tau_{fk}^{2}}\left(\cdot|V_{fk,1:T}^{1:N},V_{Afk,1:T-1}^{1:N},\theta_{-\tau_{fk}^{2}},\tau_{fk}^{2}\right)$.
\item Run the SMC algorithm and obtain $\widehat{Z}\left(V_{fk,1:T}^{1:N},V_{Afk,1:T-1}^{1:N},\tau_{fk}^{2*},\theta_{-\tau_{fk}^{2}}\right)$.
\item Accept the proposed values $\tau_{fk}^{2*}$ with probability
\begin{align*}
1\land\frac{\widehat{Z}\left(V_{fk,1:T}^{1:N},V_{Afk,1:T-1}^{1:N},\tau_{fk}^{2*},\theta_{-\tau_{fk}^{2}}\right)p\left(\tau_{fk}^{2*}\right)}{\widehat{Z}\left(V_{fk,1:T}^{1:N},V_{Afk,1:T-1}^{1:N},\tau_{fk}^{2},\theta_{-\tau_{fk}^{2}}\right)p\left(\tau_{fk}^{2}\right)} & \times\frac{q_{\tau_{fk}^{2}}\left(\tau_{fk}^{2}|V_{fk,1:T}^{1:N},V_{Afk,1:T-1}^{1:N},\theta_{-\tau_{fk}^{2}},\tau_{fk}^{2*}\right)}{q_{\tau_{fk}^{2}}\left(\tau_{fk}^{2*}|V_{fk,1:T}^{1:N},V_{Afk,1:T-1}^{1:N},\theta_{-\tau_{fk}^{2}},\tau_{fk}^{2}\right)}.
\end{align*}
\end{enumerate}
For $s=1,....,S$,
\begin{enumerate}
\item Sample $\left(\tau_{\epsilon s}^{2*},\rho_{\epsilon s}^{*}\right)\sim q_{\tau_{\epsilon s}^{2},\rho_{\epsilon s}}\left(\cdot|V_{\epsilon s,1:T}^{1:N},V_{A\epsilon s,1:T-1}^{1:N},\tau_{\epsilon s}^{2},\rho_{\epsilon s},\theta_{-\tau_{\epsilon s}^{2},\rho_{\epsilon s}}\right)$.
\item Run the SMC algorithm and obtain $\widehat{Z}\left(V_{\epsilon s,1:T}^{1:N},V_{A\epsilon s,1:T-1}^{1:N},\tau_{\epsilon s}^{2*},\rho_{\epsilon s}^{*},\theta_{-\tau_{\epsilon s}^{2},\rho_{\epsilon s}}\right)$.
\item Accept the proposed values $\left(\tau_{\epsilon s}^{2*},\rho_{\epsilon s}^{*}\right)$
with probability
\begin{align*}
1\land\frac{\widehat{Z}\left(V_{\epsilon s,1:T}^{1:N},V_{A\epsilon s,1:T-1}^{1:N},\tau_{\epsilon s}^{2*},\rho_{\epsilon s}^{*},\theta_{-\tau_{\epsilon s}^{2},\rho_{\epsilon s}}\right)p\left(\tau_{\epsilon s}^{2*},\rho_{\epsilon s}^{*}\right)}{\widehat{Z}\left(V_{\epsilon s,1:T}^{1:N},V_{A\epsilon s,1:T-1}^{1:N},\tau_{\epsilon s}^{2},\rho_{\epsilon s},\theta_{-\tau_{\epsilon s}^{2},\rho_{\epsilon s}}\right)p\left(\tau_{\epsilon s}^{2},\rho_{\epsilon s}\right)} & \\
\times\frac{q_{\tau_{\epsilon s}^{2},\rho_{\epsilon s}}\left(\tau_{\epsilon s}^{2},\rho_{\epsilon s}|V_{\epsilon s,1:T}^{1:N},V_{A\epsilon s,1:T-1}^{1:N},\theta_{-\tau_{\epsilon s}^{2},\rho_{\epsilon s}},\tau_{\epsilon s}^{2*},\rho_{\epsilon s}^{*}\right)}{q_{\tau_{\epsilon s}^{2},\rho_{\epsilon s}}\left(\tau_{\epsilon s}^{2*},\rho_{\epsilon s}^{*}|V_{\epsilon s,1:T}^{1:N},V_{A\epsilon s,1:T-1}^{1:N},\theta_{-\tau_{\epsilon s}^{2},\rho_{\epsilon s}},\tau_{\epsilon s}^{2},\rho_{\epsilon s}\right)}.
\end{align*}
\end{enumerate}

\item [Part 2:] Sample $J_{\epsilon,1:S,1:T}\sim\widetilde{\pi}^{N}\left(\cdot|V_{\epsilon,1:S,1:T}^{1:N},V_{A\epsilon,1:S,1:T-1}^{1:N},{\theta}\right)$
and sample $J_{f,1:K,1:T}\sim\widetilde{\pi}^{N}\left(\cdot|V_{f,1:K,1:T}^{1:N},V_{Af,1:K,1:T-1}^{1:N},{\theta}\right)$.

\item[Part 3:] PG sampling.
\begin{enumerate}
\item Sample ${\beta}|{\lambda}_{1:T}^{{J}_{f,1:T}},{J}_{f,1:T},{h}_{1:T}^{{J}_{\epsilon,1:T}},{J}_{\epsilon,1:T},{\theta}_{-{\beta}},{y}_{1:T}$
using \eqref{eq:Bfactor-1}.
\item Redraw the diagonal elements of ${\beta}$ through the deep
interweaving procedure described in section~\ref{S: deep interweaving}.
\item Sample ${f}_{1:T}|{\lambda}_{1:T}^{{J}_{f,1:T}},{J}_{f,1:T},{h}_{1:T}^{{J}_{\epsilon,1:T}},
    {J}_{\epsilon,1:T},{\theta}_{-{f}_{1:T}},{y}_{1:T}$
using \eqref{eq:factordraws-1}.
\item Sample $\phi_{fk}|\lambda_{k1:T}^{j_{fk1:T}},j_{fk1:T},{\theta}_{-\phi_{fk}}$ for $k=1,...,K$.
\item Sample $\phi_{\epsilon s}|h_{s1:T}^{j_{\epsilon s1:T}},j_{\epsilon s1:T},{\theta}_{-\phi_{\epsilon s}}$
and sample $\mu_{\epsilon s}|h_{s1:T}^{j_{\epsilon s1:T}},j_{\epsilon s1:T},{\theta}_{-\mu_{\epsilon s}}$ for $s=1,...,S$.
\end{enumerate}
\end{enumerate}
\end{algorithm}

\begin{algorithm*}
This is a continuation of Algorithm~\ref{alg:Sampling-Scheme:-The correlated PMMH+PG for factor SV model }
\begin{enumerate}
\item[Part 4:] Sample $\left(V_{fk,1:T}^{1:N},V_{Afk,1:T-1}^{1:N}\right)$ from $\widetilde{\pi}^{N}\left(\cdot|\lambda_{k1:T}^{j_{fk1:T}},j_{fk1:T},{\theta}\right)$
 using CCSMC 
(Algorithm~\ref{alg:The-conditional Sequential-Monte carlo algorithm}) and obtain  $\widehat{Z}\left(V_{fk,1:T}^{1:N},V_{Afk,1:T-1}^{1:N},\tau_{fk}^{2},\theta_{-\tau_{fk}^{2}}\right)$ for $k=1,...,K$.
\pagebreak
\item [Part 5:] Sample $\left(V_{\epsilon s,1:T}^{1:N},V_{A\epsilon s,1:T-1}^{1:N}\right)$
from $\widetilde{\pi}^{N}\left(\cdot|h_{s1:T}^{j_{\epsilon s1:T}},j_{\epsilon s1:T},{\theta}\right)$
 using the CCSMC algorithm
(Algorithm~\ref{alg:The-conditional Sequential-Monte carlo algorithm}) and obtain
$\widehat{Z}\left(V_{\epsilon s,1:T}^{1:N},V_{A\epsilon s,1:T-1}^{1:N},\tau_{\epsilon s}^{2},\rho_{\epsilon s},\theta_{-\tau_{\epsilon s}^{2},\rho_{\epsilon s}}\right)$ for $s=1,...,S$.
\end{enumerate}
\end{algorithm*}

\clearpage
\subparagraph{Further discussion of part~3, (d) and (e) of  Algorithm~\ref{alg:Sampling-Scheme:-The correlated PMMH+PG for factor SV model }}

For $k=1,...,K$, sample the autoregressive coefficient $\phi_{fk}$
from $\widetilde{\pi}^{N}\left(\cdot|\lambda_{k1:T}^{j_{fk1:T}},j_{fk1:T},{\theta}_{-\phi_{fk}}\right)$.
We draw a proposed value $\phi_{fk}^{*}$ from $N\left(\mu_{\phi_{fk}},\sigma_{\phi_{fk}}^{2}\right)$
truncated within $\left(-1,1\right)$, where
\begin{equation}
\mu_{\phi_{fk}}=\frac{\sigma_{\phi_{fk}}^{2}}{\tau_{fk}^{2}}\sum_{t=2}^{T}\lambda_{kt}\lambda_{kt-1},\;\sigma_{\phi_{fk}}^{2}=\frac{\tau_{fk}^{2}}{\sum_{t=2}^{T-1}\lambda^2_{kt}}.
\end{equation}
The candidate is accepted with probability
\begin{equation}
\min\left\{1,  \frac{p\left(\phi_{fk}^{*}\right)\sqrt{1-\phi_{fk}^{2*}}}{p\left(\phi_{fk}\right)\sqrt{1-\phi_{fk}^{2}}}\right\} .
\end{equation}

For $s=1,...,S$, sample $\mu_{\epsilon s}$ from $N\left(\mu_{\mu_{\epsilon s}},\sigma_{\mu_{\epsilon s}}^{2}\right)$,
where
\begin{equation}
\mu_{\mu_{\epsilon s}}=\sigma_{\mu_{\epsilon s}}^{2}\frac{h_{s1}\left(1-\phi_{\epsilon s}^{2}\right)\left(1-\rho_{\epsilon s}^{2}\right)+\left(1-\phi_{\epsilon s}\right)\sum_{t=2}^{T}h_{st}-\phi h_{st-1}-\rho_{\epsilon s}\tau_{\epsilon s}\epsilon_{st-1}^{*}}{\tau_{\epsilon s}^{2}\left(1-\rho_{\epsilon s}^{2}\right)},
\end{equation}
 $\epsilon_{st-1}^{*}=\left(y_{st-1}-{\beta}_{s}f_{t-1}\right)\exp\left(-h_{st-1}/2\right)$
and
\begin{equation}
\sigma_{\mu_{\epsilon s}}^{2}=\frac{\tau_{\epsilon s}^{2}\left(1-\rho_{\epsilon s}^{2}\right)}{\left(1-\phi_{\epsilon s}^{2}\right)\left(1-\rho_{\epsilon s}^{2}\right)+\left(T-1\right)\left(1-\phi_{\epsilon s}\right)^{2}}.
\end{equation}
For $s=1,...,S$, sample $\phi_{\epsilon s}$, by drawing a proposed
value $\phi_{\epsilon s}^{*}$ from $N\left(\mu_{\phi_{\epsilon s}},\sigma_{\phi_{\epsilon s}}^{2}\right)$
truncated within $\left(-1,1\right)$, where
\begin{equation}
\mu_{\phi_{\epsilon s}}=\frac{\sum_{t=2}^{T}\left(h_{s,t}-\mu_{\epsilon s}\right)\left(h_{st-1}-\mu_{\epsilon s}\right)-\rho_{\epsilon s}\tau_{\epsilon s}\left(h_{st-1}-\mu_{\epsilon s}\right)\epsilon_{st-1}^{*}}{\sum_{t=2}^{T}\left(h_{st-1}-\mu_{\epsilon s}\right)^{2}-\left(h_{s1}-\mu_{\epsilon s}\right)^{2}\left(1-\rho_{\epsilon s}^{2}\right)},
\end{equation}
and
\[
\sigma_{\phi_{\epsilon s}}^{2}=\frac{\tau_{\epsilon s}^{2}\left(1-\rho_{\epsilon s}^{2}\right)}{\sum_{t=2}^{T}\left(h_{st-1}-\mu_{\epsilon s}\right)^{2}-\left(h_{s1}-\mu_{\epsilon s}\right)^{2}\left(1-\rho_{\epsilon s}^{2}\right)}.
\]
The candidate is accepted with probability
\begin{equation}
\min\left\{ \frac{p\left(\phi_{\epsilon s}^{*}\right)\sqrt{1-\phi_{\epsilon s}^{2*}}}{p\left(\phi_{\epsilon s}\right)\sqrt{1-\phi_{\epsilon s}^{2}}},1\right\} .
\end{equation}
In all the examples, the MWG step uses the bootstrap filter to evaluate
the particles and the adaptive random walk as the proposal density
for the parameters.

\section{The factor loading matrix and the latent factors\label{S: sampling 
factor loading matrix} }
This section discusses the parameterization and sampling of the factor loading matrix and the factors.

To identify the parameters of the factor loading matrix
$\boldsymbol{\beta}$, it is necessary to impose some further constraints.
The factor loading matrix $\boldsymbol{\beta}$ is usually assumed to be
lower triangular, i.e.,  $\beta_{sk}=0$ for $k>s$;
furthermore, one of two constraints are used: i) 
the $\boldsymbol{f}_{kt}$ have unit variance \citep{Geweke1996}; or 
ii) $\beta_{ss}=1$, for $s=1,...,S$,
and the variance of $\boldsymbol{f}_{t}$ is diagonal but unconstrained.
The main drawback of the lower triangular assumption on $\boldsymbol{\beta}$
is that the resulting inference can depend on the order in which the
components of $y_{t}$ are chosen \citep{Chan:2017}.
We use the following approach for $K = m$ factors to obtain an appropriate ordering of the returns that does not conflict with
the  data.  We follow  \cite{Conti2014} and \cite{Kastner:2017} 
and run and post-process the draws from the unrestricted
sampler  by choosing from column 1 the stock $i=i_1$ with the largest value of $|\beta_{i,1}|$.
 We repeat this for column 2, except that  now we seek that $i=2,\dots, S, i \neq i_1$ maximizing $| \beta_{i,2}|$.
 We proceed similarly for columns 3 to $m$. By an unrestricted sampler we mean that we do not restrict $
 \bs \beta$ to be lower triangular.
Furthermore, as noted by \citet{Kastner:2017}, the second
set of constraints impose that the first $K$ variables are the leading
factors, making the variable ordering dependence stronger.
We follow \citet{Kastner:2017} and leave the diagonal elements $\beta_{ss}$
unrestricted and set the level $\mu_{2k}$ of the factor log-volatilities
$\lambda_{kt}$ to zero for $k=1,...,K$.

Let $k_{s}$ denote the number of unrestricted elements in row $s$
of $\beta$ and define
\[
F_{s}=\left[\begin{array}{ccc}
f_{11} & \cdots & f_{k_{s}1}\\
\vdots &  & \vdots\\
f_{1T} & \cdots & f_{k_{s}T}
\end{array}\right],\;\widetilde{V}_{s}=\left[\begin{array}{ccc}
\exp\left(h_{s1}\right) & \cdots & 0\\
0 & \ddots & \vdots\\
0 & \cdots & \exp\left(h_{sT}\right)
\end{array}\right].
\]
Then, the factor loadings $\beta_{s,.}=\left(\beta_{s1},...,\beta_{sk_{s}}\right)^{\transp}$
for $s=1,...,S$, are sampled independently for each $s$ by performing a Gibbs-update using
\begin{equation}
\beta_{s,.}^{\transp}|f_{1:T},y_{s,1:T},h_{s,.}\sim N_{k_{s}}\left(a_{sT},b_{sT}\right),\label{eq:Bfactor-1}
\end{equation}
where $b_{pT}=\left(\left(F_{S}^{T}\widetilde{V}_{S}^{-1}F_{S}\right)+I_{k_{s}}\right)^{-1}$
and $a_{sT}=b_{sT}F_{s}^{T}\left(\widetilde{V}_{s}^{-1}y_{s,1:T}\right).$

\subparagraph{Sampling $\left\{ {f}_{t}\right\} |{y},\left\{ {h}_{t}\right\} ,\left\{ {\lambda}_{t}\right\} ,{\beta}$}
After  some algebra, we can show that $\left\{ f_{t}\right\} $
can be sampled from
\begin{equation}
\left\{ {f}_{t}\right\} |{y},\left\{ {h}_{t}\right\} ,\left\{ {\lambda}_{t}\right\} , {\beta}\sim N\left(a_{t},b_{t}\right),\label{eq:factordraws-1}
\end{equation}
where $b_{t}=\left({\beta}^{T}{V}_{t}^{-1}{\beta}+ {D}_{t}^{-1}\right)^{-1}$
and $a_{t}=b_{t} {\beta}^{T}\left( {V}_{t}^{-1}{y}_{t}\right).$

\section{Deep Interweaving\label{S: deep interweaving}}

It is well-known that sampling the factor loading matrix ${\beta}$
conditional on $\left\{ {f}_{t}\right\} $ and then sampling
$\left\{ {f}_{t}\right\} $ conditional on ${\beta}$
is inefficient and leads to extremely slow convergence and poor
mixing. We use an approach based on an 
ancillarity-sufficiency
interweaving strategy (ASIS), and in particular the deep interweaving strategy,
introduced by \citet{Kastner:2017}, that is now briefly described. The parameterisation underlying deep interweaving
is given by
\begin{equation}
{y}_{t}={\beta}^{*}{f}_{t}^{*}+{V}_{t}^{\frac{1}{2}}\varepsilon_{t},\;\;{f}_{t}^{*}|{\lambda}^{*}\sim N_{K}\left(0,\textrm{diag}\left(e^{\lambda_{1t}^{*}},...,e^{\lambda_{Kt}^{*}}\right)\right),\label{eq:deep interweaving factor model}
\end{equation}
with a lower triangular factor loading matrix ${\beta}^{*}$,
where ${\beta}_{11}^{*}=1,...,{\beta}_{KK}^{*}=1$.
The factor model can be reparameterised in \eqref{eq:deep interweaving factor model}
using the linear transformation
\[
{f}_{t}^{*}=D{f}_{t},{\beta}^{*}={\beta}D^{-1},
\]
where $D=\textrm{diag}\left(\beta_{11},...,\beta_{KK}\right)$, for
$t=1,..,T$. The $K$ latent factor volatilities $\lambda_{kt}^{*}$
follow the following  univariate SV models having levels $\mu_{fk}=\log\beta_{kk}^{2}$,
rather than zero, as in the factor SV model. The transformed
factor volatilities are  given by
\[
\lambda_{kt}^{*}=\lambda_{kt}+\log\beta_{kk}^{2},\quad t=0,...,T,\,\,\, k=1,...,K.
\]
We add the following deep interweaving algorithm in between sampling
the factor loading matrix and sampling the latent factors and perform these
steps independently for each $k=1,..,K$,
\begin{itemize}
\item Determine the vector ${\beta}_{.,k}^{*}$, where $\beta_{sk}^{*}=\beta_{sk}^{old}/\beta_{kk}^{old}$
in the $k$th column of the transformed factor loading matrix ${\beta}^{*}$.
\item Define ${\lambda}_{k,1:T}^{*}={\lambda}_{k,1:T}^{old}+2\log|\beta_{kk}^{old}|$
and sample $\beta_{kk}^{new}$ from $p\left(\beta_{kk}|{\beta}_{.,k}^{*},{\lambda}_{k,.}^{*},\phi_{fk},\tau_{fk}^{2}\right)$
for factor log-volatilites follows SV process.
\item Update ${\beta}_{.,k}=\frac{\beta_{kk}^{new}}{\beta_{kk}^{old}}{\beta}_{.,k}^{old}$,
${f}_{k,.}=\frac{\beta_{kk}^{old}}{\beta_{kk}^{new}}{f}_{k,.}^{old}$,
and ${\lambda}_{k,1:T}={\lambda}_{k,1:T}^{old}+2\log|\frac{\beta_{kk}^{old}}{\beta_{kk}^{new}}|$.

In the deep interweaving representation, we sample the scaling parameter
$\beta_{kk}$ indirectly through $\mu_{fk}$, $k=1,...,K$. The implied
prior $p\left(\mu_{fk}\right)\propto\exp\left(\mu_{fk}/2-\exp\left(\mu_{fk}\right)/2\right)$,
the density $p\left({\beta}_{.,k}^{*}|\mu_{fk}\right)\sim N\left(0,\exp\left(-\mu_{fk}\right)I_{k_{l}}\right)$
so that
\[
p\left(\mu_{fk}|{\beta}_{.,k}^{*},{\lambda}_{k,.}^{*},\phi_{fk},\tau_{fk}^{2}\right)\propto p\left({\lambda}_{fk,.}^{*}|\mu_{fk},\phi_{fk},\tau_{fk}^{2}\right)p\left({\beta}_{.,k}^{*}|\mu_{fk}\right)p\left(\mu_{fk}\right),
\]
which is not in an easily recognisable form for sampling. Instead, we draw the proposal $\mu_{fk}^{prop}$
from $N\left(A,B\right)$, where
\[
A=\frac{\sum_{t=2}^{T-1}\lambda_{kt}^{*}+\left(\lambda_{kT}^{*}-\phi_{fk}\lambda_{k1}\right)/\left(1-\phi_{fk}\right)}{T-1+1/B_{0}},
\quad B=\frac{\tau_{fk}^{2}/\left(1-\phi_{fk}\right)^{2}}{T-1+1/B_{0}}.
\]
Denoting the current value $\mu_{fk}$ by $\mu_{fk}^{old}$, the new
value $\mu_{fk}^{prop}$ gets accepted with probability $\min\left(1,R\right)$,
where
\[
R=\frac{p\left(\mu_{fk}^{prop}\right)p\left(\lambda_{k1}^{*}|\mu_{fk}^{prop},\phi_{fk},\tau_{fk}^{2}\right)p\left({\beta}_{.,k}^{*}|\mu_{fk}^{prop}\right)}{p\left(\mu_{fk}^{old}\right)\left(\lambda_{k,1}^{*}|\mu_{fk}^{old},\phi_{fk},\tau_{fk}^{2}\right)p\left({\beta}_{.,k}^{*}|\mu_{fk}^{old}\right)}\times\frac{p_{aux}\left(\mu_{fk}^{old}|\phi_{fk},\tau_{fk}^{2}\right)}{p_{aux}\left(\mu_{fk}^{prop}|\phi_{fk},\tau_{fk}^{2}\right)},
\]
where
\[
p_{aux}\left(\mu_{fk}^{old}|\phi_{fk},\tau_{fk}^{2}\right)\sim N\left(0,B_{0}\tau_{fk}^{2}/\left(1-\phi_{fk}\right)^{2}\right).
\]

We follow \citet{Kastner:2017} and set the constant $B_{0}$  to the large value $10^{5}$.
\end{itemize}

\section{Further empirical results for the factor  SV model with leverage\label{S: further empirical results for factor SV model}}

This section gives further empirical results for the multivariate factor SV models discussed in section \ref{Multivariate example}. 

Tables \ref{tab:IACT T=00003D1000 N=00003D100} to \ref{tab:Inefficiency-factor-(IACT) T1000 N1000} 
show the mean and maximum IACT values for each parameter in
the factor SV model with leverage for $T=1000$ observations, $S=26$ stock returns, and $K=1$ factor for the \corrPMMHPG{}, the
\PHS{}, the PGBS and the refined PGDA samplers.
Tables~\ref{tab:Inefficiency-factor-(IACT) N500 T3000} to \ref{tab:Inefficiency-factor-(IACT) N2000 T3001}
show the mean and maximum IACT values for each parameter in
the factor SV model with leverage for $T=3000$ observations, $S=26$ stock returns, $K=4$ factors,
for the \corrPMMHPG{}, the
\PHS{}, the PGBS and the refined PGDA samplers. Tables \ref{tab:Inefficiency-factor-(IACT) N500 T3000-diffusion} 
and \ref{tab:Inefficiency-factor-(IACT) N1000 T3000-diffusion-1} report the IACT estimates for 
all the parameters for the factor SV model with the 
idiosyncratic log-volatilities following GARCH diffusion models for the CPHS, the PHS, and the PG sampler. The tables in this section also show that CPHS generally outperforms the competing samplers. Full discussions are given in section \ref{results multivariate examples}.

\begin{table}[H]
\caption{Inefficiency factor (IACT) of the parameters of the factor SV model
with leverage for US stock return data with $T=1000$ observations, $S=26$ stock returns, and
$K=1$ factor.
Sampler I: $\textrm{\corrPMMHPG}\left({\rho}_{\epsilon},{\tau}_{\epsilon}^{2},{\tau}_{f}^{2};{f}_{1:T},{\beta},{\mu}_{\epsilon},\phi_{\epsilon},\phi_{f}\right)$,
Sampler II: $\textrm{PHS \ensuremath{\left(\tau_{f}^{2},\tau_{\epsilon}^{2},\rho_{\epsilon};{f}_{1:T}, \mu_{\epsilon},\phi_{\epsilon},\phi_{f},\beta\right)}}$,
Sampler III: $\textrm{PGBS}\left({f}_{1:T}, \mu_{\epsilon},\phi_{\epsilon},\phi_{f},\beta,\tau_{f}^{2},\tau_{\epsilon}^{2}\right)$,
and Sampler IV: $\textrm{PGDA}\left({f}_{1:T}, \mu_{\epsilon},\phi_{\epsilon},\phi_{f},\beta,\tau_{f}^{2},\tau_{\epsilon}^{2}\right)$
with $N=100$ particles.
The table shows the mean and maximum IACT values for the
parameters and the factor and idiosyncratic log-volatilities. The IACT is defined in section \ref{SS: preliminaries}.\label{tab:IACT T=00003D1000 N=00003D100}}

\centering{}%
\begin{tabular}{ccccccccc}
\hline 
 & \multicolumn{2}{c}{I } & \multicolumn{2}{c}{II} & \multicolumn{2}{c}{III} & \multicolumn{2}{c}{IV}\tabularnewline
\cline{2-9} \cline{3-9} \cline{4-9} \cline{5-9} \cline{6-9} \cline{7-9} \cline{8-9} \cline{9-9} 
 & Mean & Max & Mean & Max & Mean & Max & Mean & Max\tabularnewline
\hline 
$\beta_{1}$ & 1.70 & 1.77 & 1.68 & 1.72 & 1.82 & 1.86 & 8.64 & 9.56\tabularnewline
$\mu$ & 2.43 & 6.89 & 3.18 & 10.99 & 3.65 & 9.16 & 614.16 & 3536.65\tabularnewline
$\tau^{2}$ & 39.69 & 63.46 & 383.03 & 4454.61 & 709.27 & 1966.36 & 712.94 & 3432.51\tabularnewline
$\phi$ & 37.55 & 69.00 & 66.18 & 364.48 & 192.76 & 687.80 & 947.72 & 6052.27\tabularnewline
$\rho$ & 14.49 & 29.60 & 240.43 & 1638.51 & 264.94 & 499.54 & 420.50 & 1274.57\tabularnewline
$h_{1,1:T}$ & 2.30 & 11.32 & 2.82 & 11.39 & 3.85 & 30.22 & 23.14 & 110.68\tabularnewline
$h_{2,1:T}$ & 2.03 & 8.22 & 2.46 & 10.35 & 3.05 & 30.73 & 35.14 & 172.40\tabularnewline
$h_{10,1:T}$ & 4.12 & 24.54 & 4.31 & 17.42 & 3.67 & 22.23 & 23.52 & 99.98\tabularnewline
$h_{11,1:T}$ & 1.99 & 8.34 & 3.06 & 12.49 & 7.46 & 74.66 & 45.47 & 276.13\tabularnewline
$h_{12,1:T}$ & 1.62 & 7.40 & 3.46 & 21.08 & 2.30 & 12.26 & 25.53 & 277.76\tabularnewline
$\lambda_{1,1:T}$ & 1.46 & 5.27 & 1.42 & 6.27 & 1.69 & 10.89 & 40.15 & 289.65\tabularnewline
\hline 
\end{tabular}
\end{table}

\begin{table}[H]
\caption{Inefficiency factor (IACT) of the parameters of the factor SV model
with leverage for US stock return data with $T=1000$ observations, $S=26$ stock returns, and
$K=1$ factor. Sampler I: $\textrm{\corrPMMHPG}\left({\rho}_{\epsilon},{\tau}_{\epsilon}^{2},{\tau}_{f}^{2};{f}_{1:T},{\beta},{\mu}_{\epsilon},\phi_{\epsilon},\phi_{f}\right)$ with $N=100$, Sampler II: $\textrm{PHS \ensuremath{\left(\tau_{f}^{2},\tau_{\epsilon}^{2},\rho_{\epsilon};{f}_{1:T}, \mu_{\epsilon},\phi_{\epsilon},\phi_{f},\beta\right)}}$,
Sampler III: $\textrm{PGBS}\left({f}_{1:T},\mu_{\epsilon},\phi_{\epsilon},\phi_{f},\beta,\tau_{f}^{2},\tau_{\epsilon}^{2}\right)$,
and Sampler IV: $\textrm{PGDA}\left({f}_{1:T},\mu_{\epsilon},\phi_{\epsilon},\phi_{f},\beta,\tau_{f}^{2},\tau_{\epsilon}^{2}\right)$
with $N=250$ particles. The table shows the mean and maximum IACT values for the
parameters and the factor and idiosyncratic log-volatilities. The IACT is defined in section \ref{SS: preliminaries}. \label{tab:Inefficiency-factor-(IACT) T=00003D1000 N250}}

\centering{}%
\begin{tabular}{ccccccccc}
\hline 
 & \multicolumn{2}{c}{I} & \multicolumn{2}{c}{II} & \multicolumn{2}{c}{III } & \multicolumn{2}{c}{IV}\tabularnewline
\cline{2-9} \cline{3-9} \cline{4-9} \cline{5-9} \cline{6-9} \cline{7-9} \cline{8-9} \cline{9-9} 
 & Mean & Max & Mean & Max & Mean & Max & Mean & Max\tabularnewline
\hline 
$\beta_{1}$ & 1.70 & 1.77 & 1.86 & 1.89 & 1.73 & 1.75 & 5.66 & 6.38\tabularnewline
$\mu$ & 2.43 & 6.89 & 2.64 & 5.49 & 3.55 & 10.18 & 47.73 & 195.85\tabularnewline
$\tau^{2}$ & 39.69 & 63.46 & 62.16 & 284.63 & 753.48 & 2158.37 & 103.36 & 365.68\tabularnewline
$\phi$ & 37.55 & 69.00 & 44.71 & 103.56 & 230.30 & 476.30 & 188.84 & 1111.59\tabularnewline
$\rho$ & 14.49 & 29.60 & 50.95 & 493.83 & 292.73 & 898.86 & 54.43 & 259.50\tabularnewline
$h_{1,1:T}$ & 2.30 & 11.32 & 2.32 & 11.56 & 8.88 & 62.02 & 6.47 & 14.90\tabularnewline
$h_{2,1:T}$ & 2.03 & 8.22 & 1.98 & 7.24 & 3.87 & 37.82 & 9.25 & 25.22\tabularnewline
$h_{10,1:T}$ & 4.12 & 24.54 & 3.58 & 17.34 & 4.28 & 39.29 & 7.53 & 15.85\tabularnewline
$h_{11,1:T}$ & 1.99 & 8.34 & 2.15 & 8.810.94 & 7.94 & 96.06 & 8.73 & 33.27\tabularnewline
$h_{12,1:T}$ & 1.62 & 7.40 & 1.88 & 7.75 & 2.80 & 16.70 & 6.81 & 43.58\tabularnewline
$\lambda_{1,1:T}$ & 1.46 & 5.27 & 1.34 & 4.12 & 1.69 & 6.79 & 7.75 & 21.07\tabularnewline
\hline 
\end{tabular}
\end{table}

\begin{table}[H]
\caption{Inefficiency factor (IACT) of the parameters of the factor SV model
with leverage for US stock return data with $T=1000$ observations, $S=26$ stock returns, and
$K=1$ factor. Sampler I: $\textrm{\corrPMMHPG}\left({\rho}_{\epsilon},{\tau}_{\epsilon}^{2},{\tau}_{f}^{2};{f}_{1:T},{\beta},{\mu}_{\epsilon},\phi_{\epsilon},\phi_{f}\right)$ with $N=100$, Sampler II: $\textrm{PHS \ensuremath{\left(\tau_{f}^{2},\tau_{\epsilon}^{2},\rho_{\epsilon};{f}_{1:T}, \mu_{\epsilon},\phi_{\epsilon},\phi_{f},\beta\right)}}$,
Sampler III: $\textrm{PGBS}\left({f}_{1:T},\mu_{\epsilon},\phi_{\epsilon},\phi_{f},\beta,\tau_{f}^{2},\tau_{\epsilon}^{2}\right)$,
and Sampler IV: $\textrm{PGDA}\left({f}_{1:T},\mu_{\epsilon},\phi_{\epsilon},\phi_{f},\beta,\tau_{f}^{2},\tau_{\epsilon}^{2}\right)$with
$N=500$ particles. The table shows the mean and maximum IACT values for the parameters
and the factor and idiosyncratic log-volatilities. The IACT is defined in section \ref{SS: preliminaries}. \label{tab:Inefficiency-factor-(IACT) T=00003D1000 N500}}

\centering{}%
\begin{tabular}{ccccccccc}
\hline 
 & \multicolumn{2}{c}{I} & \multicolumn{2}{c}{II} & \multicolumn{2}{c}{III } & \multicolumn{2}{c}{IV}\tabularnewline
\cline{2-9} \cline{3-9} \cline{4-9} \cline{5-9} \cline{6-9} \cline{7-9} \cline{8-9} \cline{9-9} 
 & Mean & Max & Mean & Max & Mean & Max & Mean & Max\tabularnewline
\hline 
$\beta_{1}$ & 1.70 & 1.77 & 1.74 & 1.84 & 1.77 & 1.89 & 2.69 & 2.73\tabularnewline
$\mu$ & 2.43 & 6.89 & 2.64 & 7.51 & 3.32 & 9.71 & 21.80 & 54.10\tabularnewline
$\tau^{2}$ & 39.69 & 63.46 & 45.15 & 88.33 & 675.79 & 2016.98 & 42.00 & 121.80\tabularnewline
$\phi$ & 37.55 & 69.00 & 41.33 & 72.59 & 208.45 & 451.59 & 57.57 & 192.41\tabularnewline
$\rho$ & 14.49 & 29.60 & 19.12 & 51.10 & 284.35 & 703.52 & 20.61 & 60.60\tabularnewline
$h_{1,1:T}$ & 2.30 & 11.32 & 2.19 & 10.46 & 7.12 & 77.47 & 4.45 & 10.43\tabularnewline
$h_{2,1:T}$ & 2.03 & 8.22 & 1.78 & 6.54 & 3.96 & 41.94 & 4.85 & 10.92\tabularnewline
$h_{10,1:T}$ & 4.12 & 24.54 & 3.11 & 13.18 & 2.60 & 13.39 & 3.91 & 7.51\tabularnewline
$h_{11,1:T}$ & 1.99 & 8.34 & 2.19 & 9.01 & 3.90 & 34.44 & 4.88 & 16.06\tabularnewline
$h_{12,1:T}$ & 1.62 & 7.40 & 1.68 & 6.34 & 2.36 & 13.52 & 3.78 & 22.50\tabularnewline
$\lambda_{1,1:T}$ & 1.46 & 5.27 & 1.27 & 2.67 & 1.66 & 5.95 & 4.76 & 12.81\tabularnewline
\hline 
\end{tabular}
\end{table}

\begin{table}[H]
\caption{Inefficiency factor (IACT) of the parameters of the factor SV model
with leverage for US stock return data with $T=1000$ observations, $S=26$ stock returns, and
$K=1$ factor. Sampler I: $\textrm{\corrPMMHPG}\left({\rho}_{\epsilon},{\tau}_{\epsilon}^{2},{\tau}_{f}^{2};{f}_{1:T},{\beta},{\mu}_{\epsilon},\phi_{\epsilon},\phi_{f}\right)$ with $N=100$, Sampler II: $\textrm{PHS \ensuremath{\left(\tau_{f}^{2},\tau_{\epsilon}^{2},\rho_{\epsilon};{f}_{1:T}, \mu_{\epsilon},\phi_{\epsilon},\phi_{f},\beta\right)}}$,
Sampler III: $\textrm{PGBS}\left({f}_{1:T},\mu_{\epsilon},\phi_{\epsilon},\phi_{f},\beta,\tau_{f}^{2},\tau_{\epsilon}^{2}\right)$,
and Sampler IV: $\textrm{PGDA}\left({f}_{1:T},\mu_{\epsilon},\phi_{\epsilon},\phi_{f},\beta,\tau_{f}^{2},\tau_{\epsilon}^{2}\right)$
with $N=1000$ particles. The table shows the mean and maximum IACT values for the
parameters and the factor and idiosyncratic log-volatilities. The IACT is defined in section \ref{SS: preliminaries}.\label{tab:Inefficiency-factor-(IACT) T1000 N1000}}

\centering{}%
\begin{tabular}{ccccccccc}
\hline 
 & \multicolumn{2}{c}{I} & \multicolumn{2}{c}{II} & \multicolumn{2}{c}{III} & \multicolumn{2}{c}{IV}\tabularnewline
\cline{2-9} \cline{3-9} \cline{4-9} \cline{5-9} \cline{6-9} \cline{7-9} \cline{8-9} \cline{9-9}  
 & Mean & Max & Mean & Max & Mean & Max & Mean & Max\tabularnewline
\hline 
$\beta_{1}$ & 1.70 & 1.77 & 1.71 & 1.78 & 1.79 & 1.86 & 2.44 & 2.54\tabularnewline
$\mu$ & 2.43 & 6.89 & 2.27 & 6.69 & 3.54 & 11.06 & 12.52 & 34.59\tabularnewline
$\tau^{2}$ & 39.69 & 63.46 & 41.15 & 73.33 & 731.78 & 1781.30 & 26.91 & 91.04\tabularnewline
$\phi$ & 37.55 & 69.00 & 38.19 & 75.54 & 219.92 & 596.81 & 34.65 & 98.03\tabularnewline
$\rho$ & 14.49 & 29.60 & 15.06 & 26.79 & 252.63 & 424.56 & 11.36 & 24.36\tabularnewline
$h_{1,1:T}$ & 2.30 & 11.32 & 2.01 & 9.29 & 4.30 & 43.27 & 3.01 & 7.42\tabularnewline
$h_{2,1:T}$ & 2.03 & 8.22 & 1.67 & 5.62 & 5.51 & 65.39 & 3.09 & 6.78\tabularnewline
$h_{10,1:T}$ & 4.12 & 24.54 & 2.94 & 12.91 & 6.54 & 63.92 & 2.78 & 6.43\tabularnewline
$h_{11,1:T}$ & 1.99 & 8.34 & 1.85 & 6.96 & 5.57 & 51.36 & 3.38 & 9.40\tabularnewline
$h_{12,1:T}$ & 1.62 & 7.40 & 1.46 & 4.58 & 2.03 & 12.49 & 2.26 & 9.80\tabularnewline
$\lambda_{1,1:T}$ & 1.46 & 5.27 & 1.23 & 2.34 & 1.78 & 6.25 & 3.12 & 6.74\tabularnewline
\hline 
\end{tabular}
\end{table}

\begin{table}[H]
\caption{Inefficiency factor (IACT) of the parameters of the factor SV model
with leverage for US stock return data with $T=3001$ observations, $S=26$ stock returns, and
$K=4$ factors. Sampler I: $\textrm{\corrPMMHPG}\left({\rho}_{\epsilon},{\tau}_{\epsilon}^{2},{\tau}_{f}^{2};{f}_{1:T},{\beta},{\mu}_{\epsilon},\phi_{\epsilon},\phi_{f}\right)$
with $N=100$ particles, Sampler II: $\textrm{PHS \ensuremath{\left(\tau_{f}^{2},\tau_{\epsilon}^{2},\rho_{\epsilon};{f}_{1:T}, \mu_{\epsilon},\phi_{\epsilon},\phi_{f},\beta\right)}}$,
Sampler III: $\textrm{PGBS}\left({f}_{1:T},\mu_{\epsilon},\phi_{\epsilon},\phi_{f},\beta,\tau_{f}^{2},\tau_{\epsilon}^{2}\right)$,
and Sampler IV: $\textrm{PGDA}\left({f}_{1:T},\mu_{\epsilon},\phi_{\epsilon},\phi_{f},\beta,\tau_{f}^{2},\tau_{\epsilon}^{2}\right)$
with $N=500$ particles. The table shows the mean and maximum IACT values for the
parameters and the factor and idiosyncratic log-volatilities. The IACT is defined in section \ref{SS: preliminaries}. \label{tab:Inefficiency-factor-(IACT) N500 T3000}}

\centering{}%
\begin{tabular}{ccccccccc}
\hline 
 & \multicolumn{2}{c}{I } & \multicolumn{2}{c}{II } & \multicolumn{2}{c}{III } & \multicolumn{2}{c}{IV}\tabularnewline
\cline{2-9} \cline{3-9} \cline{4-9} \cline{5-9} \cline{6-9} \cline{7-9} \cline{8-9} \cline{9-9} 
 & Mean & Max & Mean & Max & Mean & Max & Mean & Max\tabularnewline
\hline 
$\beta_{1}$ & 2.10 & 2.94 & 2.30 & 3.16 & 2.18 & 3.18 & 22.31 & 33.64\tabularnewline
$\beta_{2}$ & 21.79 & 22.62 & 27.11 & 29.41 & 25.83 & 28.32 & 48.26 & 53.79\tabularnewline
$\beta_{3}$ & 41.71 & 57.64 & 25.64 & 32.89 & 33.22 & 42.64 & 46.28 & 89.75\tabularnewline
$\beta_{4}$ & 25.35 & 72.44 & 30.76 & 78.58 & 35.49 & 95.55 & 83.74 & 211.21\tabularnewline
$\mu$ & 3.51 & 36.74 & 3.93 & 32.92 & 4.19 & 30.60 & 481.95 & 3534.25\tabularnewline
$\phi$ & 33.92 & 99.00 & 487.81 & 5003.03 & 851.10 & 2517.54 & 860.24 & 4507.77\tabularnewline
$\tau^{2}$ & 42.10 & 119.37 & 91.85 & 839.92 & 329.45 & 929.22 & 818.95 & 4330.93\tabularnewline
$\rho$ & 18.67 & 42.93 & 364.29 & 2083.37 & 301.97 & 1249.85 & 585.77 & 4815.88\tabularnewline
$h_{1,1:T}$ & 2.05 & 4.91 & 1.93 & 4.55 & 2.22 & 9.84 & 52.15 & 425.00\tabularnewline
$h_{2,1:T}$ & 2.84 & 11.97 & 2.75 & 9.75 & 3.14 & 23.59 & 25.90 & 81.61\tabularnewline
$h_{10,1:T}$ & 2.95 & 12.87 & 2.94 & 17.37 & 3.66 & 57.20 & 28.40 & 70.90\tabularnewline
$h_{11,1:T}$ & 1.57 & 8.19 & 1.75 & 45.49 & 1.54 & 12.55 & 844.88 & 3756.39\tabularnewline
$h_{12,1:T}$ & 1.80 & 43.77 & 5.36 & 417.19 & 2.57 & 113.63 & 1602.90 & 5088.59\tabularnewline
$\lambda_{1,1:T}$ & 2.27 & 6.79 & 1.89 & 3.63 & 1.86 & 3.62 & 29.23 & 62.14\tabularnewline
\hline 
\end{tabular}
\end{table}

\begin{table}[H]
\caption{Inefficiency factor (IACT) of the parameters of the factor SV model
with leverage for US stock return data with $T=3001$ observations, $S=26$ stock returns, and
$K=4$ factors. Sampler I: $\textrm{\corrPMMHPG}\left({\rho}_{\epsilon},{\tau}_{\epsilon}^{2},{\tau}_{f}^{2};{f}_{1:T},{\beta},{\mu}_{\epsilon},\phi_{\epsilon},\phi_{f}\right)$
with $N=100$ particles, Sampler II: $\textrm{PHS \ensuremath{\left(\tau_{f}^{2},\tau_{\epsilon}^{2},\rho_{\epsilon};{f}_{1:T}, \mu_{\epsilon},\phi_{\epsilon},\phi_{f},\beta\right)}}$,
Sampler III: $\textrm{PGBS}\left({f}_{1:T},\mu_{\epsilon},\phi_{\epsilon},\phi_{f},\beta,\tau_{f}^{2},\tau_{\epsilon}^{2}\right)$,
and Sampler IV: $\textrm{PGDA}\left({f}_{1:T},\mu_{\epsilon},\phi_{\epsilon},\phi_{f},\beta,\tau_{f}^{2},\tau_{\epsilon}^{2}\right)$
with $N=1000$ particles. The table shows the mean and maximum IACT values for the
parameters and the factor and idiosyncratic log-volatilities. The IACT is defined in section \ref{SS: preliminaries}. \label{tab:Inefficiency-factor-(IACT) N1000 T3001}}

\centering{}%
\begin{tabular}{ccccccccc}
\hline 
 & \multicolumn{2}{c}{I} & \multicolumn{2}{c}{II } & \multicolumn{2}{c}{III } & \multicolumn{2}{c}{IV}\tabularnewline
\cline{2-9} \cline{3-9} \cline{4-9} \cline{5-9} \cline{6-9} \cline{7-9} \cline{8-9} \cline{9-9}  
 & Mean & Max & Mean & Max & Mean & Max & Mean & Max\tabularnewline
\hline 
$\beta_{1}$ & 2.10 & 2.94 & 1.92 & 2.25 & 2.20 & 2.73 & 6.39 & 9.24\tabularnewline
$\beta_{2}$ & 21.79 & 22.62 & 24.96 & 26.79 & 21.28 & 22.55 & 33.73 & 36.62\tabularnewline
$\beta_{3}$ & 41.71 & 57.64 & 27.67 & 42.14 & 37.96 & 55.49 & 57.38 & 97.75\tabularnewline
$\beta_{4}$ & 25.35 & 72.44 & 36.24 & 98.53 & 36.91 & 91.91 & 50.54 & 142.85\tabularnewline
$\mu$ & 3.51 & 36.74 & 3.61 & 32.47 & 3.90 & 33.35 & 193.95 & 2269.26\tabularnewline
$\phi$ & 33.92 & 99.00 & 180.66 & 2333.81 & 839.33 & 2233.80 & 319.60 & 1789.54\tabularnewline
$\tau^{2}$ & 42.10 & 119.37 & 56.82 & 301.98 & 341.68 & 907.59 & 231.05 & 876.51\tabularnewline
$\rho$ & 18.67 & 42.93 & 331.17 & 2467.29 & 281.52 & 900.63 & 167.79 & 1535.67\tabularnewline
$h_{1,1:T}$ & 2.05 & 4.91 & 1.88 & 4.90 & 2.19 & 8.72 & 14.65 & 60.42\tabularnewline
$h_{2,1:T}$ & 2.84 & 11.97 & 2.69 & 12.73 & 2.93 & 25.47 & 11.33 & 43.50\tabularnewline
$h_{10,1:T}$ & 2.95 & 12.87 & 2.15 & 8.41 & 3.01 & 27.36 & 13.14 & 31.69\tabularnewline
$h_{11,1:T}$ & 1.57 & 8.19 & 1.18 & 2.72 & 1.53 & 10.69 & 187.96 & 776.01\tabularnewline
$h_{12,1:T}$ & 1.80 & 43.77 & 5.27 & 404.64 & 3.14 & 148.20 & 391.46 & 991.11\tabularnewline
$\lambda_{1,1:T}$ & 2.27 & 6.79 & 1.80 & 3.47 & 1.88 & 4.08 & 11.86 & 24.52\tabularnewline
\hline 
\end{tabular}
\end{table}

\begin{table}[H]
\caption{Inefficiency factor (IACT) of the parameters of the factor SV model
with leverage for US stock return data with $T=3001$ observations, $S=26$ stock returns, and
$K=4$ factors. Sampler I: $\textrm{\corrPMMHPG}\left({\rho}_{\epsilon},{\tau}_{\epsilon}^{2},{\tau}_{f}^{2};{f}_{1:T},{\beta},{\mu}_{\epsilon},\phi_{\epsilon},\phi_{f}\right)$
with $N=100$ particles, Sampler II: $\textrm{PHS \ensuremath{\left(\tau_{f}^{2},\tau_{\epsilon}^{2},\rho_{\epsilon};{f}_{1:T}, \mu_{\epsilon},\phi_{\epsilon},\phi_{f},\beta\right)}}$,
Sampler III: $\textrm{PGBS}\left({f}_{1:T},\mu_{\epsilon},\phi_{\epsilon},\phi_{f},\beta,\tau_{f}^{2},\tau_{\epsilon}^{2}\right)$,
and Sampler IV: $\textrm{PGDA}\left({f}_{1:T},\mu_{\epsilon},\phi_{\epsilon},\phi_{f},\beta,\tau_{f}^{2},\tau_{\epsilon}^{2}\right)$
with $N=2000$ particles. The table shows the mean and maximum IACT values for the
parameters and the factor and idiosyncratic log-volatilities. The IACT is defined in section \ref{SS: preliminaries}. \label{tab:Inefficiency-factor-(IACT) N2000 T3001}}

\centering{}%
\begin{tabular}{ccccccccc}
\hline 
 & \multicolumn{2}{c}{I} & \multicolumn{2}{c}{II} & \multicolumn{2}{c}{III} & \multicolumn{2}{c}{IV}\tabularnewline
\cline{2-9} \cline{3-9} \cline{4-9} \cline{5-9} \cline{6-9} \cline{7-9} \cline{8-9} \cline{9-9}  
 & Mean & Max & Mean & Max & Mean & Max & Mean & Max\tabularnewline
\hline 
$\beta_{1}$ & 2.10 & 2.94 & 2.22 & 2.83 & 2.03 & 2.95 & 3.96 & 5.75\tabularnewline
$\beta_{2}$ & 21.79 & 22.62 & 25.99 & 28.77 & 28.71 & 30.87 & 26.85 & 32.07\tabularnewline
$\beta_{3}$ & 41.71 & 57.64 & 27.97 & 41.74 & 29.53 & 38.56 & 43.13 & 53.90\tabularnewline
$\beta_{4}$ & 25.35 & 72.44 & 37.43 & 103.88 & 31.17 & 75.47 & 39.38 & 101.97\tabularnewline
$\mu$ & 3.51 & 36.74 & 2.99 & 24.55 & 4.04 & 33.44 & 73.71 & 650.02\tabularnewline
$\tau^{2}$ & 42.10 & 119.37 & 86.37 & 676.53 & 977.36 & 1914.02 & 173.70 & 1379.56\tabularnewline
$\phi$ & 33.92 & 99.00 & 44.30 & 267.85 & 423.88 & 1170.47 & 122.16 & 690.24\tabularnewline
$\rho$ & 18.67 & 42.93 & 92.40 & 908.42 & 293.47 & 1167.13 & 66.55 & 342.51\tabularnewline
$h_{1,1:T}$ & 2.05 & 4.91 & 1.83 & 4.08 & 2.13 & 8.58 & 7.29 & 30.69\tabularnewline
$h_{2,1:T}$ & 2.84 & 11.97 & 2.38 & 9.65 & 2.03 & 10.24 & 5.93 & 13.46\tabularnewline
$h_{10,1:T}$ & 2.95 & 12.87 & 2.22 & 13.18 & 3.51 & 40.56 & 9.27 & 14.81\tabularnewline
$h_{11,1:T}$ & 1.57 & 8.19 & 1.41 & 8.49 & 1.44 & 8.74 & 82.24 & 273.58\tabularnewline
$h_{12,1:T}$ & 1.80 & 43.77 & 5.35 & 410.38 & 2.90 & 64.23 & 158.85 & 308.91\tabularnewline
$\lambda_{1,1:T}$ & 2.27 & 6.79 & 1.81 & 3.61 & 1.88 & 5.30 & 6.59 & 13.54\tabularnewline
\hline 
\end{tabular}
\end{table}

\begin{table}[H]
\caption{Inefficiency factor (IACT) of the parameters of the factor SV model
with GARCH diffusion processes for the idiosyncratic volatility for
US stock return data with $T=3001$, $S=26$, and $K=4$. Sampler
I:~$\textrm{CPHS \ensuremath{\left(\tau_{\epsilon}^{2},\tau_{f}^{2},\mu_{\epsilon},\alpha_{\epsilon};\phi_{f},f_{1:T},\beta\right)}}$
with $N=100$, Sampler II: $\textrm{\textrm{PG}\ensuremath{\left(\tau_{\epsilon}^{2},\tau_{f}^{2},\mu_{\epsilon},\alpha_{\epsilon},\phi_{f},f_{1:T},\beta\right)}}$
with $N=500$, 
Sampler III:~$\textrm{PHS}\left(\tau_{\epsilon}^{2},\tau_{f}^{2},\mu_{\epsilon},\alpha_{\epsilon};\phi_{f},f_{1:T},\beta\right)$
with $N=500$. The table shows the mean and maximum IACT values for the
parameters. The IACT is defined in section \ref{SS: preliminaries}. \label{tab:Inefficiency-factor-(IACT) N500 T3000-diffusion}}

\centering{}%
\begin{tabular}{ccccccc}
\hline 
 & \multicolumn{2}{c}{I} & \multicolumn{2}{c}{II} & \multicolumn{2}{c}{III}\tabularnewline
\cline{2-7} \cline{3-7} \cline{4-7} \cline{5-7} \cline{6-7} \cline{7-7} 
 & Mean & Max & Mean & Max & Mean & Max\tabularnewline
\hline 
$\beta_{1}$ & 2.42 & 3.53 & 2.36 & 3.19 & 2.37 & 3.33\tabularnewline
$\beta_{2}$ & 19.49 & 22.97 & 17.04 & 22.03 & 19.51 & 24.22\tabularnewline
$\beta_{3}$ & 33.95 & 42.48 & 35.62 & 47.86 & 36.56 & 45.38\tabularnewline
$\beta_{4}$ & 31.86 & 77.27 & 35.39 & 78.63 & 26.67 & 44.04\tabularnewline
$\alpha$ & 18.67 & 57.59 & 183.96 & 1365.97 & 162.16 & 445.83\tabularnewline
$\mu$ & 18.00 & 75.67 & 109.82 & 385.13 & 115.98 & 225.77\tabularnewline
$\tau^{2}$ & 18.38 & 52.19 & 2344.98 & 6610.81 & 2853.09 & 9535.97\tabularnewline
$\phi$ & 8.14 & 13.39 & 46.68 & 111.56 & 26.12 & 58.64\tabularnewline
\hline 
\end{tabular}
\end{table}

\begin{table}[H]
\caption{Inefficiency factor (IACT) of the parameters of the factor SV model
with GARCH diffusion processes for the idiosyncratic volatility for
US stock return data with $T=3001$, $S=26$, and $K=4$. Sampler
I: $\textrm{CPHS \ensuremath{\left(\tau_{\epsilon}^{2},\tau_{f}^{2},\mu_{\epsilon},\alpha_{\epsilon};\phi_{f},f_{1:T},\beta\right)}}$
with $N=100$, Sampler II: $\textrm{\textrm{PG}\ensuremath{\left(\tau_{\epsilon}^{2},\tau_{f}^{2},\mu_{\epsilon},\alpha_{\epsilon},\phi_{f},f_{1:T},\beta\right)}}$
with $N=1000$, Sampler III: $\textrm{PHS}\left(\tau_{\epsilon}^{2},\tau_{f}^{2},\mu_{\epsilon},\alpha_{\epsilon};\phi_{f},f_{1:T},\beta\right)$
with $N=1000$. The table shows the mean and maximum IACT values for the
parameters. The IACT is defined in section \ref{SS: preliminaries}. \label{tab:Inefficiency-factor-(IACT) N1000 T3000-diffusion-1}}

\centering{}%
\begin{tabular}{ccccccc}
\hline 
 & \multicolumn{2}{c}{I} & \multicolumn{2}{c}{II} & \multicolumn{2}{c}{III}\tabularnewline
\cline{2-7} \cline{3-7} \cline{4-7} \cline{5-7} \cline{6-7} \cline{7-7}
 & Mean & Max & Mean & Max & Mean & Max\tabularnewline
\hline 
$\beta_{1}$  & 2.42 & 3.53 & 2.37 & 3.33 & 2.22 & 3.37\tabularnewline
$\beta_{2}$ & 19.49 & 22.97 & 19.51 & 24.22 & 21.36 & 23.52\tabularnewline
$\beta_{3}$ & 33.95 & 42.48 & 36.56 & 45.38 & 39.62 & 54.27\tabularnewline
$\beta_{4}$ & 31.86 & 77.27 & 26.67 & 44.04 & 31.03 & 69.16\tabularnewline
$\alpha$ & 18.67 & 57.59 & 162.16 & 445.83 & 35.57 & 175.38\tabularnewline
$\mu$ & 18.00 & 75.67 & 115.98 & 225.77 & 35.27 & 168.14\tabularnewline
$\tau^{2}$ & 18.38 & 52.19 & 2853.09 & 9535.97 & 35.68 & 107.18\tabularnewline
$\phi$ & 8.14 & 13.39 & 26.12 & 58.64 & 13.28 & 22.99\tabularnewline
\hline 
\end{tabular}
\end{table}

\section{The list of industry portfolios\label{S:industryportfolio}}
\begin{table}[H]
\caption{The list of industry portfolios\label{tab:The-list-of industry portfolios}}

\centering{}%
\begin{tabular}{cc}
\hline
 & Stocks\tabularnewline
\hline
1 &  Coal \tabularnewline
2 &  Health Care and Equipment\tabularnewline
3 &  Retail\tabularnewline
4 & Tobacco\tabularnewline
5 & Steel Works\tabularnewline
6 & Food Products\tabularnewline
7 & Recreation\tabularnewline
8 & Printing and Publishing\tabularnewline
9 & Consumer Goods\tabularnewline
10 & Apparel\tabularnewline
11 & Chemicals\tabularnewline
12 & Textiles\tabularnewline
13 & Fabricated Products\tabularnewline
14 & Electrical Equipment\tabularnewline
15 & Automobiles and Trucks\tabularnewline
16 & Aircraft, ships, and Railroad Equipment\tabularnewline
17 & Industrial Mining\tabularnewline
18 & Petroleum and Natural Gas\tabularnewline
19 & Utilities\tabularnewline
20 & Telecommunication\tabularnewline
21 & Personal and Business Services\tabularnewline
22 & Business Equipment\tabularnewline
23 & Transportation\tabularnewline
24 & Wholesale\tabularnewline
25 & Restaurants, Hotels, and Motels\tabularnewline
26 & Banking, Insurance, Real Estate\tabularnewline
\hline
\end{tabular}
\end{table}

\section{Ergodicity of the \corrPMMHPG \label{SSS: ergodidity}}
Proposition~\ref{ergodicity of SS1} shows that the \corrPMMHPG{} converges to $\wt \pi^N$ in total variation norm
if assumption~\ref{ass: ass for ergodicity} holds;
the conditions of assumption~\ref{ass: ass for ergodicity} hold for most applications, and in particular for the applications in our article.
This section discusses the ergodicity of the \corrPMMHPG{} under conditions
that hold for our applications. Define,
\begin{align*}
w_1^\theta(x_1) &:= \frac{g_1^\theta(y_1|x_1)f_1^\theta(x_1)}{m_1^\theta(x_1)} \quad \text{and} \quad
w_t^\theta(x_t,x_{t-1}) := \frac{g_t^\theta(y_t|x_t)f_t^\theta(x_t|x_{t-1})}{m_t^\theta(x_t|x_{t-1})} \quad  \text{for $t \geq 2$},
\end{align*}
and assume that,
\begin{assumption} \label{ass: ass for ergodicity}
\begin{enumerate}
\item [(i)]
$0< w_1^\theta (x_1) < \infty $ and  $0 < w_t^\theta(x_t,x_{t-1}) < \infty$ for $ t \geq 2$,
for all $\theta \in \Theta$ and $x_t,x_{t-1} \in \mathcal{X}$.
\item [(ii)] For $i=1, \dots, p_1$,
$0< q_i(\theta_i|v_{x,1:T}^{1:N}, v_{A,1:T-1}, \theta_{-i},\theta_i^*) < \infty $
for all $(\theta_i, \theta_{-i}), (\theta_i^\ast, \theta_{-i}) \in \Theta$, and for all $v_{x,1:T}^{1:N}, v_{A,1:T-1}$.
\item [(iii)]
For $i=p_1+1, \dots, p$,
$0< q_i(\theta_i|x_{1:T},j_{1:T} , \theta_{-i}, \theta_i^\ast) < \infty $
for all $(\theta_i, \theta_{-i}), (\theta_i^\ast, \theta_i) \in \Theta$, and for all $x_{1:T} \in \mathcal{X}^{1:T} $
and $j_{1:T} \in \{1,\dots, N\}^T$.
\end{enumerate}
\end{assumption}

\begin{proposition}\label{ergodicity of SS1}
Suppose that assumptions~\ref{assu:propstatespace} to
\ref{ass: ass for ergodicity} hold.
Then \corrPMMHPG{} ( Algorithm \ref{alg:Sampling-Scheme:-The correlated PMMH+PG})
converges to the target distribution $\widetilde\pi^{N}$ (\eqref{eq:Target distribution}) in total variation norm.

\begin{proof}
It is convenient to use the following notation
\begin{eqnarray*}
\widetilde{\pi}^{N}\left(dv_{x,1:T}^{1:N},dv_{A,1:T-1}^{1:N},j_{1:T},d\theta\right)
=
\widetilde{\pi}^{N}\left(dv_{x,1:T}^{j_{1:T}},dv_{x,1:T}^{-j_{1:T}},dv_{A,1:T-1}^{1:N},j_{1:T},d\theta\right)
\end{eqnarray*}
to partition the variables $v_{x,1:T}^{1:N}$ into $v_{x,1:T}^{j_{1:T}}$ and $v_{x,1:T}^{-j_{1:T}}$,
which are the variables selected and not selected by the indices $j_{1:T}$ respectively.

Without loss of generality,
let $D \in {\cal B}\left( \mathfrak{V}_{x}^{T} \right)$,
$E \in {\cal B}\left( \mathfrak{V}_{x}^{(N-1)T} \right)$,
$F \in {\cal B} \left( \mathfrak{V}_{A}^{T}  \right)$,
$j_{1:T}^{'} \in \{1, \ldots, N\}^T$
and $G \in {\cal B} \left( \Theta \right)$ be such that
$\widetilde{\pi}_{N}\left(  D \times E \times F \times \{j_{1:T}^{'}\} \times G \right) > 0$.
Denote the law of the process defined by 
Algorithm~\ref{alg:Sampling-Scheme:-The correlated PMMH+PG} 
by ${\cal L}_{E}(\cdot)$ and note that 
assumptions~\ref{assu:propstatespace} and \ref{ass: ass for ergodicity} parts (i), (ii) and (iii) imply that
\begin{eqnarray*}
{\cal L}_{E}\left( \left( V_{x,1:T}^{-j_{1:T}}(k), V_{A,1:T-1}^{1:N},j_{1:T}(k), J_{1:T}(k), \theta(k)  \right) \in E \times F \times \{j_{1:T}\} \times G \right) > 0
\end{eqnarray*}
for all $k > 0$ and all $j_{1:T} \in \{1, \ldots, N\}^T$.
Applying this result repeatedly shows that
\begin{eqnarray*}
{\cal L}_{E}\left( \left( V_{x,1:T}^{j_{1:T}}(k), V_{x,1:T}^{-j_{1:T}}(k), V_{A,1:T-1}^{1:N},j_{1:T}(k), J_{1:T}(k), \theta(k)  \right) \in D \times E \times F \times \{j_{1:T}\} \times G \right) > 0
\end{eqnarray*}
for all $k > 1$ and all $j_{1:T} \in \{1, \ldots, N\}^T$,
which proves the irreducibility and aperiodicity of the Markov chain defined by Algorithm~\ref{alg:Sampling-Scheme:-The correlated PMMH+PG}.
\end{proof}
\end{proposition}

\subsection{Ergodicity of the \corrPMMHPG{}  for the factor SV model\label{SSS: ergodidity}}
Sampling scheme~\ref{alg:Sampling-Scheme:-The correlated PMMH+PG for factor SV 
model } in section~\ref{subsec:Correlated-PMMH+PG-sampling}
has the stationary distribution~\eqref{eq:Target distribution factor model} by 
construction.
The transition kernel of sampling scheme~\ref{alg:Sampling-Scheme:-The correlated PMMH+PG for factor SV model } is a composite of the transition kernels
discussed in the proof of proposition~\ref{ergodicity of SS1} together with the transition kernels for $ \beta$ and $ f_{1:T}$ which are positive;
conditioning on the values of $ \beta$ and $ f_{1:T}$ does not change the accessible sets of the remaining variables.
Therefore, 
sampling scheme~\ref{alg:Sampling-Scheme:-The correlated PMMH+PG for factor SV model } is ergodic
using a similar proof to proposition~\ref{ergodicity of SS1}.

\end{document}